\documentclass[pra,onecolumn,superscriptaddress,nofootinbib]{revtex4}
\usepackage[a4paper, left=0.8in, right=0.8in, top=1.0in, bottom=1.0in]{geometry}

\usepackage[utf8]{inputenc}
\usepackage[english]{babel}
\usepackage[T1]{fontenc}
\usepackage{amsmath}
\usepackage{amsfonts}
\usepackage{hyperref}
\allowdisplaybreaks

\usepackage{tikz}
\usepackage{braket}
\usepackage{natbib}
\usepackage{amsthm}

\newtheorem{theorem}{Theorem}
\newtheorem{definition}{Definition}
\newtheorem{lemma}{Lemma}

\newtheorem{remark}{Remark}

\newtheorem{proposition}{Proposition}
\DeclareMathOperator{\Tr}{tr}
\DeclareMathOperator{\E}{\mathbb{E}}

\DeclareMathOperator{\Var}{\mathrm{Var}}

\newcommand{\norm}[1]{\left\lVert#1\right\rVert}
\newcommand{\vct}[1]{\mathbf{#1}}

\usepackage[noend]{algpseudocode}
\usepackage{algorithm,algorithmicx}

\algrenewcommand\alglinenumber[1]{\sf\scriptsize\color{blue}{#1}}
\algrenewcommand\algorithmicrequire{\textbf{Input:}}
\algrenewcommand\algorithmicensure{\textbf{Output:}}

\begin{document}

\title{Predicting Many Properties of a Quantum System from Very Few Measurements}

\date{\today}
\author{Hsin-Yuan Huang}
\email{hsinyuan@caltech.edu}
\affiliation{Institute for Quantum Information and Matter, Caltech, Pasadena, CA, USA}
\affiliation{Department of Computing and Mathematical Sciences, Caltech, Pasadena, CA, USA}
\author{Richard Kueng}
\affiliation{Institute for Quantum Information and Matter, Caltech, Pasadena, CA, USA}
\affiliation{Department of Computing and Mathematical Sciences, Caltech, Pasadena, CA, USA}
\affiliation{Institute for Integrated Circuits, Johannes Kepler University Linz, Austria}
\author{John Preskill}
\affiliation{Institute for Quantum Information and Matter, Caltech, Pasadena, CA, USA}
\affiliation{Department of Computing and Mathematical Sciences, Caltech, Pasadena, CA, USA}
\affiliation{Walter Burke Institute for Theoretical Physics, Caltech, Pasadena, CA, USA}

\begin{abstract}
Predicting properties of complex, large-scale quantum systems is essential for developing quantum technologies.
We present an efficient method for constructing an approximate classical description of a quantum state using very few measurements of the state. This description, called a \emph{classical shadow}, can be used to predict many different properties: order $\log M$ measurements suffice to accurately predict $M$ different functions of the state with high success probability. The number of measurements is independent of the system size, and saturates information-theoretic lower bounds. Moreover, target properties to predict can be selected after the measurements are completed.
We support our theoretical findings with extensive numerical experiments. We apply classical shadows to predict quantum fidelities, entanglement entropies, two-point correlation functions, expectation values of local observables, and the energy variance of many-body local Hamiltonians.
The numerical results highlight the advantages of classical shadows relative to previously known methods.
\end{abstract}

\maketitle

\setcounter{secnumdepth}{0}


Making predictions based on empirical observations is a central topic in statistical learning theory and is at the heart of many scientific disciplines, including quantum physics.
There, predictive tasks, like estimating target fidelities, verifying entanglement, and measuring correlations, are essential for building, calibrating and controlling quantum systems.
Recent advances in the size of quantum platforms \cite{preskill2018nisq} have pushed traditional prediction techniques --- like quantum state tomography --- to the limit of their capabilities. This is mainly due to a curse of dimensionality:  the number of parameters needed to describe a quantum system scales exponentially with the number of its constituents. Moreover, these parameters cannot be accessed directly, but must be estimated by measuring the system. An informative quantum mechanical measurement is both destructive (wave-function collapse) and only yields probabilistic outcomes (Born's rule). Hence, many identically prepared samples are required to estimate accurately even a single parameter of the underlying quantum state. Furthermore, all of these measurement outcomes must be processed and stored in memory for subsequent prediction of relevant features.
In summary, reconstructing a full description of a quantum system with $n$ constituents (e.g. qubits) necessitates a number of measurement repetitions exponential in $n$, as well as an exponential amount of classical memory and computing power.

Several approaches have been proposed to overcome this fundamental scaling problem. These include matrix product state (MPS) tomography \cite{cramer2010efficient} and neural network tomography \cite{torlai2018neural, carrasquilla2019reconstructing}. Both only require a polynomial number of samples, provided that the underlying state has suitable properties. However, for general quantum systems, these techniques still require an exponential number of samples.
We refer to the related work section (Supplementary Section \ref{sec:related-work}) for details.

Pioneering a conceptually very different line of research, Aaronson \cite{aaronson2018shadow}
pointed out that demanding full classical descriptions of quantum systems may be excessive for many concrete tasks. Instead it is often sufficient to accurately predict certain properties of the quantum system. In quantum mechanics, interesting properties are often \emph{linear} functions of the underlying density matrix $\rho$, such as the expectation values $\{o_i\}$ of a set of observables $\{O_i\}$:
\begin{align}
o_i (\rho) =& \mathrm{trace}(O_i \rho) \quad 1 \leq i \leq M.
\label{eq:linear_predictions}
\end{align}
The fidelity with a pure target state, entanglement witnesses, and the probability distribution governing the possible outcomes of a measurement are all examples that fit this framework. A
\emph{nonlinear} function of $\rho$ such as entanglement entropy, may also be of interest.
Aaronson coined the term
\cite{aaronson2018shadow,aaronson2019gentle} \emph{shadow tomography}\footnote{According to Ref.~\cite{aaronson2018shadow} it was actually S.T.~Flammia who originally suggested the name shadow tomography.} for the task of predicting properties without necessarily fully characterizing the quantum state, and he showed that a polynomial number of state copies already suffice to predict an exponential number of target functions. While very efficient in terms of samples, Aaronson's procedure is very demanding in terms of quantum hardware --- a concrete implementation of the proposed protocol requires exponentially long quantum circuits that act collectively on all the copies of the unknown state stored in a quantum memory.

\begin{figure*}[t]
    \centering
    \includegraphics[width=0.85\textwidth]{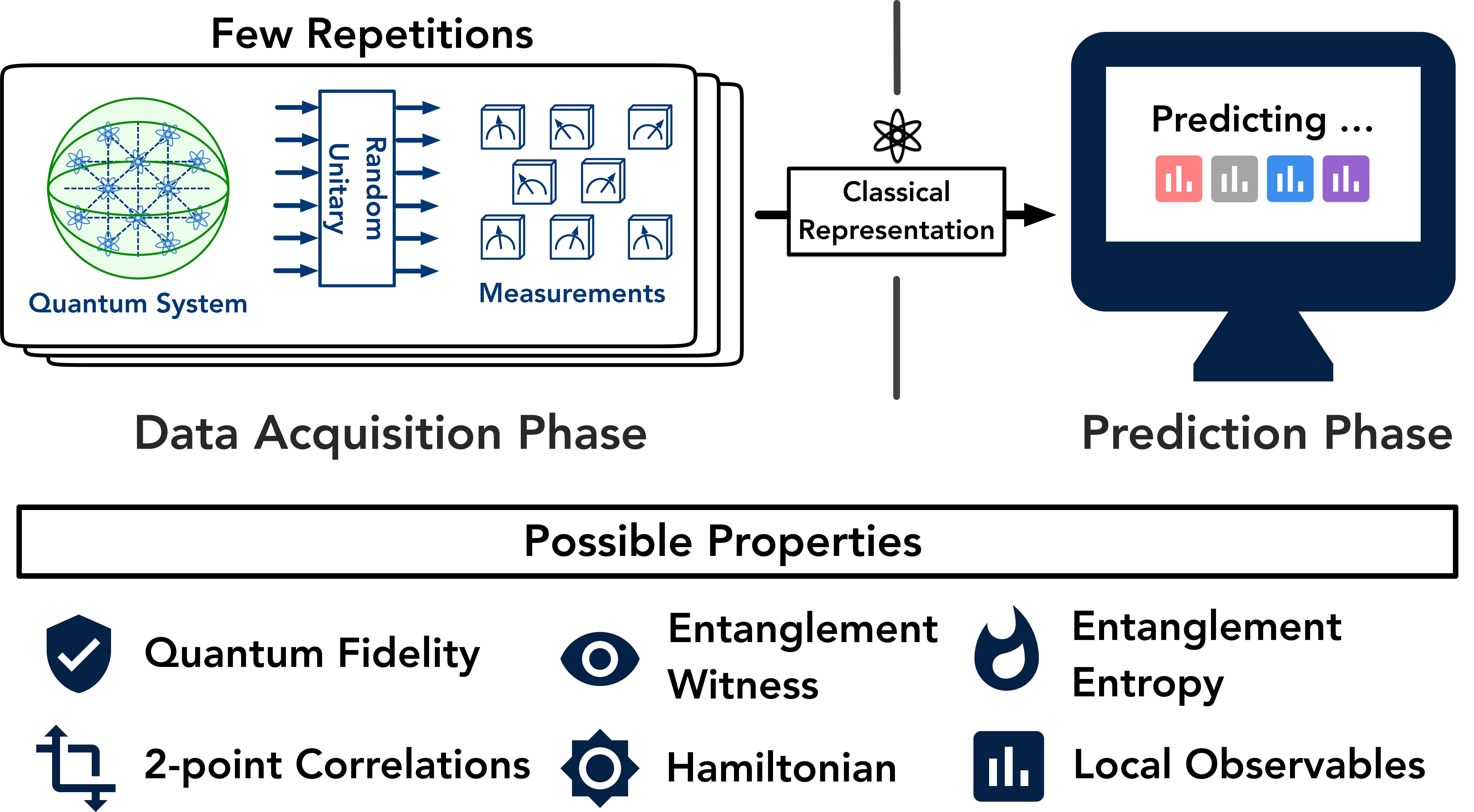}
    \caption{An illustration for constructing a classical representation, the \emph{classical shadow}, of a quantum system from randomized measurements. In the data acquisition phase, we perform a random unitary evolution and measurements on independent copies of an $n$-qubit system to obtain a classical representation of the quantum system --- the \emph{classical shadow}. Such classical shadows facilitate accurate prediction of a large number of different properties using a simple median-of-means protocol.}
    \label{fig:setup}
\end{figure*}

In this work, we combine the mindset of shadow tomography \cite{aaronson2018shadow} (predict target functions, not the full state)  with recent insights from quantum state tomography \cite{guta2018fast} (rigorous statistical convergence guarantees)  and the  stabilizer formalism \cite{gottesman1997stabilizer} (efficient implementation).
The result is a highly efficient protocol that learns a minimal classical sketch $S_\rho$ -- the \emph{classical shadow} -- of an unknown quantum state $\rho$ that can be used to predict arbitrary linear function values \eqref{eq:linear_predictions} by a simple median-of-means protocol.
A classical shadow is created by repeatedly performing a simple procedure:
Apply a unitary transformation $\rho \mapsto U \rho U^\dagger$, and then measure all the qubits in the computational basis. The number of times this procedure is repeated is called the \textit{size} of the classical shadow.
The transformation $U$ is randomly selected from an ensemble of unitaries, and different ensembles lead to different versions of the procedure that have characteristic strengths and weaknesses.
In a practical scheme, each ensemble unitary should be realizable as an efficient quantum circuit.
We consider random $n$-qubit Clifford circuits and tensor products of random single-qubit Clifford circuits as important special cases.
These two procedures turn out to complement each other nicely.
We refer to Figure~\ref{fig:setup} for a visualization and a list of important properties that can be predicted efficiently.

Our main theoretical contribution equips this procedure with rigorous performance guarantees. Classical shadows with size of order $\log(M)$ suffice to predict $ M$ target functions in Eq.~\eqref{eq:linear_predictions} simultaneously.
Most importantly, the actual system size (number of qubits) does not enter directly. Instead, the number of measurement repetitions $N$
is determined by a (squared) norm
$\|O_i \|_{\mathrm{shadow}}^2$.
This norm depends on the target functions and the particular measurement procedure used to produce the classical shadow. For example, random $n$-qubit Clifford circuits lead to the Hilbert-Schmidt norm.
On the other hand, random single-qubit Clifford circuits produce a norm that scales exponentially in the locality of target functions, but is independent of system size.
The resulting prediction technique is applicable to current laboratory experiments and facilitates the efficient prediction of few-body properties, such as two-point correlation functions, entanglement entropy of small subsystems, and expectation values of local observables.

In some cases, this scaling may seem unfavorable.
However, we rigorously prove that this is not a flaw of the method, but an unavoidable limitation rooted in quantum information theory. By relating the prediction task to a communication task~\cite{fano1961information_theory},
we establish fundamental lower bounds
highlighting that classical shadows are (asymptotically) optimal.

We support our theoretical findings by conducting numerical simulations for predicting various physically relevant properties over a wide range of system sizes. These include quantum fidelity, two-point correlation functions, entanglement entropy, and local observables.
We confirm that prediction via classical shadows scales favorably and improves on powerful existing techniques --- such as
machine learning ---
in a variety of well-motivated test cases.
An open source release for predicting many properties from very few measurements is available at \url{https://github.com/momohuang/predicting-quantum-properties}.

\section{Procedure}

Throughout this work we restrict attention to $n$-qubit systems and $\rho$ is a fixed, but unknown, quantum state in $d=2^n$ dimensions.
To extract meaningful information, we repeatedly perform a simple measurement procedure: apply a random unitary to rotate the state ($\rho \mapsto U \rho U^\dagger$) and perform a computational-basis measurement. The unitary $U$ is selected randomly from a fixed ensemble.
Upon receiving the $n$-bit measurement outcome $|\hat{b}\rangle \in \left\{0,1 \right\}^n$, we store an (efficient) classical description of $U^\dagger |\hat{b} \rangle \! \langle \hat{b}| U$ in classical memory.
It is instructive to view the average (over both the choice of unitary and the outcome distribution) mapping from $\rho$ to its classical snapshot
$U^\dagger |\hat{b}\rangle\!\langle\hat{b}| U$
as a quantum channel:
\begin{equation}
    \E \left[ U^\dagger |\hat{b}\rangle\!\langle\hat{b}| U \right]
    = \mathcal{M}(\rho) \implies  \rho = \E \left[ \mathcal{M}^{-1} \left(U^\dagger
    |\hat{b} \rangle \! \langle \hat{b}| U \right) \right]. \label{eq:snapshot}
\end{equation}
This quantum channel $\mathcal{M}$ depends on the ensemble of (random) unitary transformations.
Although the inverted channel $\mathcal{M}^{-1}$ is not physical (it is not completely positive),
 we can still apply $\mathcal{M}^{-1}$ to the (classically stored) measurement outcome $U^\dagger |\hat{b}\rangle\!\langle\hat{b}| U$
 in a completely classical post-processing step.\footnote{$\mathcal{M}$ is invertible if the ensemble of unitary transformations defines a tomographically complete set of measurements. See Supplementary Section \ref{sec:generalshadow}.}
 In doing so, we produce a single classical snapshot $\hat{\rho} = \mathcal{M}^{-1} \left( U^\dagger |\hat{b} \rangle \! \langle \hat{b}| U\right)$ of the unknown state $
 \rho$ from a single measurement.
 By construction, this snapshot exactly reproduces the underlying state in expectation (over both unitaries and measurement outcomes): $\E [\hat{\rho}]=\rho$.
Repeating this procedure $N$ times results in an array of $N$ independent, classical snapshots of $\rho$:
\begin{equation}
\mathsf{S}(\rho;N) = \left\{ \hat{\rho}_1 = \mathcal{M}^{-1}\left( U_1^\dagger \ket{\hat{b}_1}\! \bra{\hat{b}_1} U_1 \right), \ldots, \hat{\rho}_N = \mathcal{M}^{-1}\left( U_N^\dagger \ket{\hat{b}_N} \! \bra{\hat{b}_N} U_N \right) \right\}. \label{eq:classical_shadow}
\end{equation}
We call this array the \emph{classical shadow} of $\rho$.
Classical shadows of sufficient size $N$ are expressive enough to predict many properties of the unknown quantum state efficiently.
To avoid outlier corruption, we split the classical shadow up into equally-sized chunks and construct several, independent sample mean estimators.
Subsequently, we predict linear function values \eqref{eq:linear_predictions}
via \emph{median of means estimation} \cite{jerrum1986medianmeans,nemirovski1983medianmeans}. This procedure is summarized in Algorithm~\ref{alg:median_means}.
For many physically relevant properties $O_i$ and measurement channels $\mathcal{M}$, Algorithm~\ref{alg:median_means} can be carried out very efficiently without explicitly constructing the large matrix $\hat{\rho}_i$.

Median of means prediction with classical shadows can be defined for any
distribution of random unitary transformations.
Two prominent examples are:
(i) random $n$-qubit Clifford circuits; and (ii) tensor products of random single-qubit Clifford circuits.
Example (i)
results in a clean and powerful theory,
but also practical drawbacks, because $n^2 /\log (n)$ entangling gates are needed to sample from $n$-qubit Clifford unitaries.
The corresponding inverted quantum channel is $\mathcal{M}^{-1}_n(X) = (2^n+1)X - \mathbb{I}$.
Example (ii) is equivalent to measuring each qubit independently in a random Pauli basis. Such measurements can be routinely carried out in many experimental platforms.
The corresponding inverted quantum channel is $\mathcal{M}^{-1}_P = \bigotimes_{i=1}^n \mathcal{M}^{-1}_1$.
We refer to examples (i) / (ii) as random Clifford / Pauli measurements, respectively. In both cases, the resulting classical shadow can be stored efficiently in a classical memory using the stabilizer formalism.

\begin{algorithm*}[t]
{\small
\begin{algorithmic}[1]
\caption{{\small \textit{Median of means prediction}  based on a classical shadow $\mathsf{S}(\rho,N)$.}}
\label{alg:median_means}

\Statex
\Function{LinearPredictions}{$O_1,\ldots,O_M,\mathsf{S}(\rho;N),K$}

\State Import $\mathsf{S}(\rho;N) = \left[ \hat{\rho}_1, \ldots, \hat{\rho}_N \right]$
\Comment Load classical shadow

\State Split the shadow into $K$ equally-sized parts and set
\Comment Construct $K$ estimators of $\rho$
\begin{align*}
\hat{\rho}_{(k)} =& \frac{1}{\lfloor N/K \rfloor} \sum_{i = (k-1) \lfloor N/K \rfloor+1}^{k\lfloor N/K \rfloor} \hat{\rho}_i
\end{align*}

\For{$i=1$ to $M$}

\State Output
$
\hat{o}_i (N,K) = \mathrm{median} \left\{ \mathrm{tr} \left(O_i \hat{\rho}_{(1)} \right),\ldots, \mathrm{tr} \left( O_i \hat{\rho}_{(K)} \right) \right\}.
$
\Comment Median of means estimation
\EndFor
\EndFunction
\end{algorithmic}
}
\end{algorithm*}

\section{Rigorous performance guarantees}

\begin{theorem}[informal version] \label{thm:main}
Classical shadows of size $N$ suffice to predict $M$ \emph{arbitrary} linear target functions $\mathrm{tr}(O_1 \rho),\ldots,\mathrm{tr}(O_M \rho)$ up to additive error $\epsilon$ given that $N \geq \text{(order) } \log(M) \max_i \norm{O_i}^2_{\rm{shadow}} / \epsilon^2$.
The definition of the norm $\norm{O_i}_{\mathrm{shadow}}$ depends on the ensemble of unitary transformations used to create the classical shadow.
\end{theorem}

We refer to Section~\ref{sec:generalshadow} in the Supplementary Information for background, a detailed statement and proofs.
Theorem~\ref{thm:main} is most powerful when the linear functions have a bounded norm that is independent of system size.
In this case, classical shadows allow for predicting a large number of properties from only a logarithmic number of quantum measurements.

The norm $\norm{O_i}_{\mathrm{shadow}}$ in Theorem~\ref{thm:main} plays an important role in defining the space of linear functions that can be predicted efficiently.
For random Clifford measurements,
$\|O \|_{\mathrm{shadow}}^2$ is closely related to the Hilbert-Schmidt norm $\mathrm{tr}(O^2)$.
As a result, a large collection of (global) observables with a bounded Hilbert-Schmidt norm can be predicted efficiently.
For random Pauli measurements, the norm scales exponentially in the locality of the observable, not the actual number of qubits.
For an observable $O_i$ that acts non-trivially on (at most) $k$ qubits, $\norm{O_i}_{\mathrm{shadow}}^2 \leq 4^k \norm{O_i}_\infty^2$, where $\norm{\cdot}_\infty$ denotes the operator norm\footnote{This scaling can be further improved to $3^k$ if $O_i$ is a tensor product of $k$ single-qubit observables.}.
This guarantees the accurate prediction of many local observables from only a much smaller number of measurements.

\section{Illustrative example applications}

\paragraph{Quantum fidelity estimation.}

Suppose we wish to certify that an experimental device prepares a desired $n$-qubit state. Typically, this target state $| \psi \rangle \! \langle \psi|$ is pure and highly structured, e.g.\ a a GHZ state \cite{greenberger1989GHZ} for quantum communication protocols, or a
toric code ground state \cite{kitaev2002toric} for fault-tolerant quantum computation.
Theorem~\ref{thm:main} asserts that a classical shadow (Clifford measurements) of dimension-independent size suffices to accurately predict the fidelity of \emph{any} state in the lab with \emph{any} pure target state. This
improves on the best existing result on direct fidelity estimation \cite{flammia2011direct} which requires $O(2^n / \epsilon^4)$ samples in the worst case. Moreover, a classical shadow of polynomial size allows for estimating an exponential number of (pure) target fidelities all at once.

\paragraph{Entanglement verification.}

Fidelities with pure target states can also serve as (bipartite) \emph{entanglement witnesses} \cite{guehne2009entanglement}.
For every (bipartite) entangled state $\rho$, there exists a constant $\alpha$ and an observable $O = | \psi \rangle \! \langle\psi|$ such that $\Tr(O \rho) > \alpha \geq \Tr(O \rho_s)$, for all (bipartite) separable states $\rho_s$.
Establishing $\Tr(O \rho) > \alpha$ verifies the existence of entanglement in the state $\rho$.
Any $O = | \psi \rangle \! \langle\psi|$ that satisfies the above condition is known as an entanglement witness for the state~$\rho$.
Classical shadows (Clifford measurements) of logarithmic size allow for checking a large number of potential entanglement witnesses simultaneously.

\paragraph{Predicting expectation values of local observables.}

Many near-term applications of quantum devices rely on repeatedly estimating a large number of local observables. For example, low-energy eigenstates of a many-body Hamiltonian may be prepared and studied using a variational method, in which the Hamiltonian, a sum of local terms, is measured many times. Classical shadows constructed from a logarithmic number of random Pauli measurements can efficiently estimate polynomially many such local observables. Because only single-qubit Pauli measurements suffice, this measurement procedure is highly efficient. Potential applications include quantum chemistry \cite{kandala2017hardware} and lattice gauge theory \cite{kokail2019self}.

\paragraph{Predicting expectation values of global observables (non-example).}

Classical shadows are not without limitations. In our examples, the size of classical shadows must either scale with $\mathrm{tr}(O_i^2)$ (Clifford measurements) or must scale exponentially in the locality of $O_i$ (Pauli measurements). Both quantities can simultaneously become exponentially large for nonlocal observables with large Hilbert-Schmidt norm.
A concrete example is the Pauli expectation value of a spin chain: $\langle P_{i_1} \otimes \cdots \otimes P_{i_n} \rangle_\rho = \mathrm{tr} \left( O_1 \rho \right)$, where $\mathrm{tr}(O_1^2)  =2^n$ and $k=n$ (non-local observable).
In this case, classical shadows of exponential size
may be required to accurately predict a single expectation value. In contrast, a direct spin measurement achieves the same accuracy with only of order $1/\epsilon^2$ copies of the state $\rho$.

\section{Matching information-theoretic lower bounds}

The non-example above raises an important question: does the scaling of the required number of measurements with Hilbert-Schmidt norm or with the locality of observables arise from a fundamental limitation, or is it  merely an artifact of prediction with classical shadows?
A rigorous analysis reveals that this scaling is no mere artifact; rather it stems from information-theoretic reasons.

\begin{theorem}[informal version] \label{thm:main2}
Any procedure based on single-copy measurements, that can predict \emph{any} $M$ linear functions $\mathrm{tr}(O_i\rho) $ up to additive error $\epsilon$, requires at least (order) $\log(M) \max_i \|O_i\|_{\mathrm{shadow}}^2 /\epsilon^2$ measurements.
\end{theorem}

\noindent Here, $\norm{O_i}_{\mathrm{shadow}}^2$ could be taken as the Hilbert-Schmidt norm $\Tr(O_i^2)$ or as a function scaling exponentially in the locality of $O_i$.
The proof results from embedding the abstract prediction procedure into a communication protocol. Quantum information theory imposes fundamental restrictions on any quantum communication protocol and allows us to deduce stringent lower bounds.
We refer to Supplementary Section~\ref{sec:proofthm2}~and~\ref{sec:proofthm2P} for details and proofs.

The two main technical results complement each other nicely.
Theorem~\ref{thm:main} equips classical shadows with a constructive performance guarantee: an order of $\log(M) \max_i \| O_i \|_{\mathrm{shadow}}^2/\epsilon^2$ single-copy measurements suffice to accurately predict an \emph{arbitrary} collection of $M$ target functions. Theorem~\ref{thm:main2} highlights that this number of measurements is unavoidable in general.

\section{Predicting nonlinear functions}

The classical shadow $S(\rho; N) = \left\{ \hat\rho_1, \ldots, \hat\rho_N \right\}$ of the unknown quantum state $\rho$ may also be used to predict non-linear functions $f(\rho)$. We illustrate this with a quadratic function $f(\rho) = \Tr(O \rho \otimes \rho)$, where $O$ acts on two copies of the state.
Because $\hat\rho_i$ is equal to the quantum state $\rho$ in expectation, one could predict $\Tr(O \rho \otimes \rho)$ using two independent snapshots $\hat\rho_i, \hat\rho_j, i \neq j$. Because of independence, $\Tr(O \hat\rho_i \otimes \hat\rho_j)$ correctly predicts the quadratic function in expectation:
\begin{equation}
\E \Tr(O \hat\rho_i \otimes \hat\rho_j) = \Tr(O \E \hat\rho_i \otimes \E \hat\rho_j) = \Tr(O \rho \otimes \rho).
\end{equation}
To reduce the prediction error, we use $N$ independent snapshots and symmetrize over all possible pairs: $\frac{1}{N(N-1)} \sum_{i \neq j} \Tr(O \hat\rho_i \otimes \hat\rho_j)$.
We then repeat this procedure several times and form their median to further reduce the likelihood of outlier corruption (similar to median of means).
Rigorous performance guarantees are given in Supplementary Section \ref{sec:nonlinearrig}.
This approach readily generalizes to higher order polynomials using U-statistics \cite{hoeffding1992class}.

One particularly interesting nonlinear function is the second-order R\'enyi entanglement entropy: $-\log(\Tr(\rho_A^2))$, where $A$ is a subsystem of the $n$-qubit quantum system.
We can rewrite the argument in the $\log$ as $\mathrm{tr}(\rho_A^2)=\mathrm{tr} \left( S_A \rho \otimes \rho\right)$ --- where $S_A$ is the local swap operator of two copies of the subsystem $A$ --- and use classical shadows to obtain very accurate predictions.
The required number of measurements scales exponentially in the size of the subsystem $A$, but is independent of total system size.
Probing this entanglement entropy is a useful task and a highly efficient specialized approach has been proposed in \cite{brydges2019probing}.
We compare this \emph{Brydges \textit{et al.} method} to classical shadows in the numerical experiments.

For nonlinear functions, unlike linear ones, we have have not derived an information-theoretic lower bound on the number of measurements needed, though it may be possible to do so by generalizing our methods.

\section{Numerical experiments}

One of the key features of prediction with classical shadows is scalability.
The data acquisition phase is designed to be tractable for state of the art platforms (Pauli measurements) and future quantum computers (Clifford measurements), respectively.
The resulting classical shadow can be stored efficiently in classical memory.
For may important features -- such as local observables or global features with efficient stabilizer decompositions -- scalability moreover extends to the computational cost associated with median of means prediction.

These design features allowed us to conduct numerical experiments for a wide range of problems and system sizes (up to 160 qubits). The computational bottleneck is not feature prediction with classical shadows, but generating synthetic data, i.e.\ classically generating target states and simulating quantum measurements.
Needless to say, this classical bottle-neck does not occur in actual experiments.
We then use this synthetic data to learn a classical representation of $\rho$ and use this representation to predict various interesting properties.

Machine learning based approaches \cite{carrasquilla2019reconstructing, torlai2018neural} are among the most promising alternative methods that have applications in this regime, where the
Hilbert space dimension is roughly comparable to the total number of silicon atoms on earth ($2^{160} \simeq 10^{48}$).
For example, a recent version of \emph{neural network quantum state tomography} (NNQST) is a generative model that is based on a deep neural network trained on independent quantum measurement outcomes (local SIC/tetrahedral POVMs \cite{renes2004SIC}).
In this section, we consider the task of learning a classical representation of an unknown quantum state, and using the representation to predict various properties, addressing the relative merit of classical shadows and alternative methods.

\subsection{Predicting quantum fidelities (Clifford measurements)}

\begin{figure}[t]
    \centering
    \includegraphics[width=1.0\textwidth]{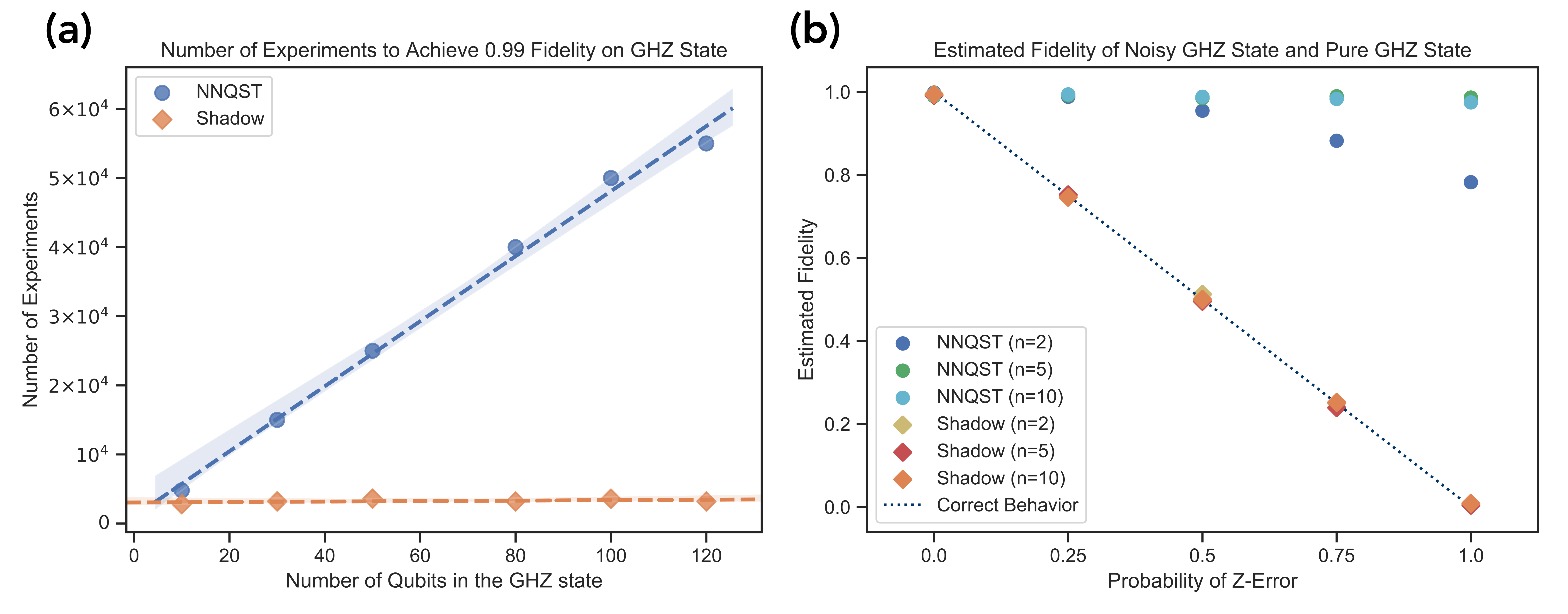}
    \caption{\emph{Predicting quantum fidelities using classical shadows (Clifford measurements) and NNQST.}\\
\textbf{(a)} \emph{(Left)}\textbf{:} Number of measurements required to identify an $n$-qubit GHZ state with 0.99 fidelity. The shaded regions are the standard deviation of the needed number of experiments over ten independent runs. \\
\textbf{(b)} \emph{(Right)}\textbf{:} Estimated fidelity between a perfect GHZ target state and a noisy preparation, where $Z$-errors can occur with probability $p \in [0,1]$, under $6\times 10^4$ experiments. The dotted line represents the true fidelity as a function of $p$. \\
NNQST can only estimate an upper bound on quantum fidelity efficiently, so we consider this upper bound for NNQST and use quantum fidelity for the classical shadow.
}
    \label{fig:tomoGHZ}
\end{figure}

Here we focus on classical shadows based on random Clifford measurements which are designed to predict observables with bounded Hilbert-Schmidt norm.
When the observables have efficient representations --- such as efficient stabilizer decompositions --- the computational cost for performing median of means prediction can also be efficient.\footnote{The runtime of Algorithm~\ref{alg:median_means} is dominated by the cost of computing quadratic functions $\langle \hat{b}|U O U^\dagger |\hat{b}\rangle$ in $2^n$ dimensions. If $O=|\psi \rangle \! \langle \psi|$ is a stabilizer state, the Gottesman-Knill theorem allows for evaluation in $\mathcal{O}(n^2)$-time.}
An important example is the quantum fidelity with a target state.
In \cite{carrasquilla2019reconstructing}, the viability of NNQST is demonstrated by considering GHZ states with a varying number of qubits $n$. Numerical experiments highlight that the number of measurement repetitions (size of the training data) to learn a neural network model of the GHZ state that achieves target fidelity of $0.99$ scales linearly in $n$.
We have also implemented NNQST for GHZ states and compared it to median of means prediction with classical shadows. The left-hand side of Figure~\ref{fig:tomoGHZ} confirms the linear scaling of NNQST and the assertion of Theorem~\ref{thm:main}: classical shadows of \emph{constant} size suffice to accurately estimate GHZ target fidelities, regardless of the actual system size. In addition, we have also tested the ability of both approaches to detect potential state preparation errors.
More precisely, we consider a scenario where the GHZ-source introduces a phase error with probability $p \in [0,1]$:
\begin{equation}
\rho_p = (1-p) |\psi_\mathrm{GHZ}^+(n)\rangle \! \langle \psi_\mathrm{GHZ}^+ (n)|+ p|\psi_\mathrm{GHZ}^-(n)\rangle \! \langle \psi_\mathrm{GHZ}^- (n)|, \quad | \psi_\mathrm{GHZ}^\pm (n) \rangle = \tfrac{1}{\sqrt{2}} \left( |0 \rangle^{\otimes n} \pm |1 \rangle^{\otimes n} \rangle \right).
\end{equation}
We learn a classical representation of the GHZ-source and subsequently predict the fidelity with the pure GHZ state.
The right hand side of Figure~\ref{fig:tomoGHZ} highlights that the classical shadow prediction accurately tracks the decrease in target fidelity as the error parameter $p$ increases. NNQST, in contrast, seems to consistently overestimate this target fidelity.
In the extreme case ($p=1$), the true underlying state is completely orthogonal to the target state, but NNQST nonetheless reports fidelities close to one.
This shortcoming arises because the POVM-based machine learning approach can only efficiently estimate an upper bound on the true quantum fidelity efficiently.
To estimate the actual fidelity, an exceedingly large number of measurements is needed.
Similar experiments can be found in Supplementary Section \ref{sec:morenum}, where we focus on toric code ground states and entanglement witnesses, respectively.

\subsection{Predicting  two-point correlation \& subsystem entanglement entropy (Pauli measurements)}

\begin{figure}[t]
    \centering
    \includegraphics[width=0.88\textwidth]{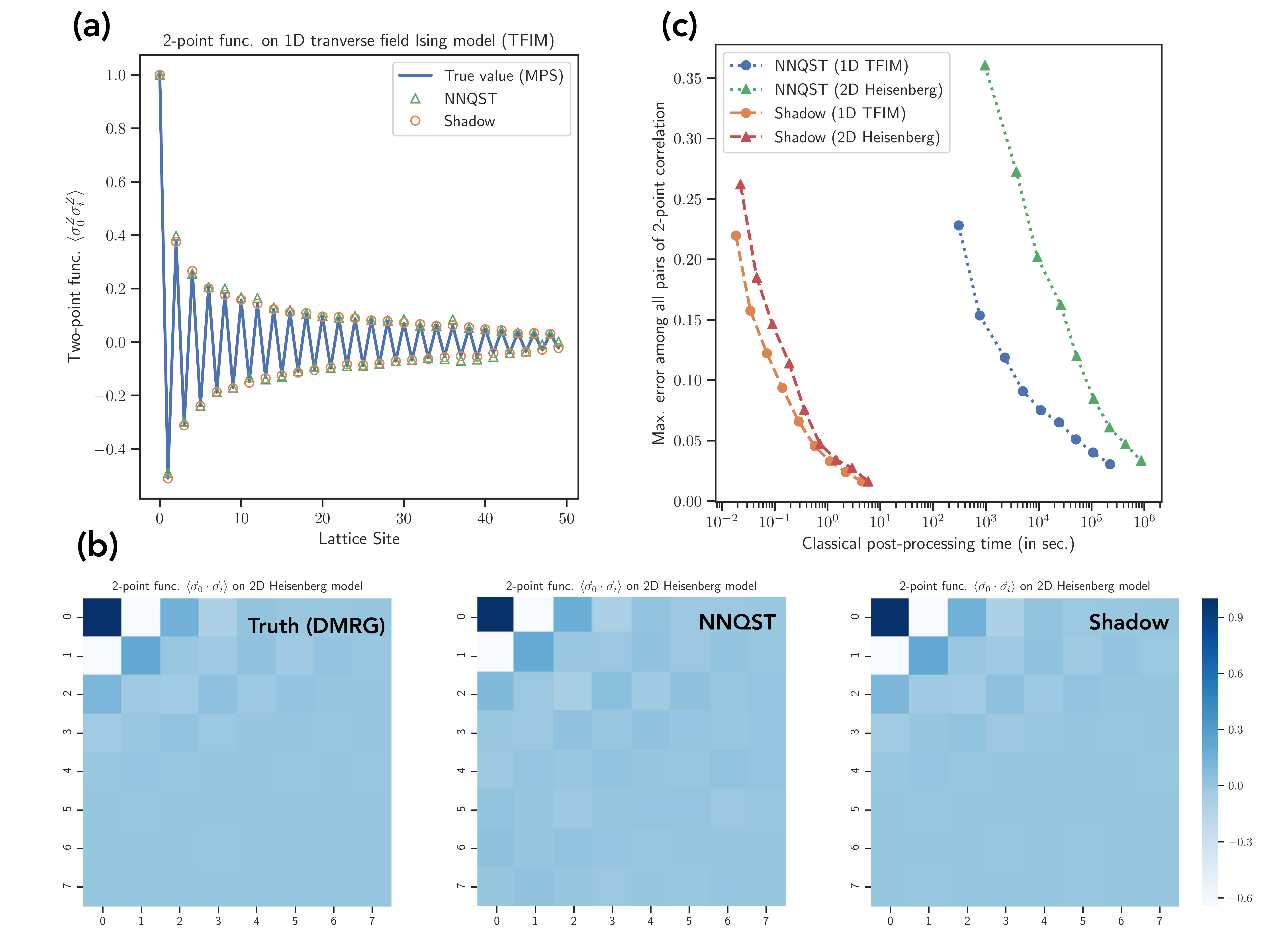}
    \caption{\emph{Predicting two-point correlation functions using classical shadows (Pauli measurements) and NNQST.} \\
\textbf{(a)} \emph{(Top Left)}\textbf{:} Predictions of two-point functions $\langle \sigma^Z_0 \sigma^Z_i \rangle$ for ground states of the one-dimensional critical anti-ferromagnetic transverse field Ising model with $50$ lattice sites. These are based on $2^9 \times 1000$ random Pauli measurements. \\
\textbf{(b)} \emph{(Bottom)}\textbf{:} Predictions of two-point functions $\langle \vec{\sigma}_0 \cdot \vec{\sigma}_i \rangle$ for the ground state of the two-dimensional anti-ferromagnetic Heisenberg model with $8 \times 8$ lattice sites. The predictions are based on $2^9 \times 1000$ random Pauli measurements.\\
\textbf{(c)} \emph{(Top Right)}\textbf{:} The classical processing time (CPU time in seconds) and the prediction error (the largest among all pairs of two-point correlations) over different number of measurements: $\{2^1, \ldots, 2^9\} \times 1000$. The quantum measurement scheme in classical shadows (Pauli) is the same as the POVM-based neural network tomography (NNQST) in \cite{carrasquilla2019reconstructing}. The only difference is the classical post-processing. As the number of measurements increases, the processing time increases, while the prediction error decreases.
}
    \label{fig:2point}
\end{figure}

\begin{figure}[t]
    \centering
    \includegraphics[width=0.97\textwidth]{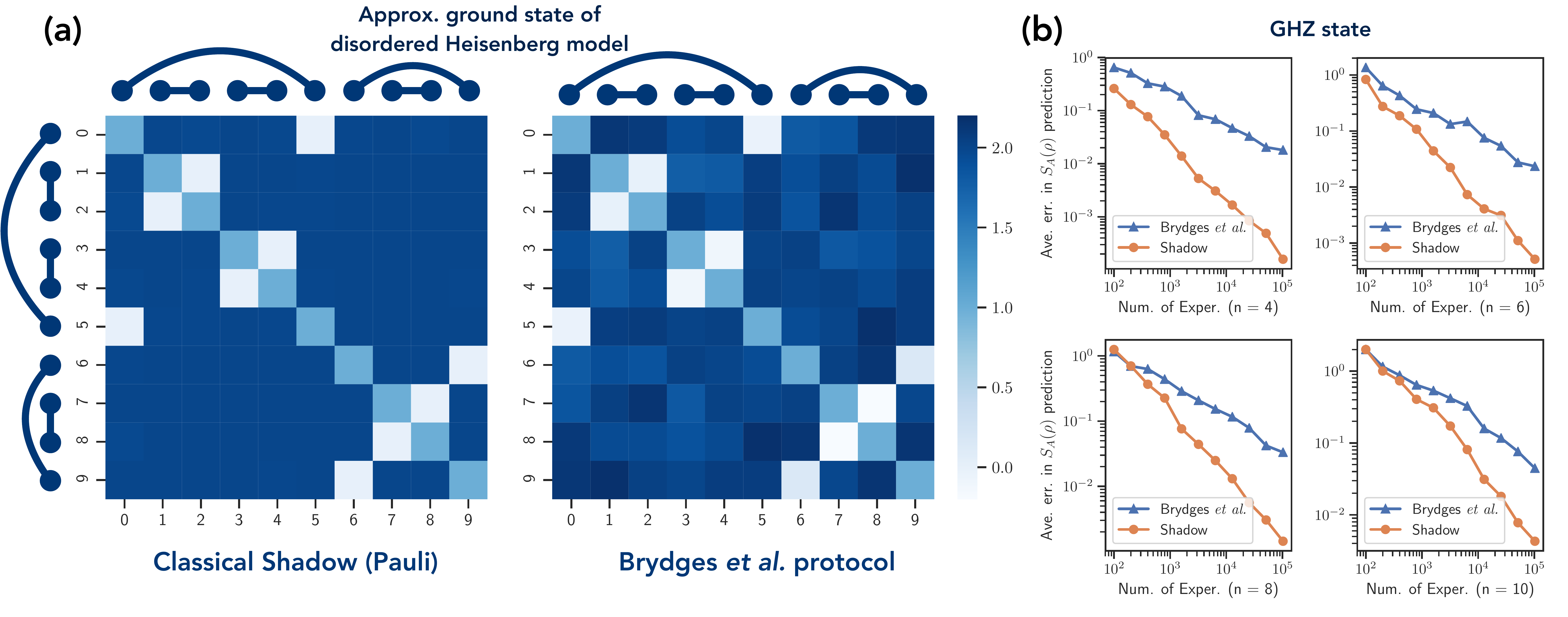}
    \caption{\emph{Predicting entanglement R\'enyi entropies using classical shadows (Pauli measurements) and the Brydges \textit{et al.} protocol.} \\
    \textbf{(a)} \emph{(Left)}\textbf{:} Prediction of second-order R\'enyi entanglement entropy for all subsystems of size at most two in the approximate ground state of a disordered Heisenberg spin chain with 10 sites and open boundary conditions.
    The classical shadow is constructed from $2500$ quantum measurements. The predicted values using the classical shadow visually match the true values with a maximum prediction error of $0.052$. The Brydges \textit{et al.} protocol
    \cite{brydges2019probing} results in a maximum prediction error of $0.24$. \\
    \textbf{(b)} \emph{(Right)}\textbf{:} Comparison of classical shadows and the Brydges \textit{et al.} protocol \cite{brydges2019probing} for estimating second-order R\'enyi entanglement entropy in GHZ states.
    We consider the entanglement entropy of the left-half subsystem with size $n / 2$.
}
    \label{fig:entropy}
\end{figure}

Classical shadows based on random Clifford measurements excel at predicting quantum fidelities.
However, random Clifford measurements can be challenging to implement in practice, because many entangling gates are needed to implement general Clifford circuits.
Next we consider classical shadows based on random local Pauli measurements, which are easier to perform experimentally.
The subsystem properties can be predicted efficiently by constructing the reduced density matrix from the classical shadow.
Therefore, the computational complexity scales exponentially only in the subsystem size, rather than the size of the entire system.
Our numerical experiments confirm that classical shadows obtained using random Pauli measurements excel at predicting few-body properties of a quantum state, such as two-point correlation functions and subsystem entanglement entropy.

\paragraph{Two-point correlation functions.}
NNQST has been shown to predict two-point correlation functions effectively \cite{carrasquilla2019reconstructing}.
Here, we compare classical shadows with NNQST for two physically motivated test cases: ground states of the anti-ferromagnetic transverse field Ising model in one dimension (TFIM) and the anti-ferromagnetic Heisenberg model in two dimensions.
The Hamiltonian for TFIM is $H = J \sum_{i} \sigma^Z_i \sigma^Z_{i+1} + h \sum_i \sigma^X_i$, where $J>0$, and we consider a chain of 50 lattice sites. The critical point occurs at $h = J$ and exhibits power-law decay of correlations rather than exponential decay.
The Hamiltonian for the 2D Heisenberg model is $H = J \sum_{\langle i, j \rangle} \vec{\sigma}_i \cdot \vec{\sigma}_j$, where $J>0$, and we consider an $8 \times 8$ triangular lattice.
We follow the approach in \cite{carrasquilla2019reconstructing}, where the ground state is approximated by a tensor network found using the density matrix renormalization group (DMRG).
Random Pauli measurements on the ground state may then be simulated using this tensor network.
The two methods are compared in Figure~\ref{fig:2point}.
On the top left (a) and bottom (b), we can see that both the classical shadow (with Pauli measurements) and NNQST perform well at predicting two-point correlations. However, NNQST has a larger error for the 2D Heisenberg model; note that for larger separations (the lower right corner of the surface plot), NNQST produces some fictitious oscillations that are not visible in the results from DMRG and classical shadows.
The two approaches use the same quantum measurement data; the only difference is the classical post-processing.
On the top right side (c) of Figure~\ref{fig:2point}, we compare the cost of this classical post-processing, finding roughly a $10^4$ times speedup in classical processing time using the classical shadow instead of NNQST.

\paragraph{Subsystem entanglement entropies.}
An important nonlinear property that can be predicted with classical shadows is subsystem entanglement entropy. The required number of measurements scales exponentially in subsystem size, but is independent of the total number of qubits.
Moreover, these measurements can be used to predict many subsystem entanglement entropies at once.
This problem has also been studied extensively in \cite{brydges2019probing}, where a specialized approach (which we refer to here as the ``Brydges \textit{et al.} protocol'') was designed to efficiently estimate second-order R\'enyi entanglement entropies using random local measurements.
In \cite{brydges2019probing}, a random unitary rotation is reused several times. Predictions using classical shadows could also be slightly modified to adapt to this scenario.
Results from our numerical experiments are shown in Figure~\ref{fig:entropy}.
On the left (a), we predict the entanglement entropy for all subsystems of size $\leq 2$ from only $2500$ measurements of the approximate ground state of the disordered Heisenberg model in one dimension. This is a prototypical model for studying many-body localization \cite{nandkishore2015many}.
The ground state is approximated by a set of singlet states $\{\frac{1}{\sqrt{2}}(\ket{01} - \ket{10})\}$  found using the strong-disorder renormalization group \cite{ma1979random, dasgupta1980low}.
Both, the classical shadow protocol and the Brydges \text{et al.} method use random single-qubit rotations and basis measurements to find a classical representation of the quantum state; the only difference between the methods is in the classical post-processing.
For these small subsystems,
we find that the prediction error of the classical shadow is smaller than the error of the Brydges \text{et al.} protocol.
On the right hand side of Figure~\ref{fig:entropy} (b), we consider predicting the entanglement entropy in a GHZ state for system sizes ranging from $n = 4$ to $n = 10$ qubits.
We focus on the entanglement entropy of the left-half subsystem with system size $n / 2$.
Note that this entanglement entropy is equal to one bit for any system size $n$.
To achieve an error of $0.05$, classical shadows require several times fewer measurements and the discrepancy increases as we require smaller error.

\subsection{Application to quantum simulation of the lattice Schwinger model (Pauli measurements)}

\begin{figure}[t]
    \centering
    \includegraphics[width=0.97\textwidth]{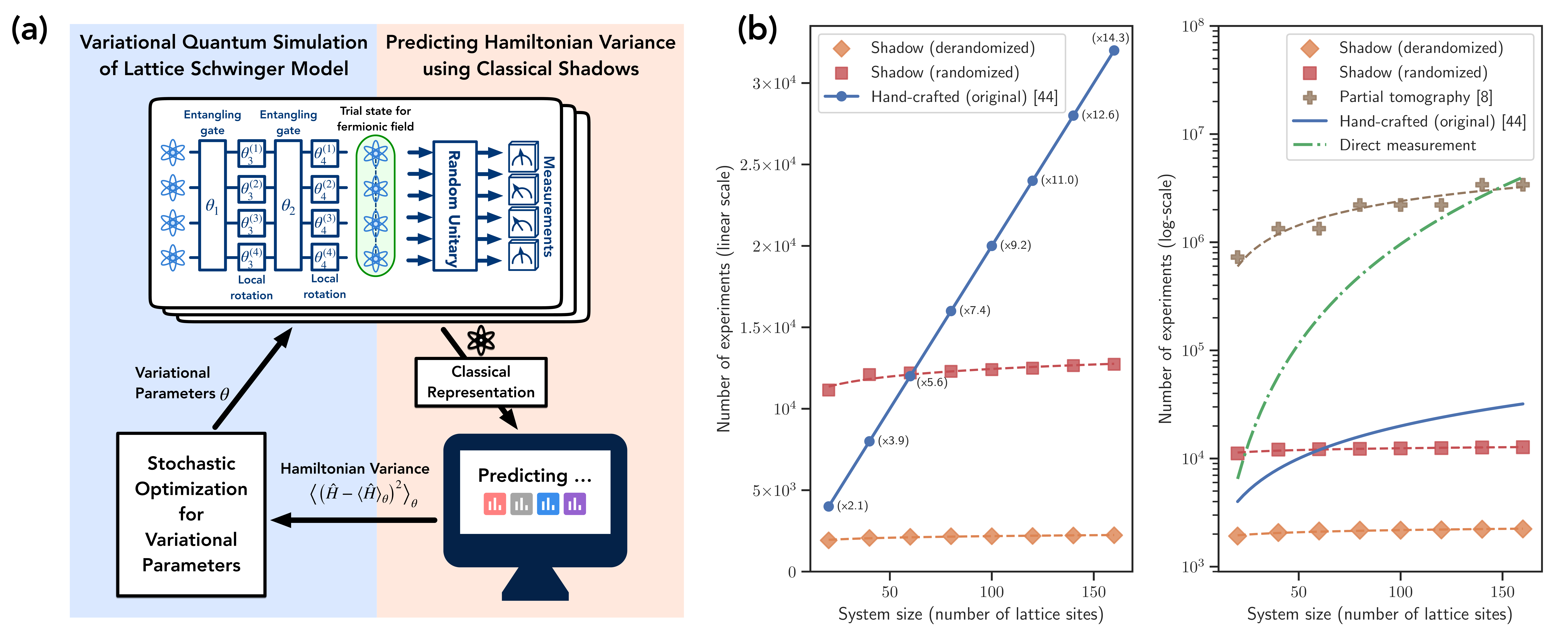}
    \caption{\emph{Application of classical shadows (Pauli measurements) to variational quantum simulation of the lattice Schwinger model.} \\
    \textbf{(a)} \emph{(Left)}\textbf{:} An illustration of variational quantum simulation and the role of classical shadows.
    \\
    \textbf{(b)} \emph{(Right)}\textbf{:} The comparison between different approaches in the number of measurements needed to predict all $4$-local Pauli observables in the expansion of $\langle \big(\hat{H} - \langle\hat{H}\rangle_\theta\big)^2 \rangle_\theta$ with an error equivalent to measuring each Pauli observable at least $100$ times. We include a linear-scale plot that compares classical shadows with the original hand-designed measurement scheme in \cite{kokail2019self} and a log-scale plot that compares with other approaches. In the linear-scale plot, $(\times T)$ indicates that the original scheme uses $T$ times the number of measurements compared to classical shadows (derandomized).
}
    \label{fig:QFT}
\end{figure}

Simulations of quantum field theory using quantum computers may someday advance our understanding of fundamental particle physics. Although high impact discoveries may still be a ways off, notable results have already been achieved in studies of one-dimensional lattice gauge theories using quantum platforms.

For example, in \cite{kokail2019self} a 20-qubit trapped ion analog quantum simulator was used to prepare low-energy eigenstates of the lattice Schwinger model (one-dimensional quantum electrodynamics). The authors prepared a family of quantum states $\{\ket{\psi(\theta)}\}$, where $\theta$ is a variational parameter, and computed the variance of the energy $\langle \big(\hat{H} - \langle\hat{H}\rangle_\theta\big)^2 \rangle_\theta$ for each value of $\theta$. Here $\hat{H}$ is the Hamiltonian of the model, and $\langle\hat O\rangle_\theta =\langle \psi(\theta)|\hat O|\psi(\theta)\rangle$ is the expectation value of the operator $\hat O$ in the state $|\psi(\theta)\rangle$. Because energy eigenstates, and only energy eigenstates, have vanishing energy dispersion, adjusting $\theta$ to minimize the variance of energy prepares an energy eigenstate.

After solving the Gauss law constraint to eliminate the gauge fields, the Hamiltonian $\hat H$ of the Schwinger model is 2-local, though not geometrically local in one dimension. Hence the quantity $\langle \big(\hat{H} - \langle\hat{H}\rangle_\theta\big)^2 \rangle_\theta$ is a sum of expectation values of 4-local observables, which can be measured efficiently using a classical shadow derived from random Pauli measurements. This is illustrated on the left side of Figure~\ref{fig:QFT} (a). On the right side of Figure~\ref{fig:QFT} (b), we compare the performance of classical shadows to the measurement scheme for $4$-local observables designed in \cite{kokail2019self}, and also to a recent method \cite{bonet2019nearly} for measuring local observables, as well as the standard approach that directly measures all observables independently.

The results show, for the methods we considered, the number of copies of the quantum state needed to measure the expectation value of all $4$-local Pauli observables in $\langle \big(\hat{H} - \langle\hat{H}\rangle_\theta\big)^2 \rangle_\theta$ with an error equivalent to measuring each of these observables at least $100$ times.
In \cite{kokail2019self}, such a relatively small number of measurements per local observable already yielded results comparable to theoretical predictions based on exact diagonalization.
We find that the performance of the classical shadow method is better than the method used in \cite{kokail2019self} only for system size larger than $50$ qubits, and may actually be worse for small system sizes. However, classical shadows provide a good prediction for any set of local observables, while the method of \cite{kokail2019self} was hand-crafted for the particular task of estimating the variance of the energy in the Schwinger model.

To make a more apt comparison, we constructed a deterministic version of classical shadows, using a fixed set of measurements rather than random Pauli measurements, specifically adapted for the purpose of estimating $\langle \big(\hat{H} - \langle\hat{H}\rangle_\theta\big)^2 \rangle_\theta$ in the lattice Schwinger model. This deterministic collection of Pauli measurements is obtained by a powerful technique called \textit{derandomization} \cite{raghavan1988pessimistic,spencer1994lectures}.
This procedure simulates the classical shadow scheme based on randomized measurements and makes use of the rigorous performance bound we developed.
When a coin is tossed in the randomized scheme to decide which measurement to perform next, the next measurement in the derandomized version is chosen to have the best possible performance bound for the rest of the protocol. It turns out that this derandomization of the classical shadow method can be carried out very efficiently; full details will appear in upcoming work.
Not surprisingly, the derandomized version, also included in Figure~\ref{fig:QFT}, outperforms the randomized version by a considerable margin. We then find that the derandomized classical shadow method is significantly more efficient than the other methods we considered, including the hand-crafted method from \cite{kokail2019self}.
Finally, we emphasize that the derandomization procedure is fully automated (see \url{https://github.com/momohuang/predicting-quantum-properties} for open source code) and not problem-specific.
It could be used for any pre-specified set of local observables.

\section{Outlook}

A classical shadow is a succinct classical description of a quantum state, which can be extracted by performing reasonably simple single-copy measurements on a reasonably small number of copies of the state. We have shown that, given its classical shadow, many properties of a quantum state can be accurately and efficiently  predicted with a rigorous performance guarantee. In the case of classical shadows based on random Pauli measurements, our methods are feasible using current quantum platforms, and our numerical experiments indicate that many properties can be predicted more efficiently using classical shadows than by using other methods. We therefore anticipate that classical shadows will be useful in near-term experiments characterizing noise in quantum devices and exploring variational quantum algorithms for optimization, materials science, and chemisty. Our results also suggest a variety of avenues for further theoretical exploration. Can the classical shadow of a quantum state be updated efficiently as the state undergoes time evolution governed by a local Hamiltonian? Can we use classical shadows to predict properties of quantum \emph{channels} rather than states? What are the applications of classical shadows based on other ensembles of unitary transformations, for example ensembles of shallow random quantum circuits? More broadly, by mapping many-particle quantum states to succinct classical data, classical shadows open opportunities for applying \emph{classical} machine learning methods to numerous challenging problems in quantum many-body physics \cite{carleo2017solving, carrasquilla2017machine, torlai2018neural}, such as the classification of quantum phases of matter and simulation of strongly correlated quantum phenomena.

\section{Data availability}

Source data are available for this paper. All other data that support the plots within this paper and other findings of this study are available from the corresponding author upon reasonable request.

\section{Code availability}

Source code for an efficient implementation of the proposed procedure is available at \url{https://github.com/momohuang/predicting-quantum-properties}.



\subsection*{Acknowledgments:} The authors want to thank Victor Albert, Fernando Brand\~{a}o, Manuel Endres, Ingo Roth, Joel Tropp, Thomas Vidick and John Wright for valuable inputs and inspiring discussions.
Leandro Aolita and Giuseppe Carleo provided helpful advice regarding presentation.
Our gratitude extends, in particular, to Joseph Iverson who helped us devising a numerical sampling strategy for toric code ground states. We also thank Marco Paini and Amir Kalev for informing us about their related work \cite{paini2019approximate}, where they discussed succinct classical ``snapshots'' of quantum states obtained from randomized local measurements.
HH is supported by the Kortschak Scholars Program.
RK acknowledges funding provided by the Office of Naval Research (Award N00014-17-1-2146) and the Army Research Office (Award W911NF121054). JP acknowledges funding from  ARO-LPS, NSF, and DOE. The Institute for Quantum Information and Matter is an NSF Physics Frontiers Center.

\subsection*{Author Contributions:} H.H. and R.K. developed the theoretical aspects of this work. H.H. conducted the numerical experiments and wrote the open source code. J.P. conceived the applications of classical shadows. H.H., R.K. and J.P. wrote the manuscript.

\subsection*{Competing interests:} The authors declare no competing interests.

\vspace{2em}
\begin{center}
\textbf{\Large Supplementary information}
\end{center}

\renewcommand{\figurename}{Supplementary Figure}
\renewcommand{\theequation}{S\arabic{equation}}
\renewcommand{\thesection}{\arabic{section}}
\renewcommand{\thesubsection}{\Alph{subsection}}
\setcounter{figure}{0}
\setcounter{equation}{0}
\setcounter{theorem}{0}
\setcounter{secnumdepth}{3}

\section{General framework for constructing classical shadows}
\label{sec:generalshadow}

\subsection{Data acquisition and classical shadows}

Throughout this work we restrict attention to multi-qubit systems and
$\rho$ is a fixed, but unknown, quantum state in $d = 2^n$ dimensions.
We present a general-purpose strategy for predicting many properties of this unknown state.
To extract meaningful information about $\rho$, we need to perform a collection of measurements.

\begin{definition}[measurement primitive]
We can apply a restricted set of unitary evolutions $\rho \mapsto U \rho U^\dagger$, where $U$ is chosen from an ensemble $\mathcal{U}$.
Subsequently, we can measure the rotated state in the computational basis $\left\{|b \rangle:\; b\in \left\{0,1\right\}^n \right\}$. Moreover, we assume that this collection is tomographically complete, i.e.\ for each $\sigma \neq \rho$ there exist $U \in \mathcal{U}$ and $b$ such that $\langle b| U \sigma U^\dagger |b \rangle \neq \langle b| U \rho U^\dagger |b \rangle$.
\end{definition}

Based on this primitive, we repeatedly perform a simple randomized measurement procedure: randomly rotate the state $\rho \mapsto U \rho U^\dagger$ and perform a computational basis measurement. Then, after the measurement, we apply the inverse of $U$ to the resulting computational basis state. This procedure collapses $\rho$ to
\begin{equation}
U^\dagger |\hat{b} \rangle \! \langle \hat{b}|U \quad \text{where} \quad \mathrm{Pr} \big[ \hat{b}=b \big] = \langle b| U \rho U^\dagger |b \rangle,\; b \in \left\{0,1\right\}^n \quad \text{(Born's rule)}.
\label{eq:born-rule}
\end{equation}
This random snapshot contains valuable information about $\rho$ in expectation:
\begin{equation}
\mathbb{E} \left[ U^\dagger | \hat{b} \rangle \! \langle \hat{b} | U \right]
= \mathbb{E}_{U \sim \mathcal{U}} \sum_{b \in \left\{0,1\right\}^n} \langle b| U \rho U^\dagger |b \rangle U^\dagger |b \rangle \! \langle b| U=\mathcal{M}(\rho).
\label{eq:measurement-channel}
\end{equation}
For any unitary ensemble $\mathcal{U}$, this relation describes a quantum channel $\rho \mapsto \mathcal{M}(\rho)$.
Tomographic completeness ensures that $\mathcal{M}$ --- viewed as a linear map --- has a unique inverse $\mathcal{M}^{-1}$ and we set
\begin{align}
\hat{\rho} = \mathcal{M}^{-1} \left( U^\dagger | \hat{b} \rangle \! \langle \hat{b}|U \right) && \text{(classical shadow)}.
\label{eq:classical-shadow-appendix}
\end{align}
The classical shadow is a modified post-measurement state that has unit trace, but need not be positive semi-definite. However, it
is designed to reproduce the underlying state $\rho$ exactly in expectation: $\mathbb{E} \left[ \hat{\rho} \right] = \rho$.
This classical shadow $\hat\rho$ corresponds to the linear inversion (or least squares) estimator of $\rho$ in the single-shot limit.
Linear inversion estimators have been used to perform full quantum state tomography \cite{sugiyama2013tomography,guta2018fast}, where an exponential number of measurements is needed.
We wish to show that $\hat \rho$ can predict many properties from only very few measurements.

\subsection{Predicting linear functions with classical shadows}
\label{sec:linearfunc-general}

Classical shadows are well suited to predict \emph{linear} functions in the unknown state $\rho$:
\begin{equation}
o_i = \mathrm{tr} \left( O_i \rho \right) \quad 1 \leq i \leq M.
\end{equation}
To achieve this goal, we simply replace the (unknown) quantum state $\rho$ by a classical shadow $\hat{\rho}$. Since classical shadows are random, this produces a random variable that yields the correct prediction in expectation:
\begin{equation}
\hat{o}_i= \mathrm{tr} \left( O_i \hat{\rho}\right) \quad \text{obeys} \quad \mathbb{E} \left[ \hat{o} \right] = \mathrm{tr} \left( O_i \rho \right).
\end{equation}
Fluctuations of $\hat{o}$ around this desired expectation are controlled by the variance.

\begin{lemma} \label{lem:variance}
Fix $O$ and set $\hat{o}=\mathrm{tr} \left( O \hat{\rho}\right)$, where $\hat{\rho}$ is a classical shadow \eqref{eq:classical-shadow-appendix}. Then
\begin{equation}\label{eq:Var-bound}
\mathrm{Var} \left[ \hat{o} \right] = \mathbb{E} \left[ \left(\hat{o}-\mathbb{E} \left[ \hat{o} \right] \right)^2 \right]
\leq \left\| O - \tfrac{\mathrm{tr}(O)}{2^n} \mathbb{I} \right\|_{\mathrm{shadow}}^2.
\end{equation}
The norm $\norm{\cdot}_{\mathrm{shadow}}$ only depends on the measurement primitive:
\begin{equation}
\| O \|_{\mathrm{shadow}}
= \max_{\sigma: \text{ state}}\Big(\mathbb{E}_{U \sim \mathcal{U}}\sum_{b \in \left\{0,1\right\}^n}\langle b| U \sigma U^\dagger|b \rangle \langle b| U \mathcal{M}^{-1} \left(O\right) U^\dagger |b \rangle^2 \Big)^{1/2}.
\label{eq:shadow-norm}
\end{equation}
\end{lemma}

\noindent It is easy to check that $\|O\|_{\mathrm{shadow}}$ is nonnegative and homogeneous ($\|0\|_{\mathrm{shadow}}=0$). After some work, one can verify that this expression also obeys the triangle inequality, and so is indeed a norm.

\begin{proof}
Classical shadows have unit trace
by construction ($\mathrm{tr}(\hat{\rho})=1$). This feature implies that the variance only depends on the traceless part $O_0 = O - \frac{\mathrm{tr}(O)}{2^n}\mathbb{I}$ of  $O$, not $O$ itself:
\begin{equation}
\hat{o}-\mathbb{E} [\hat{o}] = \mathrm{tr} \left( O \hat{\rho}\right) - \mathrm{tr} \left( O \rho \right)
= \mathrm{tr} \left( O_0 \hat{\rho} \right) -\mathrm{tr} \left( O_0 \rho \right).
\label{eq:traceless}
\end{equation}
Moreover, it is easy to check that the inverse of $\mathcal{M}$ \eqref{eq:measurement-channel} is self-adjoint  ($\mathrm{tr}\left( X \mathcal{M}^{-1}(Y) \right) = \mathrm{tr} \left( \mathcal{M}^{-1}(X) Y \right)$ for any pair of matrices $X,Y$ with compatible dimension).
These two observations allow us to rewrite the variance in the following fashion:
\begin{align}
\mathrm{Var} \left[ \hat{o} \right]
=& \mathbb{E} \left[\left( \hat{o} - \mathbb{E} \hat{o} \right)^2 \right]
= \mathbb{E}\left[ \left(\mathrm{tr}(O_0 \hat{\rho})\right)^2\right]
-\left(\mathrm{tr} \left( O_0 \E\left[\hat \rho\right] \right)\right)^2
= \mathbb{E}\left[ \langle \hat{b} | U \mathcal{M}^{-1} (O_0) U^\dagger |\hat{b} \rangle^2\right] - \left(\mathrm{tr} \left( O_0 \rho \right)\right)^2.
\end{align}
Classical shadows arise from mixing two types of randomness: (i) a (classical) random choice of unitary $U \sim \mathcal{U}$ and (ii) a random choice of computational basis state $| \hat{b} \rangle$ that is governed by Born's rule \eqref{eq:born-rule}. Inserting the average over computational basis states produces
a (squared) norm that closely resembles the advertised expression, but does depend on the underlying state:
\begin{equation}
\mathbb{E} \langle \hat{b}  | U \mathcal{M}^{-1} (O_0) U^\dagger |\hat{b} \rangle^2
= \mathbb{E}_{U \sim \mathcal{U}} \sum_{b \in \left\{0,1\right\}^n}\langle b| U \rho U^\dagger |b \rangle
\langle {b}  | U \mathcal{M}^{-1} (O_0) U^\dagger |{b} \rangle^2.
\end{equation}
Maximizing over all possible states $\sigma$ removes this implicit dependence and produces a universal upper bound on the variance.
Ignoring the subtraction of $\left(\mathrm{tr} \left( O_0 \rho \right)\right)^2$ (which can only make the bound tighter),
we obtain (\ref{eq:Var-bound}).
\end{proof}

Lemma~\ref{lem:variance} sets the stage for successful linear function estimation with classical shadows. A single classical shadow \eqref{eq:classical-shadow-appendix} correctly predicts \emph{any} linear function $o_i=\mathrm{tr}(O_i\rho)$ in expectation.
Convergence to this desired target can be boosted by forming empirical averages of multiple independent shadow predictions. The \emph{empirical mean} is the canonical example for such a procedure. Construct $N$ independent classical shadows $\hat{\rho}_1,\ldots,\hat{\rho}_N$ and set
\begin{equation}
\label{eq:mean_estimator}
\hat{o}_i(N,1) = \frac{1}{N}\sum_{j=1}^N \mathrm{tr} \left( O_i \hat{\rho}_j \right).
\end{equation}
Each summand is an independent random variable with correct expectation and variance bounded by Lemma~\ref{lem:variance}.
Convergence to the expectation value $\mathrm{tr}(O_i \rho)$ can be controlled by
classical concentration arguments (e.g. Chernoff or Hoeffding inequalities).
In order to achieve a failure probability of (at most) $\delta$, the number of samples must scale like $N=\mathrm{Var} \left[\hat{o}_i \right]/(\delta \epsilon^2)$. While the scaling in variance and approximation accuracy $\epsilon$ is optimal, the dependence on $1/\delta$ is particularly bad. Unfortunately, this feature of sample mean estimators cannot be avoided without imposing additional assumptions (that do not apply to classical shadows).
\emph{Median of means} \cite{nemirovski1983medianmeans,jerrum1986medianmeans}
is a conceptually simple trick that addresses this issue.
Instead of using all samples to construct a single empirical mean \eqref{eq:mean_estimator}, construct $K$
independent sample means and form their median:
\begin{equation}
\hat{o}_i(N,K) = \mathrm{median}\left\{ \hat{o}_i^{(1)}(N,1),\ldots,\hat{o}_i^{(K)}(N,1) \right\} \quad \textrm{where} \quad  \hat{o}_i^{(k)} = \tfrac{1}{N}\sum_{j=N(k-1)+1}^{Nk} \mathrm{tr} \left( O_i \hat{\rho}_{j}\right)
\label{eq:median-of-means}
\end{equation}
for $1 \leq k \leq K$.
This estimation technique requires $NK$ samples in total, but it is much more robust with respect to outlier corruption.  Indeed, $| \hat{o}(N,K)-\mathrm{tr}(O \rho) | > \epsilon$ if and only if more than half of the empirical means individually deviate by more than $\epsilon$. The probability associated with such an undesirable event decreases exponentially with the number of batches $K$. This results in an exponential improvement over sample mean estimation in terms of failure probability. The main result of this work capitalizes on this improvement.

\begin{theorem} \label{thm:general}
Fix a measurement primitive $\mathcal{U}$, a collection $O_1,\ldots,O_M$ of $2^n \times 2^n$ Hermitian matrices and accuracy parameters $\epsilon,\delta \in [0,1]$.
Set
\begin{equation}
K = 2 \log (2M/\delta) \quad \text{and} \quad N = \frac{34}{\epsilon^2}\max_{1 \leq i \leq M} \| O_i- \tfrac{\mathrm{tr}(O_i)}{2^n}\mathbb{I} \|_{\mathrm{shadow}}^2,
\end{equation}
where $\| \cdot \|_{\mathrm{shadow}}$ denotes the norm defined in Eq.~\eqref{eq:shadow-norm}. Then, a collection of $NK$ independent classical shadows allow for accurately predicting all features
via median of means prediction \eqref{eq:median-of-means}:
\begin{equation}
\left| \hat{o}_i (N,K) - \mathrm{tr} \left( O_i \rho \right) \right| \leq \epsilon \quad \text{for all $1 \leq i \leq M$}
\end{equation}
with probability at least $1-\delta$.
\end{theorem}

\begin{proof}
The claim follows from combining the variance estimates from Lemma~\ref{lem:variance} with a rigorous performance guarantee for median of means estimation \cite{nemirovski1983medianmeans,jerrum1986medianmeans}: Let $X$ be a random variable with variance $\sigma^2$. Then, $K$ independent sample means of size $N=34 \sigma^2/\epsilon^2$ suffice to construct a median of means estimator $\hat{\mu}(N,K)$ that obeys
$\mathrm{Pr} \left[ \left| \hat{\mu}(N,K)-\mathbb{E} \left[ X \right] \right| \geq \epsilon \right] \leq 2 \mathrm{e}^{-K/2}$ for all $\epsilon >0$.
The parameters $N$ and $K$ are chosen such that this general statement ensures
$
\mathrm{Pr} \left[ \left| \hat{o}_i (N,K) - \mathrm{tr} \left( O_i \rho \right) \right| \geq \epsilon \right] \leq \frac{\delta}{M}
$ for all $1 \leq i \leq M$. Apply a union bound over all $M$ failure probabilities to deduce the claim.
\end{proof}

\begin{remark}[Constants in Theorem~\ref{thm:general}]
The numerical constants featuring in $N$ and $K$ result from a conservative (worst case) argument that is designed to be simple, not tight. We expect that the actual constants are \emph{much smaller} in practice.
\end{remark}

Each classical shadow is the result of a single quantum measurement on $\rho$. Viewed from this angle,
Theorem~\ref{thm:general} asserts that a total of
\begin{align}
N_{\mathrm{tot}}=& \mathcal{O} \left( \frac{\log(M)}{\epsilon^2}\max_{1 \leq i \leq M} \left\|O_i- \tfrac{\mathrm{tr}(O_i)}{2^n}\mathbb{I} \right\|_{\mathrm{shadow}}^2 \right) & \text{(sample complexity)}
\label{eq:sample-rate}
\end{align}
measurement repetitions suffice to accurately predict a collection of $M$ linear target functions $\mathrm{tr}(O_i \rho)$.

 Importantly, this sample complexity only scales logarithmically in the number of target functions $M$. Moreover, the problem dimension $2^n$ does not feature explicitly.
The sample complexity does, however, depend on the measurement primitive via the norm $\norm{\cdot}_{\mathrm{shadow}}$.
This term reflects expressiveness and structure of the measurement primitive in question.
This subtle point is best illustrated with two concrete examples. We defer technical derivations to subsequent sections and content ourselves with summarizing the important aspects here.

\paragraph{Example 1: Random Clifford measurements}
Clifford circuits are generated by CNOT, Hadamard and Phase gates and form the group $\mathrm{Cl}(2^n)$.
The ``random global Clifford basis'' measurement primitive --- $\mathcal{U}=\mathrm{Cl}(2^n)$ (endowed with uniform weights) --- implies the following simple expression for classical shadows and the associated norm $\norm{\cdot}_{\mathrm{shadow}}$:
\begin{equation}
\hat{\rho} = (2^n+1) U^\dagger |\hat{b} \rangle \! \langle \hat{b}|U-\mathbb{I} \quad \textrm{and} \quad \left\| O - \tfrac{\mathrm{tr}(O)}{2^n}\mathbb{I} \right\|_{\mathrm{shadow}}^2 \leq 3 \mathrm{tr}(O^2).
\end{equation}
We refer to Supplementary Section~\ref{sub:clifford-shadow-details} for details and proofs.
Combined with Eq.~\eqref{eq:sample-rate}, this ensures that $\mathcal{O}(\log (M) \max_i \mathrm{tr}(O_i^2)/\epsilon^2)$ random global Clifford basis measurements suffice to accurately predict $M$ linear functions.
This prediction technique is most powerful, when the target functions have constant Hilbert-Schmidt norm. In this case, the sample rate is completely independent of the problem dimension $2^n$.
Prominent examples include estimating quantum fidelities (with pure states), or entanglement witnesses.

\paragraph{Example 2: Random Pauli measurements}
Although (global) Clifford circuits are believed to be much more tractable than general quantum circuits, they still feature entangling gates, like CNOT. Such gates are challenging to implement reliably on today's devices. The ``random Pauli basis'' measurement primitive takes this serious drawback into account and assumes that one is only able to apply single-qubit Clifford gates, i.e. $U=U_{1} \otimes \cdots \otimes U_{n} \sim \mathcal{U}=\mathrm{Cl}(2)^{\otimes n}$ (endowed with uniform weights). This is equivalent to assuming that we can perform arbitrary Pauli (basis) measurements, i.e., measuring each qubit in the $X$-, $Y$- and $Z$-basis, respectively. Such basis measurements decompose nicely into tensor  products ($U|\hat{b} \rangle = \bigotimes_{j=1}^n U_{j} |b_j \rangle$ for $b=(b_1,\ldots,b_n) \in \left\{0,1\right\}^n$) and respect locality. The associated classical shadows and the norm $\norm{\cdot}_{\mathrm{shadow}}$ inherit these desirable features:
\begin{equation}
\hat{\rho} = \bigotimes_{j=1}^n \left( 3 U_{j}^\dagger|\hat{b}_j \rangle \! \langle \hat{b}_j| U_{j}-\mathbb{I} \right)
\quad \textrm{and} \quad \left\| O-\tfrac{\mathrm{tr}(O)}{2^n} \right\|_{\mathrm{shadow}}^2 \leq 4^{\mathrm{locality}(O)} \|O \|_\infty^2.
\end{equation}
Here, $\mathrm{locality}(O)$ counts the number of qubits on which $O$ acts nontrivially.
We refer to Supplementary Section~\ref{sub:pauli-shadow-details} for details and proofs.
Combined with Eq.~\eqref{eq:sample-rate} this ensures that $\mathcal{O} \left( \log (M) 4^k/\epsilon^2\right)$ local Clifford (Pauli) basis measurements suffice to predict $M$ bounded observables that are at most $k$-local.
For observables that are the tensor product of $k$ single-qubit observables, the sample complexity can be further improved to $\mathcal{O} \left( \log (M) 3^k/\epsilon^2\right)$.
This prediction technique is most powerful when the target functions do respect some sort of locality constraint. Prominent examples include $k$-point correlators, or individual terms in a local Hamiltonian.

\paragraph{Discussion and information-theoretic optimality}

These two examples complement each other nicely. Random Clifford measurements excel at performing useful subroutines in quantum computing and communication tasks, such as certifying (global) entanglement, which will be feasible using sufficiently advanced hardware. Their practical utility, however, hinges on the ability to execute circuits with many entangling gates.
Random Pauli measurements, on the other hand, are much less demanding from a hardware perspective. In today's NISQ era, local Pauli operators can be accurately measured using available hardware platforms. While not well-suited for predicting global features, Pauli measurements excel at making local predictions. Furthermore, for both kinds of randomized measurements, linear prediction based on classical shadows saturates fundamental lower bounds from information theory.

\begin{theorem}[random Clifford measurements; informal version] \label{thm:main2Cl}
\emph{Any} procedure based on a fixed set of single-copy measurements that can predict, with additive error $\epsilon$, $M$ arbitrary linear functions $\mathrm{tr}(O_i\rho) $, requires at least $\Omega(\log (M) \max_i \Tr(O_i^2) / \epsilon^2)$ copies of the state $\rho$.
\end{theorem}
\begin{theorem}[random Pauli measurements; informal version] \label{thm:main2P}
\emph{Any} procedure based on a fixed set of single-copy local measurements that can predict, with additive error $\epsilon$, $M$ arbitrary $k$-local linear functions $\mathrm{tr}(O_i\rho) $, requires at least $\Omega(\log (M) 3^k / \epsilon^2)$ copies of the state $\rho$.\
\end{theorem}

We refer to Supplementary Section \ref{sec:proofthm2} (Clifford) and \ref{sec:proofthm2P} (Pauli) for further context, details and proofs.
In the random Pauli basis measurement setting, classical shadows provably saturate this lower bound only for tensor product observables. For general $k$-local observables, there is a small discrepancy between $4^k$ (upper bound) and $3^k$ (lower bound).

\subsection{Predicting nonlinear functions with classical shadows}

Feature prediction with classical shadows readily extends beyond the linear case. Here, we shall focus on quadratic functions, but the procedure and analysis readily extend to higher order polynomials.
Every quadratic function in an unknown state $\rho$ can be recast as a linear function acting on the tensor product $\rho \otimes \rho$:
\begin{equation}
\hat{o}_i = \mathrm{tr} \left( O_i \rho \otimes \rho \right) \quad 1 \leq i \leq M. \label{eq:quadratic-functions}
\end{equation}
 An immediate generalization of linear feature prediction with classical shadows suggests the following procedure. Take two independent snapshots $\hat{\rho}_1,\hat{\rho}_2$ of the unknown state $\rho$ and set
\begin{equation}
\hat{o}_i = \mathrm{tr} \left( O_i \hat{\rho}_1 \otimes \hat{\rho}_2 \right) \quad \text{such that} \quad \mathbb{E} \hat{o}_i = \mathrm{tr} \left( O_i \mathbb{E} \hat{\rho}_1 \otimes \mathbb{E} \hat{\rho}_2 \right) = \mathrm{tr} \left( O_i \rho \otimes \rho \right) = o_i.
\end{equation}
This random variable is designed to yield the correct target function in expectation. Similar to linear function prediction we can boost convergence to this desired target by forming empirical averages.
To make the best of use of $N$ samples, we average over all $N (N-1)$ (distinct) pairs:
\begin{equation}
\hat{o}_i(N,1)= \frac{1}{N(N-1)}\sum_{j \neq l}\mathrm{tr} \left( O_i \hat{\rho}_j \otimes \hat{\rho}_l \right). \label{eq:quadratic-estimator}
\end{equation}
This idea provides a systematic approach for constructing estimators for nonlinear (polynomial) functions.
Estimators of this form always yield the desired target in expectation.
For context, we point out that the estimator \eqref{eq:quadratic-estimator} closely resembles the sample variance, while estimators of higher order polynomials are known as \emph{U-statistics} \cite{hoeffding1992class}.
Fluctuations of $\hat{o}_i (N,1)$ around its desired expectation are once more controlled by the variance. U-statistics estimators are designed to minimize this variance and therefore considerably boost the rate of convergence.

\begin{lemma} \label{lem:symm}
Fix $O$ and a sample size $N$. Then, the variance of the U-statistics estimator \eqref{eq:quadratic-estimator} obeys
\begin{align}
    \Var[\hat{o}(N, 1)] \leq
    \frac{2}{N} \left(
      \Var[\Tr(O \hat\rho_1 \otimes \rho) ]
     + \Var[\Tr(O \rho \otimes \hat\rho_1) ]+\frac{1}{N}\Var[\Tr(O \hat\rho_1 \otimes \hat\rho_2)] \right). \label{eq:varsimplebound}
\end{align}
\end{lemma}

We emphasize that this variance decreases with the number of samples $N$. This sets the stage for successful quadratic function prediction with classical shadows.
Similar to the linear case, we will not use all samples to construct a single U-statistics estimator.
Instead, we construct $K$ of them and form their median:
\begin{align}
\hat{o}_i(N,K) =& \mathrm{median}\left\{ \hat{o}_i^{(1)}(N,1),\ldots,\hat{o}_i^{(K)}(N,1) \right\}, \quad \textrm{where} \nonumber \\ \hat{o}_i^{(k)}(N,1) =& \tfrac{1}{N(N-1)}\sum_{\substack{j \neq l \\ j,l \in \{N(k-1)+1, \ldots , Nk \}}} \mathrm{tr} \left( O_i \hat{\rho}_j \otimes \hat{\rho}_l\right)\quad \text{for $1 \leq k \leq K$}. \label{eq:quadratic-median}
\end{align}
This renders the entire estimation procedure more robust to outliers and exponentially suppresses failure probabilities.

\begin{theorem} \label{thm:general-quadratic}
Fix a measurement primitive $\mathcal{U}$, a collection $O_1,\ldots,O_M$ of (quadratic) target functions and accuracy parameters $\epsilon,\delta \in [0,1]$.
Set
\begin{align}
K =& 2 \log (2M/\delta) \quad \text{and} \nonumber \\
N =& \frac{34}{\epsilon^2}\max_{1 \leq i \leq M} 8 \times \max\left(\Var[\Tr(O_i \rho \otimes \hat\rho_1)], \Var[\Tr(O_i \hat{\rho}_1 \otimes \rho)], \sqrt{\Var[\Tr(O_i \hat{\rho}_1 \otimes \hat{\rho}_2)]}\right).
\end{align}
Then, a collection of $NK$ independent classical shadows allow for accurately predicting all quadratic features via the median of U-statistics estimators \eqref{eq:quadratic-median}:
\begin{equation}
\left| \hat{o}_i (N,K) - \mathrm{tr} \left( O_i \rho \otimes \rho \right) \right| \leq \epsilon \quad \text{for all $1 \leq i \leq M$}
\end{equation}
with probability at least $1-\delta$.
\end{theorem}
\begin{proof}
The proof is similar to the argument for linear prediction. We combine the bound on the variance of U-statistics estimators from Lemma~\ref{lem:symm} with a rigorous performance guarantee for median estimation \cite{nemirovski1983medianmeans,jerrum1986medianmeans}. Let $Z$ be a random variable with variance at most $\epsilon^2 / 34$. Then,
setting $\hat{\mu}=\mathrm{median}\left\{ Z_1,\ldots,Z_k \right\}$
produces an estimator that obeys
$\mathrm{Pr} \left[ \left| \hat{\mu}-\mathbb{E} \left[ Z \right] \right| \geq \epsilon \right] \leq 2 \mathrm{e}^{-K/2}$.
The parameter $N$ is chosen ensure that each $\hat{o}^{(k)}_i(N, 1)$ has variance at most $\epsilon^2 / 34$.
The parameter $K$ is chosen such that each probability of failure is at most $\delta/M$.
The advertised statement then follows from taking a union bound over all $M$ target estimations.
\end{proof}

\begin{remark}[Constants in Theorem~\ref{thm:general-quadratic}]
The numerical constants featuring in $N$ and $K$ result from a conservative (worst case) argument that is designed to be simple, not tight. We expect that the actual constants are \emph{much smaller} in practice.
\end{remark}

Theorem~\ref{thm:general-quadratic} is a general statement that provides upper bounds for the sample complexity associated with predicting quadratic target functions:
\begin{equation} \label{eq:totalnonlinear}
N_{\mathrm{tot}} = \mathcal{O}\left( \frac{\log(M)}{\epsilon^2} \max_{1 \leq i \leq M} \max\left(\Var[\Tr(O_i \rho \otimes \hat\rho_1)], \Var[\Tr(O_i \hat{\rho}_1 \otimes \rho)], \sqrt{\Var[\Tr(O_i \hat{\rho}_1 \otimes \hat{\rho}_2)]} \right) \right)
\end{equation}
independent randomized measurements suffice to accurately predict a collection of $M$ nonlinear target functions $\Tr(O_i \rho \otimes \rho)$. This sampling rate once more depends on the measurement primitive and it is instructive to consider concrete examples.

\paragraph{Example 1: Random Pauli measurements}

We first discuss the practically more relvant example for today's NISQ era: classical shadows constructed from random single-qubit Pauli basis measurements. This measurement primitive remains well-suited for predicting \emph{local} quadratic features $\mathrm{tr}(O \rho \otimes \rho)$. Suppose that $O$ acts nontrivially on $k$ qubits in the first state copy and on $k$ qubits in the second state copy. Thus, when viewed as an observable for a $2n$-qubit system, $O$ is $2k$-local.
A technical argument shows that  the maximum of the variances in Equation~\eqref{eq:totalnonlinear} is bounded by $4^k$. We emphasize that this scaling is much better than the naive guess $4^{2k}$ -- one of the key advantages of U-statistics.
Hence we only need a total number of $N_{\mathrm{tot}} = \mathcal{O}(\log(M) 4^{k} / \epsilon^2)$ random Pauli basis measurements to predict $M$ quadratic functions $\Tr(O_i \rho \otimes \rho)$.
An important concrete application of this procedure is the prediction of subsystem R\'enyi-2 entanglement entropies.

\paragraph{Example 2: Random Clifford measurements}

Theorem~\ref{thm:general-quadratic} also applies to the global Clifford measurement primitive.
There, the maximum of the variances in Equation~\eqref{eq:totalnonlinear} can be bounded by $\sqrt{9+6/2^n} \Tr(O_i^2) \simeq 3 \Tr (O_i^2)$.
Hence we only need a total number of $N_{\mathrm{tot}} = \mathcal{O}(\log(M) \max_i \Tr(O_i^2) / \epsilon^2)$ random Clifford basis measurements to predict $M$ quadratic functions $\Tr(O_i \rho \otimes \rho)$.
While a clean extension of linear feature prediction with Clifford basis measurements, the applicability of this result seems somewhat limited. Interesting global quadratic features tend to have prohibitively large Hilbert-Schmidt norms.
The purity $\mathrm{tr}(\rho^2)$ provides an instructive non-example. It can be written as $\mathrm{tr} \left( S \rho \otimes \rho \right)$, where $S|\psi \rangle \otimes | \phi \rangle = | \phi \rangle \otimes | \psi \rangle$ denotes the swap operator. Alas,  $\mathrm{tr}(S^2)=\mathrm{tr}(\mathbb{I})=2^n$ which scales exponentially in the number of qubits.
Nonetheless,
quadratic feature prediction with Clifford measurements is by no means useless.
For instance, it can help provide statistical \textit{a posteriori} guarantees on the quality of linear feature prediction --- for example, by estimating sample variances to construct confidence intervals.

\section{Additional numerical experiments}
\label{sec:morenum}

In this section we report additional numerical experiments that demonstrate the viability of linear feature prediction with classical shadows. We focus on the Clifford basis measurement primitive, \textit{i.e.} applying a random Clifford circuit to $\rho$ and then measuring in the computational basis.

\subsection{Direct fidelity estimation for the toric code ground state}
\label{sec:toricexp}

\begin{figure}[t]
    \centering
    \includegraphics[width=1.0\textwidth]{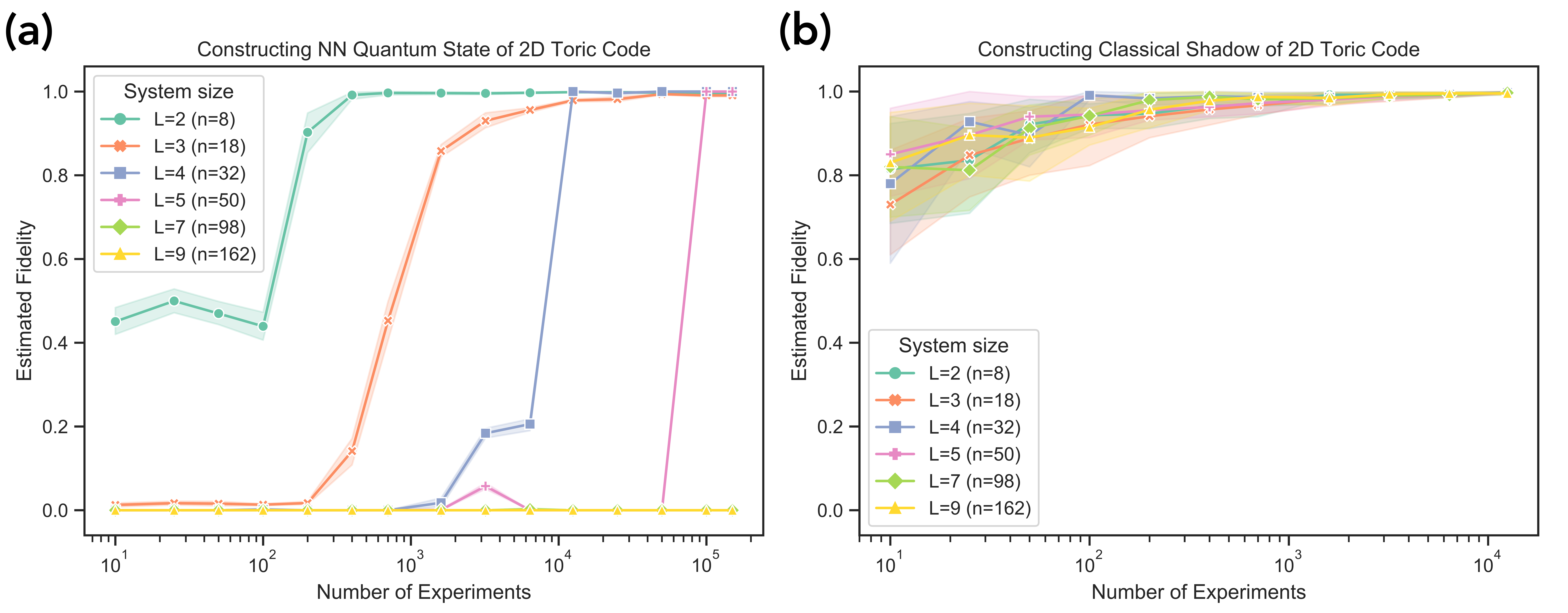}
    \caption{\emph{Comparison between classical shadow and neural network tomography (NNQST); toric code.} \\
\textbf{(a)} \emph{(Left)}\textbf{:} Number of measurements required for neural network tomography to identify a particular toric-code ground state. We use classical fidelity for NNQST, which is an upper bound for quantum fidelity.\\
\textbf{(b)}\emph{(Right)}\textbf{:} Performance of classical shadows for the same problem. We use quantum fidelity for classical shadows.
The shaded regions are the standard deviation of the estimated fidelity over ten runs.
}
    \label{fig:tomoToric}
\end{figure}

In the main text, we have considered direct fidelity estimation for GHZ states and compared it with neural network quantum state tomography (NNQST).
While highly instructive from a theoretical perspective, GHZ states comprised of 100 qubits are very fragile and challenging to implement in practice.
To conduct experiments for more physical target states, we consider \emph{Toric code ground states} \cite{kitaev2002toric}.
Not only are they the most prominent example of a topological quantum error correcting code and thus highly relevant for quantum computing devices. They also correspond to ground states of a Hamiltonian: $H=-\sum_v A_v - \sum_p B_p$, where $A_v$ and $B_p$ denote vertex- and plaquette operators\footnote{$A_v$ is the product of four Pauli-$X$ operators around a vertex $v$, while $B_p$ is the product of four Pauli-$Z$ operators around the plaquette $p$.}.
The ground space of $H$ is four-fold degenerate and we select the superposition of all closed-loop configurations ($| \psi \rangle \propto \sum_{S:\textrm{ closed loop}} |S \rangle$) as a test state for both classical shadows and NNQST: how many measurement repetitions are required to accurately identify this toric code ground state with high fidelity? The results are shown in Supplementary Figure~\ref{fig:tomoToric}.
Neural network tomography based on a deep generative model seems to require a number of samples that scales unfavorably in the system size $n$ (left). In contrast, fidelity estimation with classical shadows is completely independent of the system size.
The difficulty of NNQST in learning 2D toric code may be related to some observed failures of deep learning \cite{shalev2017failures} for learning patterns with combinatorial structures.
In Supplementary Section~\ref{sec:detail}, we provide
further evidence for potential difficulties when using machine learning approaches to reconstruct some simple quantum states due to a well-known computational hardness conjecture.

\subsection{Witnesses for tripartite entanglement}

\begin{figure}[t]
    \centering
    \includegraphics[width=1.0\textwidth]{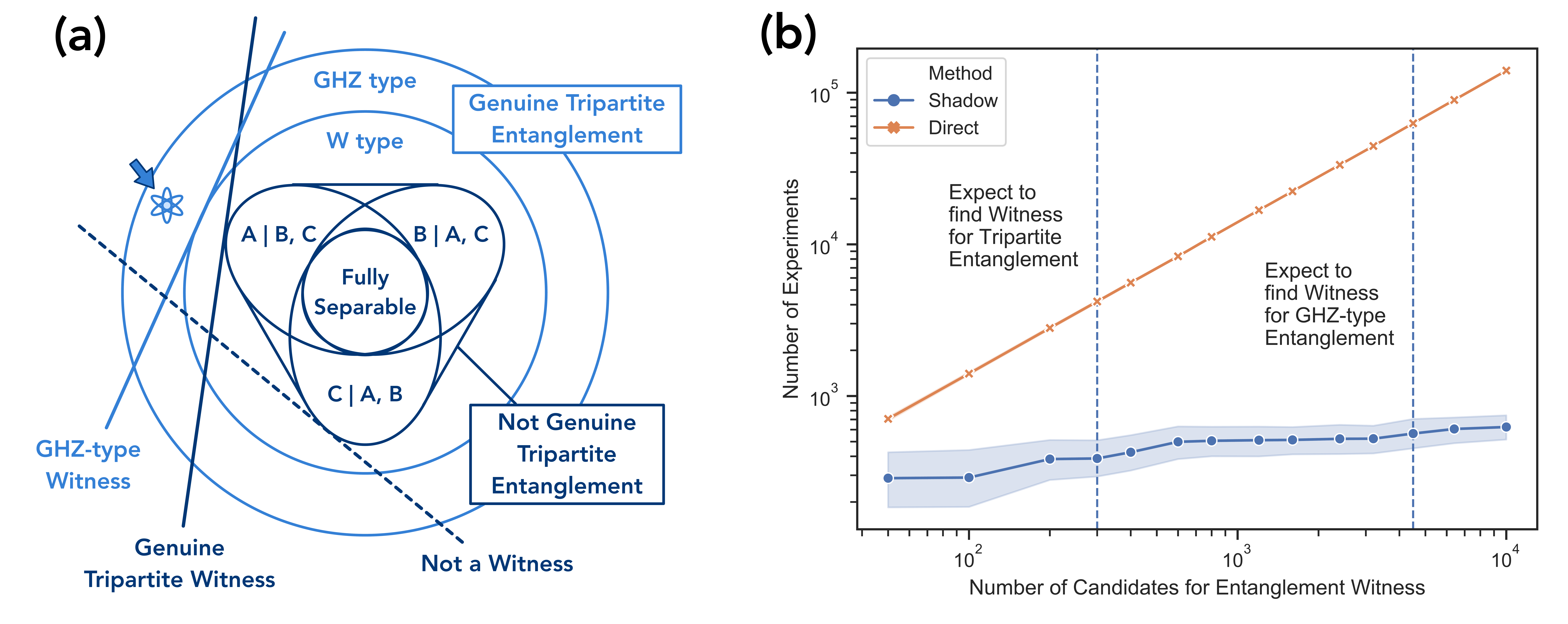}
    \caption{\emph{Detection of GHZ-type entanglement for 3-qubit states.} \\
\textbf{(a)} \emph{(Left)}\textbf{:} Schematic illustration of 3-partite entanglement. Entanglement witnesses are linear functions that separate part of one entanglement class from all other classes. \\
\textbf{(b)} \emph{(Right)}\textbf{:} Number of entanglement witnesses vs.\ number of experiments required to accurately estimate all of them. The dashed lines represent the expected number of (random) entanglement witnesses required to detect genuine three-partite entanglement and GHZ-type entanglement in a randomly rotated GHZ state. The shaded region is the standard deviation of the required number of experiments over ten independent repetitions of the entire setup.}
    \label{fig:witness}
\end{figure}

Entanglement is at the heart of virtually all quantum communication and cryptography protocols and an important resource for quantum technologies in general.
This renders the task of detecting entanglement important both in theory and practice \cite{friis2019entanglement,guehne2009entanglement}. While bipartite entanglement is comparatively well-understood, multi-partite entanglement has a much more involved structure.
Already for $n=3$ qubits, there is a variety of inequivalent entanglement classes. These include fully-separable, as well as bi-separable states, $W$-type states and finally GHZ-type states. The relations between these classes are summarized in Supplementary Figure~\ref{fig:witness} and we refer to \cite{acin2001classification} for a complete characterization. Despite this increased complexity, entanglement witnesses remain a simple and useful tool for testing which class a certain state $\rho$ belongs to.
However, any given entanglement witness only provides a one-sided test -- see Supplementary Figure~\ref{fig:witness} (left) for an illustration -- and it is often necessary to compute multiple witnesses for a definitive answer.

Classical shadows based on random Clifford measurements can considerably speed up this search: according to Theorem~\ref{thm:main} a classical shadow of moderate size allows for checking an entire list of fixed entanglement witnesses simultaneously. Supplementary Figure~\ref{fig:witness} (right) underscores the economic advantage of such an approach over measuring the individual witnesses directly.
Directly measuring $M$ different entanglement witnesses requires a number of  quantum measurements that scales (at least) linearly in $M$. In contrast, classical shadows get by with $\log (M)$-many measurements only.

More concretely, suppose that the state to be tested is a local, random unitary transformation of the GHZ state. Then, this state is genuinely tripartitely entangled and moreover belongs to the GHZ class.
The dashed vertical lines in Supplementary Figure~\ref{fig:witness} (right) denote the expected number of (randomly selected) witnesses required to detect genuine tripartite entanglement (first) and GHZ-type entanglement (later).
From the experiment, we can see that classical shadows achieve these thresholds with an exponentially smaller number of samples than the naive direct method.
Finally, classical shadows are based on random Clifford measurements and do not depend on the structure of the concrete witness in question.
In contrast, direct estimation crucially depends on the concrete witness in question and may be considerably more difficult to implement.

\section{Related work} \label{sec:related-work}

\paragraph{General quantum state tomography}

The task of reconstructing a full classical description --- the density matrix $\rho$ --- of a $d$-dimensional quantum system from experimental data is one of the most fundamental problems in quantum statistics, see e.g.\ \cite{hradil1997tomography,blumekohout2010tomography,gross2010compressed,gross2013focus} and references therein. Sample-optimal protocols, i.e.\ estimation techniques that get by with a minimal number of measurement repetitions, have only been developed recently.
Information-theoretic bounds assert that of order $\mathrm{rank}(\rho)d$ state copies are necessary to fully reconstruct $\rho$  \cite{haah2017sample}.
Constructive protocols \cite{wright2016tomography,haah2017sample} saturate this bound, but require entangled circuits and measurements that act on all state copies simultaneously. More tractable single-copy measurement procedures require of order $\mathrm{rank}(\rho)^2 d$ measurements \cite{haah2017sample}. This more stringent bound is saturated by low rank matrix recovery \cite{flammia2012quantum,kueng2017lowrank,kueng2016low} and projected least squares estimation \cite{guta2018fast,sugiyama2013tomography}.

These results highlight an exponential bottleneck for tomography protocols that work in full generality: at least $d=2^n$ copies of an unknown $n$-qubit state are necessary. This exponential scaling extends to the computational cost associated with storing and processing the measurement data.

\paragraph{Matrix product state tomography}

Restricting attention to highly structured subsets of quantum states sometimes allows for overcoming the exponential bottleneck that plagues general tomography.
Matrix product state (MPS) tomography \cite{cramer2010efficient} is the most prominent example for such an approach. It only requires a polynomial number of samples, provided that the underlying quantum state is well approximated by a MPS with low bond dimension. In quantum many-body physics this assumption is often justifiable \cite{lanyon2017efficient}. However, MPS representations of general states have exponentially large bond dimension. In this case, MPS tomography offers no advantage over general tomography.
Similar ideas could also be extended to multi-scale entangled states (MERA) tomography \cite{landon2012practical}.

\paragraph{Neural network tomography}

Recently, machine learning has also been applied to the problem of predicting features of a quantum systems.
These approaches construct a classical representation of the quantum system by means of a deep neural network that is trained by feeding in quantum measurement outcomes.
Compared to MPS tomography, neural network tomography may be more broadly applicable \cite{gao2017efficient, torlai2018neural, carrasquilla2019reconstructing}.
However, the actual class of systems that can be efficiently represented, reconstructed and manipulated is still not well understood.

\paragraph{Compressed classical description of quantum states}

To circumvent the exponential scaling in representing quantum states, Gosset and Smolin \cite{gosset2018compressed} have proposed a stabilizer sketching approach that compresses a classical description of quantum states to an accurate sketch of subexponential size.
This approach bears some similarity with classical shadows based on random Clifford measurements.
However, stabilizer sketching requires a fully-characterized classical description of the state as an input.
So, it still suffers from an exponential scaling in the resources used in practice.
Recently, Paini and Kalev \cite{paini2019approximate} have proposed an approximate classical description of a quantum state that can estimate the expectation value of an observable from Haar-random single-qubit rotations followed by computational basis measurements.
They focus on estimating a single observable,
while we focus on estimating many observables simultaneously.
In our classical shadow approach, the Haar-random single-qubit rotations \cite{paini2019approximate} are replaced by
random single-qubit Clifford rotations, or -- equivalently -- measuring each qubit in a random Pauli basis.
This simplification may be viewed as a partial derandomization and works, because the (single-qubit) Clifford group forms a 3-design \cite{kueng2015qubit,zhu2017clifford,webb2015clifford}.

\paragraph{Direct fidelity estimation}

Direct fidelity estimation is a procedure that allows for predicting a single pure target fidelity $\langle \psi| \rho | \psi \rangle$ up to accuracy $\epsilon$.
The best-known technique is based on few Pauli measurements that are selected randomly using importance sampling \cite{flammia2011direct, da2011practical}.
The required number of samples depends on the target: it can range from a dimension-independent order of $1/\epsilon^2$
(if $| \psi \rangle$ is a stabilizer state) to roughly $2^n/\epsilon^4$ in the worst case.

\paragraph{Efficient estimation of local observables}

In quantum many-body physics, many interesting observables can be decomposed into local constituents. This renders the task of accurately predicting many local observables very important --- both in theory and practice.
A series of recent works
 \cite{cotler2019quantum, bonet2019nearly, jiang2019optimal, evans2019scalable} propose different measurement strategies to measure many local observables simultaneously.
All of them focus on estimating $k$-local Pauli observables up to accuracy $\epsilon$. This would directly translate  to an approximation error $2^k \epsilon$ for general $k$-local observables.
For some measurement schemes, this general error bound seems unavoidable. But, for certain strategies a careful analysis could lead to an improved performance.
The two works \cite{cotler2019quantum, bonet2019nearly}
are based on properly analyzing the commutation relations between the $k$-local Pauli observables of interest. Subsequently, one can group commuting observables together and measure them all at once.
Different from this more standardized strategy, \cite{jiang2019optimal} uses entangled Bell-basis measurements, and \cite{evans2019scalable} is based on randomized measurements to efficiently measure local observables.
The prior earlier works \cite{cotler2019quantum, bonet2019nearly} have worse performance compared to the more recent two \cite{jiang2019optimal, evans2019scalable}.
While the latter two procedures are seemingly different from prediction with classical shadows (Pauli measurements), the sample complexities associated with all three approaches are comparable.
Derandomizing classical shadows, however, could considerably reduce the number of measurements required. We will address such a substantial and practical improvement in upcoming work.

\paragraph{Shadow tomography}

Shadow tomography aims at simultaneously estimating the outcome probabilities associated with $M$ 2-outcome measurements up to accuaracy $\epsilon$: $p_i (\rho) = \mathrm{tr}(E_i \rho)$, where each $E_i$ is a positive semidefinite matrix with operator norm at most one \cite{aaronson2018shadow,brandao2019sdp,aaronson2019gentle}.
This may be viewed as a generalization of fidelity estimation.
The best existing result is due to Aaronson and Rothblum \cite{aaronson2019gentle}. They showed that
$
N= \tilde{\mathcal{O}} \left(\log (M)^2 \log (d)^2/\epsilon^8 \right)
$ copies of the unknown state suffice to achieve this task \footnote{The scaling symbol $\tilde{\mathcal{O}}$ suppresses logarithmic expressions in other problem-specific parameters.}.
Broadly speaking, their protocol is based on performing gentle 2-outcome measurements one-by-one and subsequently (partially) reversing the damage to the quantum state caused by the measurement. This task is achieved by explicit quantum circuits of exponential size that act on all copies of the unknown state simultaneously.
This rather intricate procedure bypasses the no-go result advertised in Theorem~\ref{thm:main2}
and results in a sampling rate that is independent of the 2-outcome measurements in question --- only their cardinality $M$ matters.

\section{Details regarding numerical experiments}
\label{sec:detail}

\subsection{Predicting quantum fidelities}

This numerical experiment considers classical shadows based on random Clifford measurements.
We exploit the Gottesman-Knill theorem for efficient classical computations. This well-known result states that Clifford circuits can be simulated efficiently on classical computers; see also \cite{aaronson2004improved} for an improved classical algorithm. This has allowed us to address rather large system sizes (more than 160 qubits).
To test the performance of feature prediction with classical shadows
we first have to simulate the (quantum) data acquisition phase. We do this by repeatedly executing the following (efficient) protocol:
\begin{enumerate}
    \item Sample a Clifford unitary $U$ from the Clifford group using the algorithm proposed in \cite{koenig2014efficiently}. This Clifford unitary is parameterized by $(\alpha, \beta, \gamma, \delta, r, s)$ which fully characterize its action on Pauli operators:
    \begin{equation}
    UP^X_j U^\dagger = (-1)^{r_j} \Pi_{i=1}^n (P^X_i)^{\alpha_{ji}} (P^Z_i)^{\beta_{ji}} \quad \textrm{and} \quad
    U P^Z_j U^\dagger = (-1)^{s_j} \Pi_{i=1}^n (P^X_i)^{\gamma_{ji}} (P^Z_i)^{\delta_{ji}}
    \end{equation}
    for all $j=1,\ldots,n$. Here,
     $P^X_j, P^Z_j$ are the Pauli $X$, $Z$-operators acting on the $j$-th qubit, and $\alpha_{ji}, \beta_{ji}, \gamma_{ji}, \delta_{ji}, r_j, s_j \in \{0, 1\}$.
    \item
    Given a unitary $U$ parameterized by $(\alpha, \beta, \gamma, \delta, r, s)$, we can apply $U$ on any stabilizer state by changing the stabilizer generators and the destabilizers as defined in \cite{aaronson2004improved}.
    \item A computational basis measurement can be simulated using the standard algorithm provided in \cite{aaronson2004improved}.
\end{enumerate}

Although originally designed for pure target states $|\psi_i \rangle \! \langle \psi_i|$, we can readily extend this strategy to mixed states $\rho=\sum_i p_i | \psi_i \rangle \! \langle \psi_i|$.
Operationally speaking, mixed states  arise from sampling from a pure state ensemble. This mixing process can be simulated efficiently on classical machines.

\begin{figure}[t]
    \centering
    \includegraphics[width=0.84\textwidth]{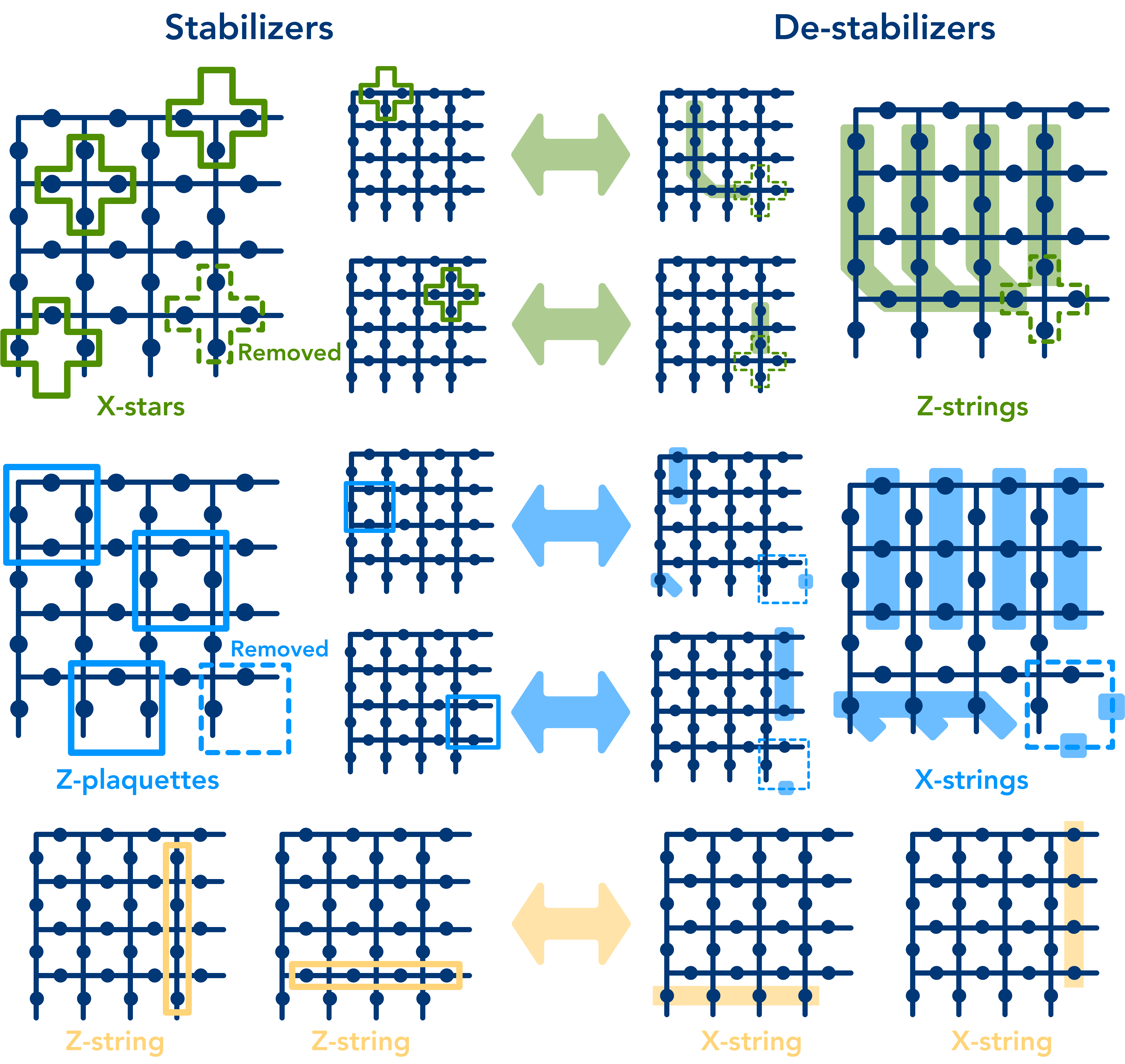}
    \caption{Stabilizers and de-stabilizers of the toric code that encodes $\ket{00}$.}
    \label{fig:tableau}
\end{figure}

For neural network quantum state tomography, we use the open-source code provided by the authors \cite{carrasquilla2019reconstructing}.
The main challenge is generating training data, i.e.\ simulating measurement outcomes.
For pure and noisy GHZ states, we use the tetrahedral POVM \cite{carrasquilla2019reconstructing}.
For the toric code ground state, we use the Psi2 POVM (which is a measurement in the computational ($Z$-) basis).
Note that measuring in the $Z$-basis is not a tomographically complete measurement, but we found machine learning models to perform better using Psi2. This is possibly because the pattern is much more obvious (closed-loop configurations) and the figure of merit used in NNQST is a classical fidelity.

A concrete algorithm for creating training data for pure GHZ states is included in the aforementioned open-source implementation of \cite{carrasquilla2019reconstructing}. It uses matrix product states to simulate quantum measurements efficiently.
The training data for noisy GHZ states is a slight modification of the existing code. With probability $1-p$, we sample a measurement outcome from the original
state $|\psi_{\mathrm{GHZ}}^+ \rangle = \tfrac{1}{\sqrt{2}} (|0 \rangle^{\otimes n} + |1 \rangle^{\otimes n})$.
And with probability $p$, we sample
a measurement outcome from
 $| \psi_{\mathrm{GHZ}}^-\rangle=\frac{1}{\sqrt{2}} (\ket{0}^{\otimes n} - \ket{1}^{\otimes n})$ (phase error).
Since the figure of merit is the fidelity with the pure GHZ state in both pure and noisy GHZ experiment, we reuse the implementation provided in \cite{carrasquilla2019reconstructing}.

Creating training data for toric code is somewhat more involved.
The goal is to sample a closed-loop configuration on a 2D torus uniformly at random.
This can again be done using classical simulations of stabilizer states \cite{aaronson2004improved}.
The main technical detail is to create a tableau that contains both the stabilizer and the de-stabilizer for the state in question. The rich structure of the toric code renders this task rather easy.
The stabilizers are the $X$-stars and the $Z$-plaquettes, with two $Z$-strings over the two loops of the torus.
The de-stabilizer of each stabilizer is a Pauli-string that anticommutes with the stabilizer, but commutes with other stabilizers and other de-stabilizers.
The full set of stabilizers and de-stabilizers for the toric code can be seen in Supplementary Figure~\ref{fig:tableau}.

\subsection{Potential obstacles for learning certain quantum states} \label{sec:ml-bottleneck}

In our numerical studies, we have seen that neural network quantum state tomography based on deep generative models seems to have difficulty learning toric code ground states.

Here, we take a closer look at this curious aspect and construct a simple class of quantum states where efficient learning of the quantum state from the measurement data would violate a well-known computational hardness conjecture.
First of all, each computational ($Z$-) basis measurement of the toric code produces a random bit-string. Most bits are sampled uniformly at random from $\{0, 1\}$ and the remaining bits are binary functions that only depend on these random bits.
Consider a simple class of quantum states that
mimic this property.  Given $a \in \{0, 1\}^{n-1}$ and $f_a(x) = \sum_i a_i x_i$ (mod $2$), we define $\ket{a} = \frac{1}{\sqrt{2^{n-1}}}\sum_{x \in \{0, 1\}^{n-1}} \ket{x} \otimes \ket{f_a(x)}$. Such states can be created by preparing $\ket{+}$ on the first $n-1$ qubits, $\ket{0}$ on the $n$-th qubit followed by CNOT gates between $i$-th qubit and $n$-th qubit for every $a_i = 1$.
Measuring $\ket{a}$ in the computational ($Z$-) basis is equivalent to sampling the first $n-1$ bits $x$ uniformly at random. The final bit is characterized by the deterministic formula $f_a(x)$.
Now, consider a (globally) depolarized version of this pure state:
\begin{equation}
\rho_a = \mathcal{D}_{\eta}(|a \rangle \! \langle a|)= (1 - \eta) |a \rangle \! \langle a| + \tfrac{\eta}{2^n} \mathbb{I}^{\otimes n} \quad \text{
for some $\eta \in (0,1)$.}
\end{equation}
One of the most widely used conjectures for building post-quantum cryptography is the hardness of learning with error (LWE) \cite{regev2009lattices}.
LWE considers the task of learning a linear $n$-ary function $f$ over a finite ring from noisy data samples $(x, f(x)+\eta)$, where $x$ is sampled uniformly at random and $\eta$ is some independent error.
An efficient learning algorithm for LWE will be able to break many post-quantum cryptographic protocals that are believed to be hard even for quantum computers.
The simplest example of LWE is called learning parity with error, where $f(x) = \sum_i a_i x_i$ (mod $2$) for $x \in \{0, 1\}^n$ and some unknown $a \in \{0, 1\}^n$. Learning parity with error is also conjectured to be computationally hard \cite{blum2003noise}.
Since learning $|a\rangle$ from computational ($Z$-) basis measurements on $\rho_a$ is equivalent to learning parity with error, it is unlikely there will be a neural network approach that can learn $\rho_a$ efficiently.

\subsection{Predicting witnesses for tripartite entanglement}

This numerical experiment considers classical shadows based on random Clifford measurements.
The numerical studies regarding entanglement witnesses are based locally rotated 3-qubit ($n=3$) GHZ states:
\begin{equation}
\ket{\psi} = U_A \otimes U_B \otimes U_C  | \psi_{\mathrm{GHZ}}^+ \rangle
\quad \text{where $U_A,U_B,U_C$ are random single-qubit rotations.}
\end{equation}
For $\rho=|\psi \rangle \! \langle \psi|$, we hope to verify the tripartite entanglement present in the system.
To this end, we consider a simple family of entanglement witnesses with compatible structure:
\begin{equation}
O:= O(V_A,V_B,V_C) = V_A \otimes V_B \otimes V_C | \psi_{\mathrm{GHZ}}^+ \rangle \! \langle \psi_{\mathrm{GHZ}}^+| V_A^\dagger \otimes V_B^\dagger \otimes V_C^\dagger.
\label{eq:witness-candidates}
\end{equation}
The single-qubit unitaries $ V_A, V_B, V_C$ parametrize different witnesses.

A complete characterization of entanglement in three-qubit systems can be found in Supplementary Figure~\ref{fig:witness}. The expectation value of an entanglement witness $O(V_A, V_B,V_C)$ in the tripartite state $\rho$ can certify that $\rho$ belongs to a particular entanglement class. For example, it is known from the analysis in \cite{acin2001classification} that for any state $\rho_s$ with only bipartite entanglement, $\Tr\left(O\rho_s\right) \le .5$, while for any state $\rho_s$ with at most W-type entanglement, $\Tr\left(O\rho_s\right) \le .75$. Therefore verifying that $\Tr\left(O\rho\right) > .5$ certifies that $\rho$ has tripartite entanglement, while $\Tr\left(O\rho\right) > .75$ certifies that $\rho$ has GHZ-type entanglement.

After choosing random unitaries $U_A, U_B, U_C$ to specify the GHZ-type state $|\psi\rangle$, we generate a list of random $V_A, V_B, V_C$ to specify a set of potential entanglement witnesses for $|\psi\rangle$:
\begin{equation}
O_1 = O(V_{A, 1}, V_{B, 1}, V_{C, 1}), \ldots, O_M = O(V_{A, M}, V_{B, M}, V_{C, M}).
\end{equation}
If the randomly generated $O_i = O(V_{A, i}, V_{B, i}, V_{C, i})$ satisfies $\Tr(O_i \ket{\psi} \! \bra{\psi}) > 0.5$, then $O_i$ is an entanglement witness for genuine tripartite entanglement, and if $\Tr(O_i \ket{\psi} \! \bra{\psi}) > 0.75$, then $O_i$ is a witness for GHZ-type entanglement.
We can compute the expected number of random candidates we have to test to find an observable $O$ such that $\Tr(O \ket{\psi} \! \bra{\psi}) > 0.5$ or $\Tr(O \ket{\psi} \! \bra{\psi}) > 0.75$; these numbers are indicated  as the dashed lines on the right side of  Supplementary Figure~\ref{fig:witness}.

Given the list of randomly generated witness candidates $O_1, \ldots, O_M$, we would like to predict $\Tr(O_i | \psi \rangle \! \langle \psi|)$ for all $1 \leq i \leq M$.
The naive approach is to directly measure all observables (witnesses).
We refer to this as the direct measurement approach.
For this approach, we consider the number of total experiments required to estimate every $\Tr(O_i |\psi \rangle \!\langle \psi|)$ up to an error $0.1$.
Note that the number of required samples may vary from witness to witness --- it depends on the variance associated with the estimation.
In the worst case, one would need $\approx 100$ measurements for each witness candidate.

Instead of this direct measurement approach, one could use classical shadows (Clifford measurements) to predict \textit{all} the observables (witnesses) $O_1, \ldots, O_M$ at once.
Because, $\mathrm{tr}(O_i^2)=1$ for al $1 \leq i \leq M$, the shadow norm obeys $\|O_i\|^2_{\mathrm{shadow}} \leq 3 \Tr\left(O_i^2\right)=3$, according to the analysis in Supplementary Section \ref{sec:linearfunc-general}. Hence Theorem~\ref{thm:main} shows that classical shadows can predict the expectation values of many candidate witnesses very efficiently.

In the numerical experiment, we gradually increased the number of random Clifford measurements we use to construct classical shadows until the classical shadows could accurately predict all $\Tr(O_i \ket{\psi} \! \bra{\psi})$ up to $0.1$-error. The results are shown in Supplementary Figure~\ref{fig:witness}. Because the system size is small ($n=3$ qubits), we simulate the quantum experiments classically by storing and processing all $2^3=8$ amplitudes. In practice, one should use statistics, like sample variance estimation or the bootstrap \cite{efron1994introduction}, to
determine confidence intervals and a posteriori guarantees. Quadratic function prediction with classical shadows (Clifford measurements) can be used to achieve this goal efficiently.

\subsection{Predicting two-point correlation functions}

Predicting two-point correlation function could be done efficiently using classical shadows based on random Pauli measurements.
To facilitate direct comparison, this numerical experiment is designed to reproduce one of the core examples in  in \cite{carrasquilla2019reconstructing}.
In particular, we use the same data, downloaded from \url{https://github.com/carrasqu/POVM_GENMODEL}.
The classical shadow (based on random Pauli basis measurements) replaces the original machine learning based approach for predicting local observables.
We use multi-core CPU for training and making prediction with the machine learning model.
The reported time is the total CPU time.
Predicting local observables $O$ using the (Pauli) classical shadow can be done efficiently by creating the reduced density matrix $\rho_A$, where $A$ is the subsystem $O$ acts on.
The reduced density matrix $\rho_A$ can be created by simply neglecting the data for the rest of the system.
Importantly, $\mathcal{M}^{-1}(U^\dagger |\hat{b} \rangle \! \langle \hat{b}|U)$ is never created as an $2^n \times 2^n$ matrix.
Taking the inner product of $\rho_A$ with the local observables $O$ yields the desired result.

\subsection{Predicting subsystem R\'enyi entanglement entropies}

We consider classical shadows based on random Pauli measurements for predicting subsystem entanglement entropies.
In the first part of the experiment, we consider the ground state of a disordered Heisenberg model. The associated Hamiltonian is $H = \sum_{i} J_i \langle S_i \cdot S_{i+1} \rangle$, where each $J_i$ is sampled uniformly (and independently) from the unit interval $[0, 1]$.
The approximate ground state is found by implementing the recursive procedure from \cite{refael2013strong}:
identify the largest $J_i$, forming singlet for the connected sites, and reduce the system by removing $J_i$. We refer to \cite{refael2013strong} for details.
In the experiment, we perform single-shot random Pauli basis measurements on the approximate ground state. I.e.\ we measure the state in a random Pauli basis only once and then choose a new random basis.
However, in physical experiments, it is often easier to
repeat a single Pauli basis measurement many times before re-calibrating to measure another Pauli basis.
Performing a single random basis measurement for many repetitions can be beneficial experimentally compared to measuring a random basis every single time.
Classical shadows (Pauli) are flexible enough to incorporate economic measurement strategies that take this discrepancy into account.
We refer to the open source implementation in \url{https://github.com/momohuang/predicting-quantum-properties} for the exact details.

To obtain a reasonable benchmark, we compare this procedure with the approach proposed by Brydges \textit{et al.}~\cite{brydges2019probing}.
For a subsystem $A$ comprised of $k$ qubits, the approach proposed in \cite{brydges2019probing} for predicting the R\'enyi entropy works as follows.
First, one samples a random single-qubit unitary rotations independently for all $k$ qubits. Then, one applies the single-qubit unitary rotation to the system and measures the system in the computational basis to obtain a string of binary values $s \in \{0, 1\}^k$. For each random unitary rotation, several repetitions are performed. The precise number of repetitions for a single random basis is a hyper-parameter that has to be optimized.
The estimator for the R\'enyi entropy takes the following form:
\begin{equation}
\Tr(\rho_A^2) =
2^k \sum_{s, s' \in \{0, 1\}^k} (-2)^{-H(s, s')} \overline{P(s) P(s')}.
\label{eq:brydges-probing-formula}
\end{equation}
The function $H(s, s')$ is the Hamming distance between strings $s$ and $s'$ (i.e,\ the number of positions at which individual bits are different), while $P(s)$ and $P(s')$ are the probabilities for measuring $\rho$ and obtaining the outcomes $s$ and $s'$, respectively.
The probability $P(s)$ is a function that depends on the randomly sampled single-qubit rotation.
$\overline{P(s) P(s')}$ is the expectation of $P(s) P(s')$ averaged over the random single-qubit rotations.

The random single-qubit rotations could be taken as single-qubit Haar-random rotations or single-qubit random Clifford rotations.
The latter choice is equivalent to random Pauli measurements -- the measurement primitive we consider for classical shadows also.
For the test cases we considered, using random Pauli measurements yields similar (and sometimes improved) performance compared to single-qubit Haar-random unitary rotation.
This allows the approach by \cite{brydges2019probing} and the procedure based on classical shadows to be compared on the same ground.
We follow the strategy in \cite{brydges2019probing} to estimate the formula in Eq.~\eqref{eq:brydges-probing-formula}.
First, we sample $N_U$ random unitary rotations.
For each random unitary rotation, we perform $N_M$ repetitions of rotating the system and measuring in the computational basis.
The $N_M$ measurement outcomes allow us to construct an empirical distribution for $P(s)$.
Thus we could use the $N_M$ measurement outcomes to estimate $2^k \sum_{s, s' \in \{0, 1\}^k} (-2)^{-H(s, s')} P(s) P(s')$ for a single random unitary rotation.
We then take the average over $N_U$ different random unitary rotations.
Choosing a suitable parameter for $N_U$ and $N_M$ is nontrivial.
We employ the strategy advocated in \cite{brydges2019probing} for finding the best parameter for $N_U$ and $N_M$. This strategy is called grid search and is performed by trying many different choices for $N_U, N_M$ and recording the best one.

\subsection{Variational quantum simulation of the lattice Schwinger model}

The application for variational quantum simulation uses classical shadows based on random Pauli measurements which is designed to predict a large number of local observables efficiently.
It is based on the seminal work presented in \cite{kokail2019self}.
After a Kogut-Susskind encoding to map fermionic configurations to a spin-$1/2$ lattice with an even number $N$ of lattice sites and a subsequent Jordan-Wigner transform, the Hamiltonian becomes
\begin{equation}
\hat H = \underbrace{\frac{w}{2} \sum_{j=1}^{N-1} P^X_j P^X_{j+1}}_{\hat{\Lambda}_X} + \underbrace{\frac{w}{2} \sum_{j=1}^{N-1} P^Y_j P^Y_{j+1}}_{\hat{\Lambda}_Y} + \underbrace{\sum_{j=1}^N d_j P_j^z + \sum_{j=1}^{N-2} \sum_{j'=j+1}^{N-1} c_{j, j'} P_j^z P_{j'}^z}_{\hat{\Lambda}_Z}.
\label{eq:hamiltonian-schwinger}
\end{equation}
Here, $P^X_j, P^Y_j, P^Z_j$ denote Pauli-$X, Y, Z$ operators acting on the $j$-th qubit ($1\leq j \leq N$).
This Hamiltonian has very advantageous structure.
Each of the three contributions can be estimated by performing a single Pauli basis measurement (measure every qubit in the $X$ basis to determine $\hat{\Lambda}_X$, measure every qubit in the $Y$ basis to determine $\hat{\Lambda}_Y$ and measure every qubit in the $Z$ basis to determine $\hat{\Lambda}_Z$).
The measurement of the Hamiltonian variance $\langle \hat H^2 \rangle - \langle \hat H \rangle^2$ is more complicated, because  $\langle \hat H^2 \rangle$ does not decompose nicely.
To determine its value, we must first measure
$\hat{\Lambda}_X^2$, $\hat{\Lambda}_Y^2$ and $\hat{\Lambda}_Z^2$. This is the easy part, because 3 measurement bases once more suffice.
However, in addition, we must also estimate the
 anti-commutators $\{\hat{\Lambda}_X, \hat{\Lambda}_Y\}, \{\hat{\Lambda}_X, \hat{\Lambda}_Z\}, \{\hat{\Lambda}_Y, \hat{\Lambda}_Z\}$.
This may be achieved by measuring the following $k$-local observables (with $k$ at most $4$):
\begin{align}
    \{\hat{\Lambda}_X, \hat{\Lambda}_Y\}: \,\, & P^X_j P^X_{j+1} P^Y_{j'} P^Y_{j'+1}, &  \forall j, j' \in \{1, N-1\}, \mbox{ s.t. } j \neq j', j \neq j'+1, j+1 \neq j', \nonumber\\
    \{\hat{\Lambda}_X, \hat{\Lambda}_Z\}: \,\, & P^X_j P^X_{j+1} P^Z_{j'} P^Z_{j''}, &  \forall j, j', j'' \in \{1, N-1\}, \mbox{ s.t. } j \neq j', j \neq j'', j+1 \neq j', j+1 \neq j'', j' < j'', \nonumber \\
    \{\hat{\Lambda}_X, \hat{\Lambda}_Z\}: \,\, & P^X_j P^X_{j+1} P^Z_{j'}, &  \forall j, j' \in \{1, N-1\}, \mbox{ s.t. } j \neq j', j+1 \neq j', \label{eq:setoflocal} \\
    \{\hat{\Lambda}_Y, \hat{\Lambda}_Z\}: \,\, & P^Y_j P^Y_{j+1} P^Z_{j'} P^Z_{j''}, & \forall j, j', j'' \in \{1, N-1\}, \mbox{ s.t. } j \neq j', j \neq j'', j+1 \neq j', j+1 \neq j'', j' < j'', \nonumber \\
    \{\hat{\Lambda}_Y, \hat{\Lambda}_Z\}: \,\, & P^Y_j P^Y_{j+1} P^Z_{j'}, &  \forall j, j' \in \{1, N-1\}, \mbox{ s.t. } j \neq j', j+1 \neq j', \nonumber
\end{align}
Although local, estimating all observables of this form is the main bottleneck of the entire procedure.
To minimize the number of measurement bases,
the original work \cite{kokail2019self} has performed an analysis of symmetry in the lattice Schwinger model.
First, the target Hamiltonian in Equation~\eqref{eq:hamiltonian-schwinger} satisfies $[\hat{H}, \sum_i P^Z_i] = 0$, which corresponds to a charge conservation symmetry in the scalar fermionic field.
\cite{kokail2019self} further consider a charge symmetry subspace with $\sum_i P^Z_i = 0$, which corresponds to a $\hat{CP}$ symmetry.
In this subspace, we have $\langle \{\hat{\Lambda}_X, \hat{\Lambda}_Z\} \rangle = \langle \{\hat{\Lambda}_Y, \hat{\Lambda}_Z\} \rangle$.
This ensures that we only have to estimate local observables corresponding to $\{\hat{\Lambda}_X, \hat{\Lambda}_Y\}$ and $\{\hat{\Lambda}_X, \hat{\Lambda}_Z\}$.
In the original setup \cite{kokail2019self}, this task was achieved by measuring roughly $2N$ bases in total.
We refer to \cite[Appendix~B and Appendix~C]{kokail2019self}
for further details and explanation.
We propose to replace this original approach by linear feature prediction with classical shadows (Pauli measurements).

For classical shadows based on random Pauli measurements, every measurement basis is an independent random $X$, $Y$, or $Z$ measurement for every qubit.
This randomized general purpose procedure does not take into account the fact that we want to measure a specific set of $k$-local observables given in Equation~\eqref{eq:setoflocal}.
The derandomized version of classical shadows is based on the concept of pessimistic estimators \cite{raghavan1988pessimistic,spencer1994lectures} (see also \cite{wigderson2008derandomization} for an application with quantum information context).
It
removes the original randomness by utilizing the knowledge of this specific set of $k$-local observables.
When we throw a dice (or coin) to decide whether we want to measure in either, the $X-$, the $Y-$, or the $Z-$basis, the derandomized version would choose the measurement basis ($X$, $Y$, or $Z$) that would lead to the best expected performance on the set of $k$-local observables given in Equation~\eqref{eq:setoflocal}.
The expected performance is computed based on random Pauli basis measurements and the analysis in Supplementary Section~\ref{sec:generalshadow}.
The derandomized version of classical shadows would perform at least as well as the original randomized version.
Furthermore, due to the dependence on the specific set of observables for choosing the measurement bases, the derandomized version can exploit advantageous structures in the set of observables we want to measure.
As detailed in the main text, classical shadows based on random Pauli measurements provide improvement only for larger system sizes (more than $50$ qubits).
A derandomized version of classical shadows improves upon the randomized version and leads to a substantial improvement in efficiency and scalability over a wide range of system sizes.
As an added benefit, derandomization can be completely automated and does not depend on the concrete set of target observables. We refer to \url{https://github.com/momohuang/predicting-quantum-properties}
for a (roughly linear time) algorithm that derandomizes random Pauli measurements for any collection of target observables with Pauli structure.

\section{Additional computations and proofs for predicting linear functions}

\subsection{Background: Clifford circuits and the stabilizer formalism}

Clifford circuits were introduced by Gottesman \cite{gottesman1997stabilizer} and form an indispensable tool in  quantum information processing. Applications range from quantum error correction \cite{nielsen2000quantum}, to measurement-based quantum computation \cite{raussendorf2001measurementbased,briegel2009measurementbased} and randomized benchmarking \cite{emerson2005randomized,knill2008randomized,magesan2011randomized}.
For systems comprised of $n$ qubits, the Clifford group is generated by CNOT, Hadamard and phase gates. This results in a finite group of cardinality $2^{\mathcal{O}(n^2)}$ that maps (tensor products of) Pauli matrices to Pauli matrices upon conjugation. This underlying structure allows for efficiently storing and simulating Clifford circuits on classical computers -- a result commonly known as Gottesman-Knill theorem.
The $n$-qubit Clifford group $\mathrm{Cl}(2^n)$ also comprises a \emph{unitary 3-design} \cite{webb2015clifford,zhu2017clifford,kueng2015qubit}. Sampling Clifford circuits uniformly at random reproduces the first 3 moments of the full unitary group endowed with the Haar measure. For $k=1,2,3$
\begin{align}
\mathbb{E}_{U \sim \mathrm{Cl}(2^n)} \left( UX U^\dagger \right)^{\otimes k}=\int_{U(d)} (U A U^\dagger )^{\otimes k} \mathrm{d} \mu_{\mathrm{Haar}}(U) \quad \text{for all $2^n \times 2^n$ matrices $A$.} \label{eq:k-design}
\end{align}
The right hand side of this equation can be evaluated explicitly by using techniques from representation theory, see e.g.\ \cite[Sec.~3.5]{gross2015partial}.
This in turn yields closed-form expressions for Clifford averages of linear and quadratic operator-valued functions.
Choose a unit vector $x \in
\mathbb{C}^{2^n}$ and let $\mathbb{H}_{2^n}$ denote the space of Hermitian $2^n \times 2^n$ matrices. Then,
\begin{align}
\mathbb{E}_{U \sim \mathrm{Cl}(2^n)}
U^\dagger |x \rangle \! \langle x| U^\dagger\langle x| U A U^\dagger |x \rangle
=& \frac{A+\mathrm{tr}(A)\mathbb{I}}{(2^n+1)2^n} = \frac{1}{2^n}\mathcal{D}_{1/(2^n+1)}(A)
\; \; \textrm{for $A \in \mathbb{H}_{2^n}$,} \label{eq:second-moment}\\
\mathbb{E}_{U \sim \mathrm{Cl}(2^n)}
U^\dagger |x \rangle \! \langle x| U
\langle x| U B_0 U^\dagger |x \rangle \langle x| U C_0 U^\dagger |x \rangle  =&
\frac{\mathrm{tr}(B_0C_0)\mathbb{I} +B_0C_0 +C_0B_0}{(2^n+2)(2^n+1)2^n}
\;\;\text{for $B_0,C_0\in \mathbb{H}_{2^n}$ traceless.} \label{eq:third-moment}
\end{align}
Here, $\mathcal{D}_{p}(A)=pA + (1-p) \frac{\mathrm{tr}(A)}{2^n} \mathbb{I}$ denotes a $n$-qubit depolarizing channel with loss parameter $p$.
Linear maps of this form can be readily inverted. In particular,
\begin{equation}\label{eq:inverse-depol}
\mathcal{D}^{-1}_{1/(2^n+1)}(A) = (2^n+1)A-\mathrm{tr}(A) \mathbb{I} \quad \text{for any $A \in \mathbb{H}_{2^n}$.}
\end{equation}
These closed-form expressions allow us to develop very concrete strategies and rigorous bounds for classical shadows based on (global and local) Clifford circuits.

\subsection{Performance bound for classical shadows based on random Clifford measurements} \label{sub:clifford-shadow-details}

\begin{proposition} \label{prop:Clifford-variance}
Adopt a ``random Clifford basis'' measurement primitive, i.e. each rotation $\rho \mapsto U \rho U^\dagger$ is chosen uniformly from the $n$ qubit Clifford group $\mathrm{Cl}(2^n)$. Then, the associated classical shadow is
\begin{equation}
\hat{\rho} = (2^n+1) U^\dagger |\hat{b} \rangle \! \langle \hat{b}| U - \mathbb{I},
\end{equation}
where $\hat{b} \in \left\{0,1\right\}^n$ is the observed computational basis measurement outcome (of the rotated state $U \rho U^\dagger$).
Moreover, the norm defined in Eq.~\eqref{eq:shadow-norm} is closely related to the Hilbert-Schmidt norm:
\begin{equation}\label{eq:shadow-norm-bound}
\mathrm{tr} \left( O_0^2 \right) \leq \| O_0 \|_{\mathrm{shadow}}^2
\leq 3 \mathrm{tr} \left( O_0^2 \right)
\quad \text{for any traceless $O_0 \in \mathbb{H}_{2^n}$.}
\end{equation}
\end{proposition}

\noindent
Note that passing from $O$ to its traceless part $O_0 = O - \tfrac{\mathrm{tr}(O)}{2^n}\mathbb{I}$ is a contraction in Hilbert-Schmidt norm:
\begin{equation}
\mathrm{tr}\left( O_0^2 \right) = \mathrm{tr}(O^2) - \frac{\mathrm{tr}(O)^2}{2^n}\leq \mathrm{tr}(O^2).
\end{equation}
Hence, we can safely replace the upper bound in Eq.~\eqref{eq:shadow-norm-bound} by $3 \mathrm{tr}(O^2)$ --- the Hilbert Schmidt norm (squared) of the original observable.

\begin{proof}
Eq.~(\ref{eq:second-moment}) readily provides a closed-form expression for the measurement channel defined in Eq.~\eqref{eq:measurement-channel}:
\begin{equation}
\mathcal{M}(\rho) = \sum_{b \in \left\{0,1\right\}^n}\mathbb{E}_{U \sim \mathrm{Cl}(2^n)} \langle b| U \rho U^\dagger |b \rangle U^\dagger |b \rangle \! \langle b| U
= \sum_{b \in \left\{0,1\right\}^n} \frac{1}{2^n}\mathcal{D}_{1/(2^n+1)}(\rho) = \mathcal{D}_{1/(2^n+1)}(\rho).
\end{equation}
This depolarizing channel can be readily inverted, see Eq.~\eqref{eq:inverse-depol}.
In particular,
\begin{equation}
\hat{\rho} = \mathcal{M}^{-1}\left( U^\dagger | \hat{b} \rangle \! \langle \hat{b}| U \right) =(2^n+1) U^\dagger | \hat{b} \rangle \! \langle \hat{b}| U-\mathbb{I} \quad \text{and} \quad \mathcal{M}^{-1}(O_0) = (2^n+1) O_0
\end{equation}
for any traceless matrix $O_0 \in \mathbb{H}_{2^n}$.
The latter reformulation considerably simplifies the expression for the norm $\|O_0 \|_{\mathrm{shadow}}^2$ defined in Eq.~\eqref{eq:shadow-norm}.
A slight reformulation allows us to furthermore capitalize on Eq.~\eqref{eq:third-moment} to exactly compute this norm for traceless observables:
\begin{align}
\| O_0 \|_{\mathrm{shadow}}^2
=& \max_{\sigma \text{ state}} \mathrm{tr} \big( \sigma\; \sum_{b \in \left\{0,1\right\}^n} \mathbb{E}_{U \sim \mathrm{Cl}(2^n)} U^\dagger |b \rangle \! \langle b| U \langle b| U (2^n+1)O_0 U^\dagger |b \rangle^2 \big) \nonumber \\
=& \max_{\sigma \text{ state}} \mathrm{tr} \left( \sigma \; \frac{(2^n+1)^2 \left(\mathrm{tr}(O_0^2)\mathbb{I} + 2 O_0^2\right)}{(2^n+2)(2^n+1)2^n}\right)
= \frac{2^n+1}{2^n+2} \max_{\sigma \text{ state}} \left( \mathrm{tr}(\sigma) \mathrm{tr}(O_0^2) + 2 \mathrm{tr} \left( \sigma O_0^2 \right) \right).
\end{align}
To further simplify this expression, recall $\mathrm{tr}(\sigma)=1$ and note that $\max_{\sigma \text{ state}} \mathrm{tr}(\sigma O_0^2) = \|O_0^2 \|_\infty$, where $\| \cdot \|_\infty$ denotes the spectral norm. The bound Eq.~(\ref{eq:shadow-norm-bound}) then foloows from the elementary relation between the spectral and Hilbert-Schmidt norms: $\|O_0^2 \|_\infty \leq \mathrm{tr}(O_0^2)$.
\end{proof}

\subsection{Performance bound for classical shadows based on random Pauli measurements}
\label{sub:pauli-shadow-details}

\begin{proposition} \label{prop:pauli-shadow}
Adopt a ``random Pauli basis'' measurement primitive, i.e.\ each rotation $\rho \mapsto U \rho U^\dagger$ is a tensor product $U_1 \otimes \cdots \otimes U_n$ of randomly selected single-qubit Clifford gates $U_1,\ldots,U_n \in \mathrm{Cl}(2)$. Then, the associated classical shadow is
\begin{equation}
\hat{\rho} = \bigotimes_{j=1}^n \left( 3 U_j^\dagger |\hat{b}_j \rangle \! \langle \hat{b}_j| U_j - \mathbb{I} \right)
\quad \text{where} \quad  |\hat{b} \rangle = |\hat{b}_1 \rangle \otimes \cdots \otimes |\hat{b}_n \rangle \quad \text{and $\hat{b}_1,\ldots,\hat{b}_n \in \left\{0,1\right\}$.}
\end{equation}
Moreover, the norm defined in Eq.~\eqref{eq:shadow-norm} respects locality.
Suppose that $O \in \mathbb{H}_2^{\otimes k}$ only acts nontrivially on $k$-qubits, e.g.\ $O = \tilde{O} \otimes \mathbb{I}^{\otimes (n-k)}$ with $\tilde{O}\in \mathbb{H}_2^{\otimes k}$. Then
$\| O \|_{\mathrm{shadow}} = \| \tilde{O} \|_{\mathrm{shadow}}$, where $\|\tilde{O} \|_{\mathrm{shadow}}$ is the same norm, but for $k$-qubit systems.
\end{proposition}

\begin{proof}
Unitary rotation and computational basis measurements factorize completely into tensor products. This insight allows us to decompose the measurement channel $\mathcal{M}$ defined in Eq.~ \eqref{eq:measurement-channel} into a tensor product of single-qubit operations. For elementary tensor products $X_1 \otimes \cdots \otimes X_n \in \mathbb{H}_2^{\otimes n}$ we can apply Eq.~\eqref{eq:second-moment} separately for each single-qubit action and infer
\begin{align}
\mathcal{M} \left( X_1 \otimes \cdots \otimes X_n \right)
=& \bigotimes_{j=1}^n \big( \sum_{b_j \in \left\{0,1\right\}}\mathbb{E}_{U_j \sim \mathrm{Cl}(2)} U_j^\dagger |b \rangle \! \langle b| U_j \langle b| U_j X_j U_j^\dagger |b \rangle \big) \nonumber \\
=& \bigotimes_{j=1}^n \big( \sum_{b_j \in \left\{0,1\right\}} \frac{1}{2}\mathcal{D}_{1/(2+1)}(\rho_j) \big)= \mathcal{D}_{1/3}^{\otimes n} \left( X_1 \otimes \cdots \otimes X_n \right).
\end{align}
Linear extension to all of $\mathbb{H}_2^{\otimes n}$ yields the following formula for $\mathcal{M}$ and its inverse:
\begin{equation}
\mathcal{M}(X) = \left(\mathcal{D}_{1/3}\right)^{\otimes n}(X) \quad \text{and} \quad \mathcal{M}^{-1}(X) = \left( \mathcal{D}_{1/3}^{-1}\right)^{\otimes n} (X) \quad \text{for all $X \in \mathbb{H}_2^{\otimes n}$,}
\end{equation}
where $\mathcal{D}_{1/3}^{-1}(Y) = 3Y - \mathrm{tr}(Y) \mathbb{I}$ according to Eq.~\eqref{eq:inverse-depol}.
This formula readily yields a closed-form expression for the classical shadow. Use $U^\dagger |\hat{b}\rangle \! \langle \hat{b}|U = \bigotimes_{j=1}^n U_j | \hat{b}_j \rangle \! \langle \hat{b}_j| U_j$ to conclude
\begin{equation}
\hat{\rho} = \mathcal{M}^{-1} \left( U^\dagger | \hat{b} \rangle \! \langle \hat{b}| U \right)
= \bigotimes_{j=1}^n \mathcal{D}_{1/3}^{-1}\left( U_j^\dagger |\hat{b}_j \rangle \! \langle \hat{b}_j| U_j \right) = \bigotimes_{j=1}^n \left( 3 U_j^\dagger | \hat{b}_j \rangle \! \langle \hat{b}_j| U - \mathbb{I} \right).
\end{equation}
For the second claim, we exploit a key feature of  depolarizing channels and their inverses. The identity matrix is a fix-point, i.e. $\mathcal{D}_{1/3}^{-1}(\mathbb{I}) = \mathbb{I} = \mathcal{D}_{1/3}(\mathbb{I})$.
For $k$-local observables, e.g.\ $O=\tilde{O} \otimes \mathbb{I}^{\otimes (n-k)}$, this feature ensures
\begin{equation}
\mathcal{M}^{-1}\left( \tilde{O} \otimes \mathbb{I}^{\otimes (n-k)}\right) = \left(\left( \mathcal{D}_{1/3}^{-1}\right)^{\otimes k} (\tilde{O}) \right) \otimes \mathbb{I}^{\otimes (n-k)}
= \tilde{\mathcal{M}}^{-1}(\tilde{O}) \otimes \mathbb{I}^{\otimes (n-k)},
\end{equation}
where $\tilde{\mathcal{M}}^{-1}(X) = (\mathcal{D}_{1/3}^{-1})^{\otimes k}(X)$ denotes the inverse channel of a $k$-qubit local Clifford measurement procedure.
This observation allows us to compress the norm \eqref{eq:shadow-norm} to the ``active'' subset of $k$ qubits.
Exploit the tensor product structure $U=U_1\otimes \cdots \otimes  U_n$ with $U_i \sim \mathrm{Cl}(2)$ to conclude
\begin{align}
\left\| \tilde{O} \otimes \mathbb{I}^{\otimes (n-k)}\right\|_{\mathrm{shadow}}^2
=& \max_{\sigma: \text{ state}}
\mathbb{E}_{U \sim \mathrm{Cl}(2)^{\otimes n}} \sum_{b \in \left\{0,1\right\}^n} \langle b| U \sigma U^\dagger |b \rangle \langle b| U \mathcal{M}^{-1}(O \otimes \mathbb{I}^{\otimes (n-k)} U^\dagger |b \rangle^2 \nonumber \\
=& \max_{\sigma: \text{ state}}
\mathbb{E}_{U \sim \mathrm{Cl}(2)^{\otimes k}} \sum_{b \in \left\{0,1\right\}^k} \langle b| U \mathrm{tr}_{k+1,\ldots,n}(\sigma) U^\dagger |b \rangle \langle b| U \tilde{\mathcal{M}}^{-1}(\tilde{O})  U^\dagger |b \rangle^2,
\end{align}
where $\mathrm{tr}_{k+1,\ldots,n}(\sigma)$ denotes the partial trace over all ``inactive'' subsystems.
Partial traces preserve the space of all quantum states. So maximizing over all partial traces $\mathrm{tr}_{k+1,\ldots,n}(\sigma)$ is equivalent to maximizing over all $k$-qubit  states and we exactly recover the norm $\| \tilde{O} \|_{\mathrm{shadow}}^2$ on $k$ qubits. Finally, it is easy to check that the actual location of the active $k$-qubit support of $O$ does not affect the argument.
\end{proof}

Recall that the (squared) norm $\| \cdot \|_{\mathrm{shadow}}^2$ is the most important figure of merit for feature prediction with classical shadows. According to Theorem~\ref{thm:general}, $\max_{1 \leq i \leq M} \|O_i \|_{\mathrm{shadow}}^2$ determines the number of samples required to accurately predict a collection of linear functions $\mathrm{tr}(O_1 \rho),\ldots,\mathrm{tr}(O_M \rho)$.
Viewed from this angle,
Proposition~\ref{prop:pauli-shadow} has profound consequences for predicting (collections of) local observables under the local Clifford measurement primitive.
For each local observable $O_i$, the norm
$\| O_i \|_{\mathrm{shadow}}^2$ collapses to its active support, regardless of its precise location. The size of these supports is governed by the locality alone, not the total number of qubits!

It is instructive to illustrate this point with a simple special case first.

\begin{lemma} \label{lem:pauli}
Let $O$ be a single $k$-local Pauli observable, e.g. $O = P_{p_1} \otimes \cdots \otimes P_{p_k} \otimes \mathbb{I}^{\otimes (n-k)}$, where $p_j \in \left\{X,Y,Z \right\}$.
Then, $\|O \|_{\mathrm{shadow}}^2 = 3^k$, for any choice of the $k$ qubits where nontrivial Pauli matrices act. This scaling can be generalized to arbitrary elementary tensor products supported on $k$ qubits, \emph{e.g.} $O=O_1 \otimes \cdots \otimes O_k \otimes \mathbb{I}^{\otimes (n-k)}$.
\end{lemma}

\begin{proof}
Pauli matrices are traceless and obey, $P_{p_j}^2=\mathbb{I}$ and  $\mathcal{D}_{1/3}^{-1}(P_{p_j})=3P_{p_j}$ for each $p_j \in \left\{X,Y,Z\right\}$. Proposition~\ref{prop:pauli-shadow} and the tensor product structure of the problem then ensure
\begin{align}
\|O \|_{\mathrm{shadow}}^2 =& \| P_{p_1} \otimes \cdots \otimes P_{p_k} \|_{\mathrm{shadow}}^2
\nonumber\\
=& \max_{\sigma: \text{ state}}
\mathbb{E}_{U \sim \mathrm{Cl}(2)^{\otimes k}}\sum_{b \in \left\{0,1\right\}^n} \langle b| U^\dagger \sigma U |b \rangle \langle b| U (\mathcal{D}_{1/3}^{-1})^{\otimes k}(P_1 \otimes \cdots \otimes P_k) U^\dagger |b \rangle^2 \nonumber\\
=& \max_{\sigma:\text{ state}}\mathrm{tr} \Big( \sigma \bigotimes_{j=1}^k
\big( \sum_{b_j \in \left\{0,1\right\}} \mathbb{E}_{U_j \sim \mathrm{Cl}(2)}
U^\dagger |b_j \rangle \! \langle b_j| U \langle b_j| U 3 P_j U^\dagger U |b_j \rangle^2
\big)
\Big) \nonumber\\
=& \max_{\sigma: \text{ state}} \mathrm{tr} \Big( \sigma \bigotimes_{j=1}^k \big(9
\sum_{b \in \left\{0,1\right\}} \frac{\mathrm{tr} \left( P_j^2 \right) \mathbb{I} + 2 P_j^2}{(2+2)(2+1)2}
\big) \Big)
= \max_{\sigma: \text{ state}} \mathrm{tr} \Big( \sigma \bigotimes_{j=1}^k 3 \mathbb{I} \Big)
= 3^k,
\end{align}
where we have used Eq.~\eqref{eq:third-moment} to explicitly evaluate the single qubit Clifford averages.

We leave the extension to more general tensor product observables as an exercise for the dedicated reader.
\end{proof}

The norm expression in Lemma~\ref{lem:pauli} scales exponentially in the locality $k$, but is independent of the total number of qubits $n$. The compression property (Proposition~\ref{prop:pauli-shadow}) suggests that this desirable feature should extend to general $k$-local observables.
And, indeed, it is relatively straightforward to obtain crude upper bounds that scale with $3^{2k}$. The additional factor of two, however, effectively doubles the locality parameter and can render conservative feature prediction with classical shadows prohibitively expensive in concrete applications.

The main result of this section considerably improves upon these crude bounds and \emph{almost} reproduces the (tight) scaling associated with $k$-local Pauli observables.

\begin{proposition} \label{prop:pauli-scaling}
Let $O$ be a $k$-local observable, e.g.\ $O=\tilde{O} \otimes \mathbb{I}^{\otimes (n-k)}$ with $\tilde{O} \in \mathbb{H}_2^{\otimes k}$
Then,
\begin{equation}
\|O \|_{\mathrm{shadow}}^2 \leq 4^k \|O \|_\infty^2, \quad \text{where $\|\cdot \|_\infty$ denotes the spectral/operator norm.}
\end{equation}
The same bound holds for the shadow norm of the traceless part of $O$: $\|O - \tfrac{\mathrm{tr}(O)}{2^n}\mathbb{I} \|_{\mathrm{shadow}}^2 \leq 4^k \|O \|_\infty^2$.
\end{proposition}

The proof is considerably more technical than the proof of Lemma~\ref{lem:pauli} and relies on the following auxiliary result.

\begin{lemma} \label{lem:aux}
Fix two $k$-qubit Pauli observables $P_{\mathbf{p}}=P_{p_1} \otimes \cdots \otimes  P_{p_k}, P_{\mathbf{q}}=P_{q_1} \otimes \cdots \otimes P_{q_k}$
with $\mathbf{p},\mathbf{q}\in \left\{\mathbb{I},X,Y,Z\right\}^k$.
Then, the following formula is true for any state $\sigma$:
\begin{equation}
\mathbb{E}_{U \sim \mathrm{Cl}(2)^{\otimes k}} \sum_{b \in \left\{0,1\right\}^k} \langle b| U \sigma U^\dagger |b \rangle \langle b| U (\mathcal{D}_{1/3}^{-1})^{\otimes k}(P_{\mathbf{p}})U^\dagger |b \rangle \langle b| U (\mathcal{D}_{1/3}^{-1})^{\otimes k} (P_{\mathbf{q}})U^\dagger |b \rangle
=  f(\mathbf{p},\mathbf{q}) \mathrm{tr} \left( \sigma P_{\mathbf{p}} P_{\mathbf{q}}\right),
\end{equation}
where $f(\mathbf{p},\mathbf{q})=0$ whenever there exists an index $i$ such that $p_i \neq q_i$ and $p_i,q_i \neq \mathbb{I}$. Otherwise, $f(\vct{p},\vct{q})=3^s$, where $s$ is the number of non-identity Pauli indices that match ($s=\left| \left\{i: p_i =q_i, p_i \neq \mathbb{I}\right\}\right|$).
\end{lemma}

This combinatorial formula follows from a straightforward, but somewhat cumbersome, case-by-case analysis based on the (single-qubit) relations \eqref{eq:second-moment} and \eqref{eq:third-moment}. We include a proof at the end of this subsection.

\begin{proof}[Proof of Proposition~\ref{prop:pauli-scaling}]

Proposition~\ref{prop:pauli-shadow} allows us to restrict our attention to the relevant $k$-qubit region on which $\tilde{O} \in \mathbb{H}_2^{\otimes k}$ acts nontrivially.
Next, expand $\tilde{O}$ in the (tensor product) Pauli basis, i.e.\ $\tilde{O} = \sum_{\mathbf{p}} \alpha_{\mathbf{p}}P_{\vct{p}}$ with $\mathbf{p} \in \left\{ \mathbb{I},X,Y,Z\right\}^k$.
Fix an arbitrary $k$-qubit state $\sigma$ and use Lemma~\ref{lem:aux} to conclude
\begin{align}
\|\tilde{O} \|_{\mathrm{shadow}}^2&= \max_{\sigma \text{ state}}\mathbb{E}_{U \sim \mathrm{Cl}(2)^{\otimes k}}
\sum_{b \in \left\{0,1\right\}^k} \langle b| U \sigma U^\dagger |b \rangle \langle b| U (\mathcal{D}^{-1}_{1/3})^{\otimes k}(\tilde{O}) U^\dagger |b \rangle^2 \nonumber\\
=&\max_{\sigma \text{ state}} \sum_{\mathbf{p},\mathbf{q}} \alpha_{\mathbf{p}}\alpha_{\mathbf{q}}
\mathbb{E}_{U \sim \mathrm{Cl}(2)^{\otimes k}} \sum_{b \in \left\{0,1\right\}^k} \langle b| U \sigma U^\dagger |b \rangle \langle b| U (\mathcal{D}_{1/3}^{-1})^{\otimes k}(P_{\mathbf{p}})U^\dagger |b \rangle \langle b| U (\mathcal{D}_{1/3}^{-1})^{\otimes k} (P_{\mathbf{q}})U^\dagger |b \rangle  \nonumber\\
=& \max_{\sigma \text{ state}} \sum_{\mathbf{p},\mathbf{q}} \alpha_{\mathbf{p}} \alpha_{\mathbf{q}}
f(\mathbf{p},\mathbf{q}) \mathrm{tr} \left( \sigma P_{\mathbf{p}} P_{\mathbf{q}} \right)
= \max_{\sigma \text{ state}}\mathrm{tr} \left( \sigma \; \sum_{\mathbf{p},\mathbf{q}} \alpha_{\mathbf{p}} \alpha_{\mathbf{q}}
f(\mathbf{p},\mathbf{q}) \mathrm{tr} \left( \sigma P_{\mathbf{p}} P_{\mathbf{q}} \right)\right) \nonumber\\
= & \left\| \sum_{\mathbf{p},\mathbf{q}} \alpha_{\mathbf{p}} \alpha_{\mathbf{q}}
f(\mathbf{p},\mathbf{q}) \mathrm{tr}  P_{\mathbf{p}} P_{\mathbf{q}} \right\|_\infty,
\end{align}
where $f(\vct{p},\vct{q})$ is the combinatorial function defined in Lemma~\ref{lem:aux}.
The last equality follows from the dual characterization of the spectral norm: $\|A \|_\infty = \max_{\sigma: \text{ state}} \mathrm{tr}(\sigma A)$ for any positive semidefinite matrix $A$.

We can further simplify this expression by introducing a partial order on Pauli strings $\vct{q},\vct{s} \in \left\{ \mathbb{I},X,Y,Z \right\}^n$. We write $\vct{q}\vartriangleright \vct{s}$ if it is possible to obtain $\vct{q}$ from $\vct{s}$ by replacing some local non-identity Paulis with $\mathbb{I}$.
Moreover, let $| \vct{q}|=\left| \left\{i:\; q_i \neq \mathbb{I} \right\} \right|$ denote the number of non-identity Pauli's in the string $\vct{q}$.
Then,
\begin{equation}
\left\| \sum_{\mathbf{p},\mathbf{q}} \alpha_{\mathbf{p}} \alpha_{\mathbf{q}}
f(\mathbf{p},\mathbf{q}) \mathrm{tr}  P_{\mathbf{p}} P_{\mathbf{q}} \right\|_\infty
= \left\| \tfrac{1}{3^k}\sum_{\vct{s} \in \left\{X,Y,Z\right\}^k} \left( \sum_{\vct{q}\vartriangleright \vct{s}} 3^{|\vct{q}|} \alpha_{\vct{q}}P_{\vct{q}} \right)^2 \right\|_\infty
\leq \frac{1}{3^k} \sum_{\vct{s} \in \left\{X,Y,Z\right\}^k}\left( \sum_{\vct{q}\vartriangleright \vct{s}} 3^{|\vct{q}|} \alpha_{\vct{q}}P_{\vct{q}} \right)^2,
\end{equation}
where we have used  $\|P_{\vct{q}}\|_\infty =1$ for all Pauli strings.
Next, note that for fixed $\vct{s} \in \left\{X,Y,Z \right\}^k$,
\begin{equation}
\sum_{\vct{q} \vartriangleright \vct{s}} 3^{|\vct{q}|}=3^k + k 3^{k-1}+ \binom{k}{2} 3^{k-2} + \cdots + 1 = 4^k.
\end{equation}
Together with Cauchy-Schwarz, this numerical insight implies
\begin{equation}
    \tfrac{1}{3^k}\sum_{\vct{s}\in \left\{X,Y,Z\right\}^k} \left( \sum_{\vct{q} \vartriangleright\vct{s}}3^{|\vct{q}|} |\alpha_{\vct{q}}| \right)^2
  \leq \tfrac{1}{3^k}\sum_{\vct{s}\in \left\{X,Y,Z\right\}^k}
\left( \sum_{\vct{q}\vartriangleright \vct{s}} 3^{|\vct{q}|} \right) \left( \sum_{\vct{q}\vartriangleright \vct{s}} 3^{|\vct{q}|}| \alpha_{\vct{p}}^2 \right)
= 4^k \sum_{\vct{s} \in \left\{X,Y,Z\right\}}\sum_{\vct{q}\vartriangleright \vct{s}} 3^{|\vct{q}|-k}|\alpha_{\vct{q}}|^2.
\end{equation}
Finally, observe that every $\vct{q}\in \left\{ \mathbb{I},X,Y,Z\right\}^k$ is dominated by exactly $3^{k-|\vct{q}|}$ different strings $\vct{s} \in \left\{X,Y,Z\right\}^k$. This ensures
\begin{equation}
4^k \sum_{\vct{s} \in \left\{X,Y,Z\right\}} 3^{|\vct{q}|-k} | \alpha_{\vct{q}}|^2 = 4^k \sum_{\vct{q}\in \left\{ \mathbb{I},X,Y,Z\right\}} | \alpha_{\vct{q}}|^2 = 4^k 2^{-k} \| \tilde{O} \|_2^2,
\end{equation}
because Pauli matrices are proportional to an orthonormal basis of $\mathbb{H}_2^{\otimes k}$:
$
\sum_{\vct{q}} |\alpha_\vct{q}|^2= \sum_{\vct{q}} \big| 2^{-k}\mathrm{tr} \big( \sigma_{\vct{q}}\tilde{O}\big) \big|^2=2^{-k} \|\tilde{O} \|_2^2.
$
The general claim then follows from the fundamental relation among Schatten norms: $\|\tilde{O} \|_2^2 \leq 2^k \|\tilde{O} \|_\infty^2=2^k \|O \|_\infty^2$.

The bound on traceless parts $O_0$ of observables is nearly analogous, because the transition from $O$ to $O_0$ respects locality. E.g.\ $O=\tilde{O} \otimes \mathbb{I}^{\otimes (n-k)}$ obeys $O_0 = \tilde{O}_0 \otimes \mathbb{I}^{\otimes (n-k)}$.
To get the same bound, we use that this transition is a contraction in Hilbert-Schmidt norm:
\begin{equation*}
\|O_0 \|_{\mathrm{shadow}}^2 = \|\tilde{O}_0 \|_{\mathrm{shadow}}^2 \leq 4^k 2^{-k}\|\tilde{O}_0 \|_2^2 \leq 4^k 2^{-k}\|\tilde{O} \|_2^2 \leq 4^k \|\tilde{O} \|_\infty^2 = \| O \|_\infty^2.
\end{equation*}

\end{proof}

\begin{proof}[Proof of Lemma~\ref{lem:aux}]
Since Pauli observables decompose nicely into tensor products, this claim readily follows from extending a single-qubit argument.
Note that $\mathcal{D}_{1/3}^{-1}(P_p) = 3 P_p$ for $p \neq \mathbb{I}$ and $\mathcal{D}_{1/3}^{-1}(\mathbb{I}) =\mathbb{I}$.
It is straightforward to evaluate the single-qubit expression for the trivial case $P_p=P_q=\mathbb{I}$. Fix a state $\sigma$ and compute
\begin{equation}
\mathbb{E}_{U \sim \mathrm{Cl}(2)} \sum_{b \in \{0, 1\}} \langle b| U \sigma U^\dagger |b \rangle \langle b| U \mathcal{D}^{-1}_{1/3}(\mathbb{I})U^\dagger |b \rangle^2 = \mathbb{E}_{U \sim \mathrm{Cl}(2)} \sum_{b \in \left\{0,1\right\}}\langle b| U \sigma U^\dagger |b \rangle = \mathbb{E}_{U \sim \mathrm{Cl}(2)} \mathrm{tr}(\sigma)
= \mathrm{tr} \left( \sigma \mathbb{I}^2 \right).
\end{equation}
Next, suppose $P_q = \mathbb{I}$, but $P_p \neq \mathbb{I}$. This single-qubit case is covered by Eq.~\eqref{eq:second-moment}:
\begin{align}
& \mathbb{E}_{U \sim \mathrm{Cl}(2)} \sum_{b \in \left\{0,1\right\}} \langle b| U \sigma U^\dagger |b \rangle \langle b| U \mathcal{D}_{1/3}^{-1}(P_p) U^\dagger |b \rangle \langle b| U \mathcal{D}_{1/3}^{-1}\mathbb{I} U^\dagger |b \rangle \nonumber\\
=& \mathrm{tr} \Big( \sigma\sum_{b \in \left\{0,1\right\}}
U^\dagger |b \rangle \! \langle b| U \langle b| U 3 P_p U^\dagger |b \rangle \Big)
= 3 \mathrm{tr} \Big( \sigma
\sum_{b \in \left\{0,1\right\}} \frac{1}{2}
\mathcal{D}_{1/3}(P_p) \Big)
= \mathrm{tr} \left( \sigma P_p \mathbb{I} \right),
\end{align}
because $\mathcal{D}_{1/3}(P_p) = \frac{1}{3}P_p$. The case $P_p=\mathbb{I}$ and $P_q \neq \mathbb{I}$ leads to analogous results. Finally, suppose that both $P_p,P_q \neq \mathbb{I}$. By assumption $\mathcal{D}_{1/3}^{-1}(P_p)$, $\mathcal{D}_{1/3}^{-1}(P_q)$ and both matrices are traceless. Hence,
we can resort to Eq.~\eqref{eq:third-moment} to conclude
\begin{align}
& \mathbb{E}_{U \sim \mathrm{Cl}(2)^{\otimes n}} \sum_{b \in \left\{0,1\right\}^k} \langle b| U \sigma U^\dagger |b \rangle \langle b| U (\mathcal{D}_{1/3}^{-1})^{\otimes k}(P_{p})U^\dagger |b \rangle \langle b| U (\mathcal{D}_{1/3}^{-1})^{\otimes k} (P_{q})U^\dagger |b \rangle \nonumber \\
=& \mathrm{tr} \Big(\sigma
\sum_{b \in \left\{0,1\right\}} U^\dagger |b \rangle \! \langle b| U \langle b| U 3 P_p U^\dagger |b \rangle \langle b| U 3 P_q U^\dagger |b \rangle \Big)
= 9 \mathrm{tr} \Big( \sigma \sum_{b \in \left\{0,1\right\}} \frac{\mathrm{tr}(P_p P_q) \mathbb{I} + P_p P_q + P_q P_p}{(2+2)(2+1)2} \Big)
\end{align}
for any state $\sigma$. Pauli matrices are orthogonal ($\mathrm{tr}(P_p P_q)=2 \delta_{p,q}$) and anticommute ($P_p P_q+P_q P_p = 2 \delta_{p,q}$). This implies that the above expression vanishes whenever $p \neq q$. If $p=q$ it evaluates to $3 \mathrm{tr}(\sigma P_p P_q)$ and we can conclude that the single qubit average always equals
\begin{equation}
 f(p,q) \mathrm{tr} \left( \sigma P_p P_q \right)
\quad \text{where} \quad
f(p,q) =
\begin{cases}
1 &\text{ if $p=\mathbb{I}$ or $q=\mathbb{I}$}, \\
3 & \text{ if $p=q \neq \mathbb{I}$}, \\
0 & \text{else}.
\end{cases}
\end{equation}
The statement then follows from extending this formula to tensor products of $k$ Pauli matrices.
\end{proof}

\section{Additional computations and proofs for predicting nonlinear functions}
\label{sec:nonlinearrig}

We focus on the particularly relevant task of predicting quadratic functions with classical shadows, using
\begin{equation}
\hat{o}(N, 1) = \frac{1}{N(N-1)} \sum_{j \neq l} \Tr(O \hat\rho_i \otimes \hat\rho_j)
\quad \text{to predict} \quad \mathrm{tr} \left( O \rho \otimes \rho \right) = \E \hat{o}(N,1).
\end{equation}

\subsection{General variance bound}

\begin{lemma}[Variance] \label{lem:symm-full}
The variance associated with the estimator $\hat{O}(N,1)$ obeys
\begin{align}
    \Var[\hat{o}(N, 1)] & = {N \choose 2}^{-1} \Big( 2(N-2)\Var[\Tr(O_s \hat\rho_1 \otimes \rho)]  + \Var[\Tr(O_s \hat\rho_1 \otimes \hat\rho_2)] \Big) \nonumber \\
    & \leq \frac{4}{N^2} \Var[\Tr(O \hat\rho_1 \otimes \hat\rho_2)] + \frac{2}{N} \Var[\Tr(O \hat\rho_1 \otimes \rho) ]
     + \frac{2}{N} \Var[\Tr(O \rho \otimes \hat\rho_1) ], \label{eq:varsimplebound-full}
\end{align}
where $O_s = (O+SOS)/2$ is the symmetrized version of $O$ and $S$ denotes the swap operator (
$S|\psi \rangle \otimes | \phi \rangle = | \phi \rangle \otimes | \psi \rangle$).
\end{lemma}
\begin{proof}
First, note that $\hat{o}(N,1)$ and the target $\mathrm{tr}(O \rho \otimes \rho)$ are invariant under symmetrization. This ensures
\begin{equation}
\hat{o}(N,1) = \binom{N}{2} \sum_{i<j} \mathrm{tr} \left( O_s \hat\otimes \hat{\rho}_j \right) \quad \text{and moreover} \quad \mathrm{tr} \left( O \rho \otimes \rho\right)=\mathrm{tr} \left( O_s \rho \otimes \rho \right).
\end{equation}
Thus, we may without loss replace the original observable $O$ by its symmetrized version $O_s$.
Next, we expand the definition of the variance:
\begin{align}
    \Var[\hat{o}(N, 1)] = &
    \mathbb{E} \left[ \left(\hat{o}(N,1)-\mathrm{tr}(O_s \rho \otimes \rho) \right)^2 \right] \nonumber\\
=&
    {N \choose 2}^{-2} \sum_{i < j} \sum_{k < l} \Big( \E \Big[ \Tr(O_{s} \hat\rho_i \otimes \hat\rho_j) \Tr(O_{s} \hat\rho_k \otimes \hat\rho_l) \Big] - \Tr(O_{s} \rho \otimes \rho)^2 \Big) \nonumber \\
    = & {N \choose 2}^{-2} \sum_{i < j} \E \Big[ \Tr(O_{s} \hat\rho_i \otimes \hat\rho_j)^2 \Big] - \Tr(O_{s} \rho \otimes \rho)^2 \Big) \nonumber\\
    +&  2 {N \choose 2}^{-2} \sum_{i < j} \sum_{l \neq i, j} \Big( \E \Big[ \Tr(O_{s} \hat\rho_i \otimes \hat\rho_j) \Tr(O_{s} \hat\rho_i \otimes \hat\rho_l) \Big] - \Tr(O_{s} \rho \otimes \rho)^2 \Big) \nonumber\\
     = & {N \choose 2}^{-1} \Var[\Tr(O_s \hat\rho_1 \otimes \hat\rho_2)] + {N \choose 2}^{-1}  2(N-2) \Var[\Tr(O_s \hat\rho_1 \otimes \rho)].
\end{align}
We can use the inequality $\Var[(A + B) / 2] \leq (\Var[A] + \Var[B]) / 2$ (for any pair of random variables $A, B$) to obtain a simplified upper bound:
\begin{align}
    \Var[\hat o(N, 1)] & = {N \choose 2}^{-1} \Var[\Tr(O_s \hat\rho_1 \otimes \hat\rho_2)] + {N \choose 2}^{-1} 2(N-2) \Var[\Tr(O_s \hat\rho_1 \otimes \rho)] \nonumber \\
    & \leq \frac{4}{N^2} \Var[\Tr(O_s \hat\rho_1 \otimes \hat\rho_2)] + \frac{4}{N} \Var[\Tr( O_s \hat\rho_1 \otimes \rho) ] \nonumber \\
    & \leq \frac{4}{N^2} \Var[\Tr(O \hat\rho_1 \otimes \hat\rho_2)] + \frac{2}{N} \Var[\Tr(O \hat\rho_1 \otimes \rho) ]
     + \frac{2}{N} \Var[\Tr(O \rho \otimes \hat\rho_1) ].
\end{align}
\end{proof}

\subsection{Concrete variance bounds for random Pauli measurements}

\begin{proposition} \label{lem:unbiasedtwoP}
Suppose that $O$ describes a quadratic function $\mathrm{tr}(O \rho \otimes \rho)$
 that acts on at most $k$-qubits in the first system and at most $k$-qubits in the second system and obeys $\|O \|_\infty \geq 1$.
Then,
\begin{equation}
\max\left(\Var[\Tr(O \rho \otimes \hat\rho_1)], \Var[\Tr(O \hat{\rho}_1 \otimes \rho)], \sqrt{\Var[\Tr(O \hat{\rho}_1 \otimes \hat{\rho}_2)]} \right) \leq 4^{k} \norm{O}^2_\infty.
\label{eq:twomean+varianceP}
\end{equation}
\end{proposition}
\begin{proof}
Because of the single-qubit tensor product structure in the random Pauli measurement and the inverted quantum channel $\mathcal{M}^{-1}_P$, the tensor product of two snapshots $\hat\rho_1 \otimes \hat\rho_2$ of the unknown quantum state $\rho$ may be viewed as a single snapshot of the tensor product state $\rho \otimes \rho$:
\begin{align}
  \hat\rho_1 \otimes \hat\rho_2 & = \bigotimes_{i=1}^n \Big( \mathcal{M}_1^{-1}(U_1^{(i)} |b_1^{(i)}\rangle\! \langle b_1^{(i)}| (U_1^{(i)})^\dagger) \Big) \bigotimes_{i=1}^n \Big(\mathcal{M}_1^{-1}(U_2^{(i)} |b_2^{(i)}\rangle\! \langle b_2^{(i)}| (U_2^{(i)})^\dagger) \Big) \nonumber \\
  &  = \bigotimes_{i=1}^{2n} \mathcal{M}_1^{-1}(U^{(i)} |b^{(i)}\rangle\! \langle b^{(i)}| (U^{(i)})^\dagger) =: \hat\rho.
\end{align}
Hence $\Tr(O \hat\rho_1 \otimes \hat\rho_2) = \Tr(O \hat\rho)$ and, by assumption, $O$ is an observable that acts on $k+k = 2k$ qubits only. The claim then follows from invoking the variance bounds for linear feature prediction presented in Proposition~\ref{prop:pauli-scaling}.
\end{proof}

\subsection{Concrete variance bounds for random Clifford measurements}

In contrast to the Pauli basis setup, variances for quadratic feature prediction with Clifford basis measurements cannot be directly reduced to its linear counterpart.
Nonetheless, a more involved direct analysis does produces bounds that do closely resemble the linear base case.

\begin{proposition} \label{lem:unbiasedtwoC}
Suppose that $O$ describes a quadratic function $\mathrm{tr}(O \rho \otimes \rho)$ and obeys $\mathrm{tr}(O^2) \geq 1$.
Then, the variance associated with classical shadow estimation (random Clifford measurements) obeys
\begin{equation}
\max\left(\Var[\Tr(O \rho \otimes \hat\rho_1)], \Var[\Tr(O \hat{\rho}_1 \otimes \rho)], \sqrt{\Var[\Tr(O \hat{\rho}_1 \otimes \hat{\rho}_2)]} \right) \leq \sqrt{9+6/2^n} \Tr(O^2).
\label{eq:twomean+varianceC}
\end{equation}
The pre-factor $\sqrt{9+6/2^n}$ converges to the constant 3 at an exponential rate in system size.
\end{proposition}

This claim is based on the following technical Lemma and insights regarding linear feature prediction.

\begin{lemma} \label{lem:twoC}
Suppose that $O$ describes a quadratic function $\mathrm{tr}(O \rho \otimes \rho)$.
Then,
\begin{equation}
\Var[\Tr(O \hat{\rho}_1 \otimes \hat{\rho}_2)] \leq 9 \Tr(O^2)+ \frac{6}{2^n} \|O \|_\infty^2.
\end{equation}
\end{lemma}

\begin{proof}[Proof of Proposition~\ref{lem:unbiasedtwoC}]
The variance of $\Tr(O \rho \otimes \hat\rho_1)$ is equivalent to the variance of $\Tr (\tilde{O}_\rho \hat{\rho})$, where $\tilde{O}_\rho = \mathrm{tr}_1 \left(\rho \otimes \mathbb{I} O \right)$ describes a linear function. According to Proposition~\ref{prop:Clifford-variance}, this variance term obeys
\begin{equation}
\mathrm{Var} \left[ \mathrm{tr} \left( O \rho \otimes \hat{\rho} \right)\right]
=\mathrm{Var} \left[ \mathrm{tr} \left( \tilde{O}_\rho \hat{\rho}_1 \right) \right]
\leq 3 \mathrm{tr} \left( \tilde{O}_\rho^2 \right)
= \mathrm{tr} \left( \mathrm{tr}_1 \left( \rho \otimes \mathbb{I} O \right)^2 \right)
\leq 3 \mathrm{tr}(O^2),
\end{equation}
because $\mathrm{tr}(\rho)=1$ and $\mathrm{tr}(\rho^2) \leq 1$.
A similar argument takes care of the second variance contribution $\mathrm{Var}\left[ \mathrm{tr} \left( O \hat{\rho}_1 \otimes \rho \right) \right]$. Lemma~\ref{lem:twoC} supplies a bound for the square of the final contribution. By assumption $\sqrt{\mathrm{tr}(O^2)} \leq \mathrm{tr}(O^2)$ and the claim follows.
\end{proof}

The remainder of this section is devoted to proving Lemma~\ref{lem:twoC}. Unfortunately, there does not seem to be a direct way to relate this task to variance bounds for linear feature prediction.
Instead, we base our analysis on the 3-design property \eqref{eq:third-moment} of Clifford circuits and a reformulation of this feature in terms of permutation operators.
This strategy is inspired by the approach developed in \cite{brandao2019complexity}, but conceptually and technically somewhat simpler. We believe that similar arguments extend to variances associated with higher order polynomials, but do refrain from a detailed analysis. Instead, we carefully outline the main ideas and leave a rigorous extension to future work.

\paragraph{Problem statement and reformulation:}

We will ignore symmetrization (which can only make the variance smaller) and focus on bounding the variance of $\mathrm{tr} \left(O \hat{\rho}_1 \otimes \hat{\rho}_2\right)$, where each $\hat{\rho}_i$ is an independent classical shadow.
To simplify notation, we set $d=2^n$ and
define the following traceless variants of $O$:
\begin{align}
O_0^{(1)}=& \mathrm{tr}_2 (O) - \frac{\mathrm{tr} \left(O\right)}{d}\mathbb{I}, \quad \text{and} \quad
O_0^{(2)}= \mathrm{tr}_1 (O) - \frac{\mathrm{tr}(O)}{d}\mathbb{I}, \quad \text{as well as} \nonumber\\
O_0^{(1,2)}=& O- \mathrm{tr}_2 (O) \otimes \frac{\mathbb{I}}{d}-\frac{\mathbb{I}}{d}\otimes \mathrm{tr}_1 (O) + \mathrm{tr}(O) \frac{\mathbb{I}}{d}\otimes \frac{\mathbb{I}}{d}.
\end{align}
Here, $\mathrm{tr}_a (O)$ with $a=1,2$ denotes the partial trace over the first and second system, respectively.
All three operators are traceless (recall $\mathrm{tr} \left( \mathrm{tr}_a (O)\right)=\mathrm{tr}(O)$) and the final (bipartite) operator has the additional property that both partial traces vanish identically: $\mathrm{tr}_a \big( O_0^{(1,2)}\big)=0$.

 Proposition~\ref{prop:Clifford-variance}
 asserts $\hat{\rho}_a = (d+1)U_a^\dagger |\hat{b}_a \rangle \! \langle \hat{b}_a| U_a-\mathbb{I}$, where each $U_a \in \mathrm{Cl}(d)$ is a random Clifford unitary and $\hat{b}_a\in \left\{0,1\right\}^n$ is the outcome of a computational basis measurement.
 These explicit formulas allow us to
decompose the expression of interest in the following fashion:
\begin{align}
\mathrm{tr} \left( O \hat{\rho}_1 \otimes \hat{\rho}_2 \right)
=& (d+1)^2 \mathrm{tr} \left( O_0^{(1,2)} U_1^\dagger| \hat{b}_1 \rangle \! \langle \hat{b}_1|U_1 \otimes U_2^\dagger| \hat{b}_1 \rangle \! \langle \hat{b}_2|U_2 \right) + \frac{\mathrm{tr}(O)^2}{d^2} \nonumber \\
+& \frac{d+1}{d} \mathrm{tr} \left( O_0^{(1)} U_1^\dagger |\hat{b}_1 \rangle \! \langle \hat{b}_1|U_1\right)
+ \frac{d+1}{d} \mathrm{tr} \left( O_0^{(2)} U_2^\dagger |\hat{b}_2 \rangle \! \langle \hat{b}_2|U_2\right). \label{eq:master-formula}
\end{align}
The variance corresponds to the expected square of this expression.
The second term is constant and does not contribute. We analyze the remaining terms on a case-by case basis.

\paragraph{Linear terms:}

The third and fourth terms in Eq.~\eqref{eq:master-formula} are linear feature functions in one classical shadow only.
Their (squared) contribution to the overall variance is characterized by Proposition~\ref{prop:Clifford-variance}:
\begin{equation}
\mathbb{E}
\left[ \left(\frac{d+1}{d} \mathrm{tr} \left( O_0^{(a)} U_a^\dagger |\hat{b}_a \rangle \! \langle \hat{b}_a|U_a\right) \right)^2\right]
\leq \frac{3}{d^2} \left\|O_0^{(a)}\right\|_2^2
\quad \text{for $a=1,2$.}
\end{equation}
Both bounds can be related to the Hilbert-Schmidt norm (squared) of the original observable:
\begin{equation}
\frac{3}{d^2} \left\|O_0^{(a)}\right\|_2^2
\leq \frac{3}{d^2}\left\| \mathrm{tr}_a (O) \right\|_2^2 \leq 3 \|O\|_2^2 = 3\mathrm{tr}\left( O^2 \right). \label{eq:linear-bound}
\end{equation}

\paragraph*{Leading-order term:}

We need to bound $\mathbb{E} \big[ (d+1)^4 \mathrm{tr} \left(O_0^{(1,2)} U_1^\dagger |\hat{b}_1 \rangle \! \langle \hat{b}_1|U_1 \otimes U_2^\dagger |\hat{b}_2 \rangle \! \langle \hat{b}_2| U_2 \right)^2 \big]$,
where $O_0^{(1,2)}$ has the special property that both partial traces vanish identically: $\mathrm{tr}_a \big( O_0^{(1,2)}\big)=0$ for $a=1,2$. Moreover, the Hilbert-Schmidt norm (squared) of this operator factorizes nicely:
\begin{equation}
\left\| O_0^{(1,2)}\right\|_2^2
= \|O \|_2^2 - \frac{1}{d}\big\| O_0^{(1)}\big\|_2^2-\big\| O_0^{(2)}\big\|_2^2- \frac{\mathrm{tr}(O)^2}{d^2}.
\end{equation}
Not only is this expression bounded by the original Hilbert-Schmidt norm $\|O\|_2^2$.
The norms of partial traces also feature explicitly with a minus sign. This will allow us to fully counter-balance the variance contributions \eqref{eq:linear-bound} from the linear terms.

Next, we use the 3-design property \eqref{eq:k-design} of Clifford circuits in dimension $d=2^n$:
\begin{equation}
\mathbb{E}_{U_a \sim \mathrm{Cl}(d)} \left[ \left( U_a^\dagger |b_a \rangle \! \langle b_a| U_a \right)^{\otimes 3}\right]
= \binom{d+2}{3}^{-1} P_{\vee^3},
\end{equation}
where $P_{\vee^3}$ is the projector onto the totally symmetric subspace of $\mathbb{C}^d \otimes \mathbb{C}^d \otimes \mathbb{C}^d$.
This formula implies
\begin{equation}
 \mathbb{E} \left[ (d+1)^4 \mathrm{tr} \left(O_0^{(1,2)} U_1^\dagger |\hat{b}_1 \rangle \! \langle \hat{b}_1|U_1 \otimes U_2^\dagger |\hat{b}_2 \rangle \! \langle \hat{b}_2| U_2 \right)^2 \right]
\leq
\mathrm{tr} \left( O_0^{(1,2)}\otimes O_0^{(1,2)}\otimes \rho \otimes \rho \; P_{\vee^3}^{(\mathrm{odd})} \otimes P_{\vee^3}^{(\mathrm{even})} \right),
\end{equation}
where the superscripts ``even'' and ``odd'' indicate on which subset of tensor factors the projectors act.

Next, we exploit the fact that symmetric projectors can be decomposed into permutation operators: $(3!) P_{\vee^3}=\sum_{\pi \in S_3}W_\pi$, where $S_3$ is the group of all six permutations of three elements and the permutation operators act like $W_\pi |\psi_1 \rangle \otimes | \psi_2 \rangle \otimes | \psi_3 \rangle=|\psi_{\pi^{-1}(1)}\rangle \otimes | \psi_{\pi^{-1}(2)} \rangle \otimes | \psi_{\pi^{-1}(3)} \rangle$:
\begin{equation}
\mathrm{tr} \left( O_0^{(1,2)}\otimes O_0^{(1,2)}\otimes \rho \otimes \rho \; P_{\vee^3}^{(\mathrm{odd})} \otimes P_{\vee^3}^{(\mathrm{even})} \right)
= \sum_{\pi,\tau \in S_3}\mathrm{tr} \left( O_0^{(1,2)}\otimes O_0^{(1,2)}\otimes \rho \otimes \rho \; W_{\pi}^{(\mathrm{odd})} \otimes W_{\tau}^{(\mathrm{even})} \right).
\end{equation}
The specific structure of $O_0^{(1,2)}$ implies that several contributions must vanish.
Permutations that have either $1$ or $2$ as a fix-point lead to a partial trace of $O_0^{(1,2)}$ that evaluates to zero.
There are only three permutations that do not have such fix-points: The flip $(1,2,3)\mapsto (2,1,3)$ and the two cycles $(1,2,3)\mapsto (3,1,2)$, $(1,2,3)\mapsto (2,3,1)$.
There are in total $9=3^2$ potential combinations of such permutations. Each of them results in a trace expression that can be upper-bounded by Hilbert-Schmidt norms. For instance the pair flip and flip produces
\begin{equation}
\mathrm{tr} \left( O_0^{(1,2)} O_0^{(1,2)}\right) \mathrm{tr}(\rho)^2 = \left\| O_0^{(1,2)}\right\|_2^2.
\end{equation}
All other 8 contributions can also be bounded by this expression and we conclude
\begin{equation}
 \mathbb{E} \left[ (d+1)^4 \mathrm{tr} \left(O_0^{(1,2)} U_1^\dagger |\hat{b}_1 \rangle \! \langle \hat{b}_1|U_1 \otimes U_2^\dagger |\hat{b}_2 \rangle \! \langle \hat{b}_2| U_2 \right)^2 \right]
\leq 9 \left\| O_0^{(1,2)}\right\|_2^2
\label{eq:leading-term}
\end{equation}

\paragraph{Bounds on cross-terms:}

Cross-terms are considerably easier to evaluate, because one (or both) random matrices only feature linearly.
We can use $\mathbb{E} \left[ U_a^\dagger| \hat{b}_a \rangle \! \langle \hat{b}_a| U_a \right] = \mathcal{D}_{1/(d+1)}(\rho) = \frac{\rho+\mathbb{I}}{d+1}$ to effectively get rid of the linear contribution.
For instance,
\begin{equation}
\left( \frac{d+1}{d}\right)^2
\mathbb{E} \left[ \prod_{a=1,2} \mathrm{tr} \left( O_0^{(1)}U_a^\dagger |\hat{b}_a \rangle \! \langle \hat{b}_a| U_a \right)\right]=  \frac{1}{d^2}\mathrm{tr} \left( O_0^{(1)} \rho \right) \mathrm{tr} \left( O_0^{(2)} \rho \right)
\leq  \frac{1}{2d^2}
\left( \|O_0^{(1)}\|_\infty^2 + \|O_0^{(2)}\|_\infty^2 \right),
\end{equation}
where $\| \cdot \|_\infty$ denotes the operator norm.
Cross terms that do feature the leading order term require slightly more work, but can be addressed in a similar fashion. Using linearity in one snapshot reduces the expression to an expectation of a quadratic function in one snapshot only. The remaining computation is similar to the proof of Proposition~\ref{prop:Clifford-variance} and yields
\begin{equation}
\frac{(d+1)^3}{d}\mathbb{E} \left[  \mathrm{tr} \left( O_0^{(1,2)} U_1^\dagger |\hat{b}_1 \rangle \! \langle \hat{b}_1| U_1 \otimes U_2^\dagger |\hat{b}_2 \rangle \! \langle \hat{b}_2| U_2 \right)  \mathrm{tr} \left( O_0^{(a)} U_a^\dagger |\hat{b}_a \rangle \! \langle \hat{b}_a| U_a \right) \right]
\leq \frac{3}{2d^2}\left( \| \tilde{O}^{(a)}_\rho \|_2^2 + \|O_0^{(a)} \|_2^2 \right),
\end{equation}
for $a=1,2$, as well as $\tilde{O}^{(1)}_\rho = \mathrm{tr}_2 \left( \mathbb{I} \otimes \rho O \right)$ and  $\tilde{O}^{(2)}_\rho=\mathrm{tr}_1 \left( \rho \otimes \mathbb{I}O \right)$, respectively.

\paragraph{Full variance bound:}

We are now ready to combine all individual bounds to control the full variance:
\begin{align}
\mathrm{Var} \left[ \hat{o} \right]
\leq & \mathbb{E} \left( (d+1)^2 \mathrm{tr} \left( O_0^{(1,2)} U_1^\dagger |\hat{b}_1 \rangle \! \langle \hat{b}_1| U_1 \otimes U_2^\dagger |\hat{b}_2 \rangle \! \langle \hat{b}_2| U_2 \right)+\sum_{a=1,2} \frac{d+1}{d}\mathrm{tr} \left( O_0^{(a)} U_a^\dagger |\hat{b}_a \rangle \! \langle \hat{b}_a| U_a \right) \right)^2 \nonumber\\
\leq & 9 \|O_0^{(1,2)}\|_2^2 + \frac{6}{2d^2}\left( \| \mathrm{tr}_2 \left(  \mathbb{I} \otimes \rho O \right) \|_2^2 + \|O_0^{(1)}\|_2^2 \right)
+ \frac{6}{2d^2}\left( \| \mathrm{tr}_1 \left(  \rho \otimes \mathbb{I} O \right) \|_2^2 \right) \nonumber\\
+ &\frac{3}{d^2}\|O_0^{(1)}\|_2^2 + \frac{3}{d^2}\|O_0^{(2)}\|_2^2 + \frac{1}{2d^2}\left( \|O_0^{(1)} \|_\infty^2 + \|O_0^{(2)}\|_\infty^2 \right).
\end{align}
Standard norm inequalities, as well as the explicit expression for $\|O_0^{(1,2)}\|_2^2$ allow for counter-balancing some of the sub-leading terms and we conclude
\begin{equation}
\mathrm{Var} \left[ \hat{o} \right]
\leq 9 \|O_0 \|_2^2 + \frac{3}{d^2}\left( \| \mathrm{tr}_2 \left( \mathbb{I} \otimes \rho O \right) \|_2^2 + \| \mathrm{tr}_1 \left( \rho \otimes \mathbb{I} O \right) \|_2^2 \right) \leq 9 \|O_0 \|_2^2 + \frac{6}{d} \|O \|_\infty^2.
\end{equation}

\section{Information-theoretic lower bound with scaling in Hilbert-Schmidt norm}
\label{sec:proofthm2}

Before stating the content of the statement, we need to introduce some additional notation.
In quantum mechanics, the most general notion of a quantum measurement is a POVM (positive operator-valued measure). A $d$-dimensional POVM $F$ consists of a collection $F_1,\ldots,F_N$ of positive semidefinite matrices that sum up to the identity matrix: $\langle x|F_i|x \rangle \geq 0$ for all $x \in \mathbb{C}^d$ and $\sum_i F_i = \mathbb{I}$.
The index $i$ is associated with different potential measurement outcomes and Born's rule asserts $\mathrm{Pr} \left[ i| \rho \right]=\mathrm{tr}(F_i \rho)$ for all $1 \leq i \leq M$ and any $d$-dimensional quantum state $\rho$.
We present a simplified version of the proof by consider the relevant case where $M \leq \exp(2^n / 32)$.
The full proof can be found in \cite{huang2019predicting}.

\subsection{Detailed statement and proof idea}
\begin{theorem}[Detailed restatement of Theorem~\ref{thm:main2} for Hilbert-Schmidt norm]
\label{thm:main2restate}
Fix a sequence of POVMs $F^{(1)},\ldots,F^{(N)}$.
Suppose that given any $M$ features $0 \preceq O_1, O_2, \ldots, O_M \preceq I$ with $\max_i\Big(\norm{O_i}_2^2\Big) \leq B$, there exists a machine (with arbitrary runtime as long as it always terminates) that can use the measurement outcomes of $F^{(1)}, \ldots, F^{(N)}$ on $N$ copies of an unknown $d$-dimensional quantum state $\rho$ to $\epsilon$-accurately predict
$\Tr(O_1 \rho), \ldots, \Tr(O_M \rho)$
with high probability.
Assuming $M \leq \exp(d / 32)$, then necessarily
\begin{equation}
N \geq \Omega \Bigg(\frac{B \log(M)}{\epsilon^2} \Bigg).
\end{equation}
\end{theorem}

It is worthwhile to put this statement into context and discuss consequences, as well as limitations.
Theorem~\ref{thm:main} (Clifford measurements) equips classical shadows with a \emph{universal} convergence guarantee: (order) $\log (M) \max_i \mathrm{tr}(O_i^2)/\epsilon^2$ single-copy measurements suffice to accurately predict \emph{any} collection of $M$ target functions in \emph{any} state.
Theorem~\ref{thm:main2restate} implies that there are cases where this number of measurements
is unavoidable.
This highlights that the sample complexity of feature prediction with classical shadows is optimal in the worst case -- a feature also known as minimax optimality.

Minimax optimality, however, does not rule out potential for further improvement in certain best-case scenarios. Advantageous structure in $\rho$ or the $O_i$'s (or both) can facilitate the design of more efficient prediction techniques.
Prominent examples include matrix product state tomography (MPST) \cite{cramer2010efficient,lanyon2017efficient} and neural network tomography (NNQST) \cite{carrasquilla2019reconstructing}.
Such tailored approaches, however, hinge on additional assumptions
about the states to be measured or the properties to be predicted.\footnote{Although tractable in theory, MPST becomes prohibitively expensive if $\rho$ is not well-approximated by a MPS with small bond dimension. Likewise, NNQST seems to struggle to identify quantum states with intricate combinatorial structure, such as toric code ground states. We refer to the other supplementary sections for numerical (Supplementary Section \ref{sec:toricexp}) and theoretical (Supplementary Section \ref{sec:ml-bottleneck}) support of this claim.}

Finally, we emphasize that Theorem~\ref{thm:main2} only applies to single-copy measurements. Another way to bypass this lower bound is to use joint quantum measurements that act on all copies of the quantum state $\rho$ simultaneously.
Although very challenging to implement, such procedures can get by with substantially fewer state copies while still being universal.
Shadow tomography \cite{aaronson2018shadow,aaronson2019gentle} is a prominent example.

\paragraph*{Proof idea:} We adapt a versatile proof technique for establishing information-theoretic lower bounds on tomographic procedures that is originally due to Flammia \textit{et al.}~\cite{flammia2012quantum}; see also  \cite{haah2017sample,roth2019recovering} for adaptations and refinements. The key idea is to consider a communication task in which Alice chooses a quantum state from among an alphabet of possible states and then sends copies of her chosen state to Bob, who measures all the copies hoping to extract a classical message from Alice.  If we choose Alice's alphabet suitably, then by learning many properties of Alice's state Bob will be able to identify the state, hence decoding Alice's message. Information-theoretical lower bounds on the number of copies Bob needs to decode the message can therefore be translated into lower bounds on how many copies Bob needs to learn the properties.

To be more specific, suppose Alice chooses her state from an ensemble of $M$ possible $n$-qubit signal states $\{\rho_1, \rho_2, \dots \rho_M\}$ and suppose there are $M$ linear operators $\{O_1, O_2, \dots O_M\}$, each with $\text{tr}\left(O_i^2\right)\le B$, such that learning the expectation values of all the operators $\{O_i\}$ up to an additive error $\epsilon$ suffices to determine $\rho_i$ uniquely. Suppose furthermore that if Bob receives $N$ copies of \textit{any} $n$-qubit state, and measures them one at a time, he is able to learn all of the properties $\{O_i\}$ with an additive error no larger than $\epsilon$ with high success probability. This provides Bob with a method for identifying the state $\rho_i$ with high probability. Therefore, if Alice chooses her signal state uniformly at random from among the $M$ possible states,
 by performing the appropriate single-copy measurements Bob can acquire $\log_2 M$ bits of information about Alice's message. A lower bound on how many copies Bob needs to gain $\log_2 M$ bits of information about Alice's state, then, becomes a lower bound on how many copies Bob needs to learn the $M$ properties $\{O_i\}$. To get the best possible lower bound, we choose Alice's signal ensemble $\{\rho_i\}$ so that it is as hard as possible for Bob to distinguish the signals using properties with $\text{tr}\left(O_i^2\right)\le B$.

So far, this lower bound on $N$ would apply even if Bob has complete knowledge of Alice's signal states and the properties he should learn to distinguish them. We can derive a stronger lower bound on $N$ by invoking a powerful feature of classical shadows --- that Bob must make his measurements \textit{before} he finds out which properties he must learn. To obtain this stronger bound, we introduce into the communication scenario a third party, named Loki\footnote{In Norse mythology, Loki is infamous for mischief and trickery. However, not entirely malicious, he often shows up in the nick of time to remedy the dire consequences of his actions.}, who tampers with the signal states. Loki chooses a Haar-random $n$-qubit unitary $U$, and replaces all $N$ copies of Alice's signal state $\rho_i$ by the rotated states $U\rho_iU^\dagger$ before presenting the states to Bob (Loki's mischief).

If Bob knew Loki's unitary $U$, he could modify his measurement procedure to learn the rotated properties $\{UO_iU^\dagger\}$. These rotated properties are just as effective for distinguishing the rotated states as the unrotated properties were effective for distinguishing the unrotated states. However, Loki keeps $U$ secret, so Bob is forced to perform his measurements on the rotated states without knowing $U$. Only after Bob's data acquisition phase is completed does Loki confide in Bob and provide him with a full classical description of the unitary he applied earlier (Loki's redemption).  This three-party scenario is illustrated in Supplementary Figure~\ref{fig:commillus}.

Suppose, though, that using the classical shadow based on his measurements, Bob can predict \textit{any} $M$ properties (with additive error bounded by $\epsilon$ and with high success probability), provided that the Hilbert-Schmidt norm is no larger than $\sqrt{B}$ for each property.
Then he is just as well equipped to learn $\{UO_iU^\dagger\}$ as $\{O_i\}$, and can therefore decode Alice's message successfully once Loki reveals $U$. It must be, then, that Bob's measurement outcomes provide $\log_2 M$ bits of information about Alice's prepared state, when $U$ is known. This is the idea we use to derive the stronger upper bound on $N$, and hence prove Theorem \ref{thm:main2restate}.

We emphasize again that quantum feature prediction with classical shadows can cope with Loki's mischief, by merely rotating the features Bob predicts, because the predicted features need not be known at the time Bob measures. The lower bound in Theorem \ref{thm:main2restate} does not apply to the task of learning features that are already known in advance.
We also emphasize again that Theorem \ref{thm:main2restate} assumes that the copies of the state are measured individually. It does not apply to protocols where collective measurements are applied across many copies.

\begin{figure}[t]
    \centering
    \includegraphics[width=1.0\textwidth]{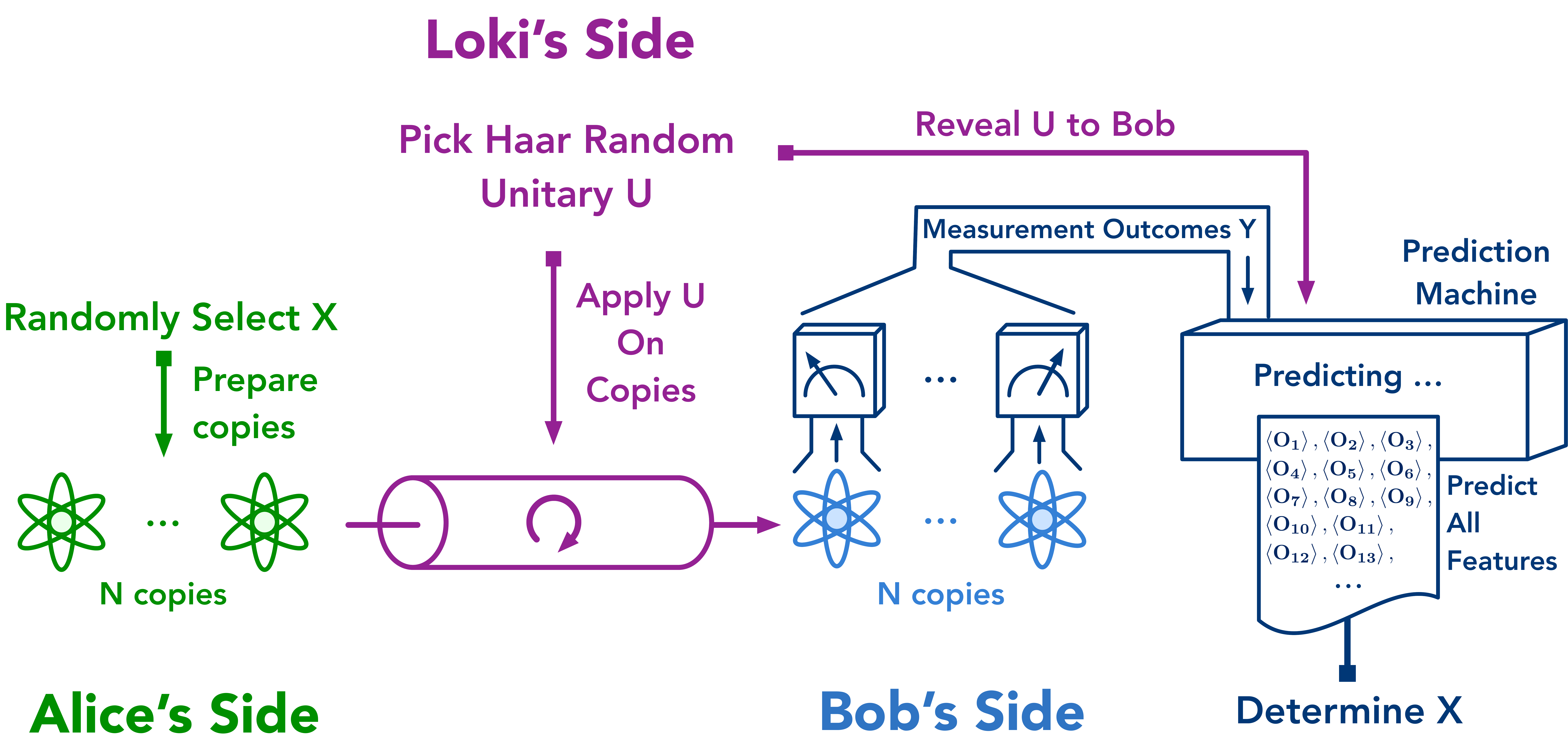}
    \caption{\emph{Illustration of the communication protocol behind Theorem~\ref{thm:main2restate} and Theorem~\ref{thm:main2restateP}.}
    Two parties (Alice and Bob) devise a protocol that allows them to communicate classical bit strings: Alice encodes a bit string $X$ in a quantum state and sends $N$ independent copies of the state to Bob. Bob performs quantum measurements and uses a black box device (e.g.\ classical shadows) to decode Alice's original message. An unpredictable trickster (Loki) tampers with this procedure by randomly rotating Alice's quantum states en route to Bob. Loki reveals his actions only after Bob has completed the measurement stage of his protocol.}

    \label{fig:commillus}
\end{figure}

\subsection{Description of the communication protocol}

We show how Alice can communicate any integer in $\{1, \ldots, M\}$ to Bob.
Alice and Bob first agree on a codebook for encoding any integer selected from $\{1, \ldots, M\}$ in a $d$-dimensional quantum state.
We denote these codebook states by $\rho_1, \ldots, \rho_M$.
Alice and Bob also agree on a set of linear features $O_1, \ldots, O_M$ that satisfies
\begin{equation}\label{eq:signal-condition}
\Tr(O_i \rho_i) \geq \max_{j \neq i} \Tr(O_j \rho_i) + 3\epsilon.
\end{equation}
Therefore, if each feature can be predicted with additive error $\epsilon$,  these features can be used to identify the state $\rho_i$.
The communication protocol between Alice and Bob is now apparent:
\begin{enumerate}
    \item Alice randomly selects an integer $X$ from $\{1, \ldots, M\}$.
    \item Alice prepares $N$ copies of the code-state $\rho_X$ associated to $X$ and sends them to Bob.
    \item Bob performs POVMs $F^{(i)}$ on individual states and receives a string of measurement outcomes $Y$.
    \item Bob inputs $Y$ into the feature prediction machine to estimate $\Tr(O_1 \rho_X), \ldots, \Tr(O_M \rho_X)$.
    \item Bob finds $\overline{X}$ that has the largest $\Tr(O_{\overline{X}} \rho_X)$.
\end{enumerate}

The working assumption is that the feature prediction machine can estimate $\Tr(O_1 \rho_X), \ldots, \Tr(O_M \rho_X)$ within $\epsilon$-error and high success probability.
This in turn ensures that this plain communication protocol is mostly successful, i.e.\
$\overline{X} = X$ with high probability.
In words:  Alice can transmit information to Bob, when no adversary is present.

We now show how they can still communicate safely in the presence of an adversary (Loki) who randomly rotates the transmitted code states en route: $\rho_X \mapsto U \rho_X U^\dagger$ and $U$ is a Haar-random unitary.

This random rotation affects the measurement outcome statistics associated with the fixed POVMs $F^{(1)},\ldots,F^{(N)}$. Each element of $Y=\left[Y^{(1)},\ldots,Y^{(N)}\right]$ is now a random variable that depends on both $X$ and $U$.
After Bob has performed the quantum measurements to obtain $Y$, the adversary confesses to Bob and reveals the random unitary~$U$.
While Bob no longer has any copies of $\rho_X$, he can still incorporate precise knowledge of $U$ by instructing the machine to predict linear features $U O_1 U^\dagger, \ldots, U O_M U^\dagger$, instead of the original $O_1, \ldots, O_M$.
This reverses the effect of the original unitary transformation, because $\Tr(U O_i U^\dagger U \rho_X U^\dagger) = \Tr(O_i \rho_X)$.
This modification renders the original communication protocol stable with respect to Loki's actions. Alice can still send any integer in $\{1, \ldots, M\}$ to Bob with high probability.

\subsection{Information-theoretic analysis}
\label{sec:infotheoanalysis}

The following arguments use properties of Shannon entropy and mutual information which can be found in standard textbooks on information theory, such as \cite{cover2012elements}.

The communication protocol is guaranteed to work with high probability, ensuring that Bob's recovered message $\bar{X}$ equals Alice's input $X$ with high probability. Moreover, we assume that Alice selects her message uniformly at random. Fano's inequality then implies
\begin{equation}
I(X : \overline{X}) = H(X) - H(X | \overline{X}) \geq \Omega(\log(M)),
\end{equation}
where $I(X : \overline{X})$ is the mutual information, and $H(X)$ is the Shannon entropy.
By assumption, Loki chooses the unitary roatation $U$ uniformly at random, regardless of the message $X$. This implies $I(X:U)=0$ and, in turn
\begin{equation}
I(X : \overline{X}) \leq I(X : \overline{X}, U) = I(X : U) + I(X : \overline{X} | U) = I(X : \overline{X} | U).
\end{equation}
For fixed $U$, $\overline{X}$ is the output of the machine that only takes into account the measurement outcomes $Y$. The data processing inequality then yields
\begin{equation}
I(X : Y | U) \geq I(X : \overline{X} | U) \geq I(X : \overline{X}) \geq \Omega(\log(M)).
\end{equation}
Recall that $Y$ is the measurement outcome of the $N$ POVMs $F_1, \ldots, F_N$. We denote the measurement outcome of $F_k$ as $Y_k$.
Because $Y_1, \ldots, Y_N$ are random variables that depend on $X$ and $U$,
\begin{align}
    I(X : Y | U) & = H(Y_1, \ldots, Y_N | U) - H(Y_1, \ldots, Y_N | X, U) \nonumber\\
    & \leq H(Y_1 | U) + \ldots + H(Y_N | U) - H(Y_1, \ldots, Y_N | X, U)\nonumber\\
    & = \sum_{k=1}^N \Big( H(Y_k | U) - H(Y_k | X, U) \Big) = \sum_{k=1}^N I(X: F_k \mbox{ on } U\rho_X U^\dagger | U).
\end{align}
The second to last equality uses the fact that when $X, U$ are fixed, $Y_1, \ldots, Y_N$ are independent.
To obtain the best lower bound, we should choose Alice's signal states $\{\rho_i\}$ such that $I(X: F_k \mbox{ on } U\rho_X U^\dagger | U)$ is as small as possible. In Sec.~\ref{subsec:detailed-decoding}, we will see that, no matter how Bob chooses his measurements $\{F_1, F_2, \dots , F_N\}$, there are signal states satisfying (\ref{eq:signal-condition}) such that
\begin{equation}
\label{eq:codestates}
I(X: F_k \mbox{ on } U\rho_X U^\dagger | U) \leq \frac{36 \epsilon^2}{B}, \forall k.
\end{equation}
Assuming that this relation holds, we have established a connection between $M$ and $N$: $\Omega(\log(M)) \leq I(X: Y | U) \leq 36 N \epsilon^2 / B$ and, therefore,
$N \geq \Omega\Big(B\log(M) / \epsilon^2\Big)$.
This establishes the claim in Theorem~\ref{thm:main2restate}.

\subsection{Detailed construction of quantum encoding and linear prediction decoding}
\label{subsec:detailed-decoding}

We now construct a codebook $\rho_1, \ldots, \rho_M$ and linear features $0 \preceq O_1, O_2, \ldots, O_M \preceq \mathbb{I}$ with $\max_i\norm{O_i}_2^2 \leq B$ that obey two key properties:
\begin{enumerate}
    \item the code states $\rho_1,\ldots,\rho_M$ obey the
    requirement displayed in Eq.~\eqref{eq:codestates}.
    \item the linear features $O_1,\ldots,O_M$ are capable of identifying a unique code state:
    \begin{equation}
    \label{eq:decoding}
    \Tr(O_i \rho_i) \geq \max_{j \neq i} \Tr(O_j \rho_i) + 3\epsilon \quad \textrm{for all} \quad 1 \leq i \leq M.
    \end{equation}
\end{enumerate}
The second condition requires each $\rho_i$ to be distinguishable from $\rho_1, \ldots, \rho_M$ via linear features $O_i$.
The first condition, on the contrary, requires $\rho_X$ to convey as little information about $X$ as possible.
The general idea would then be to create distinguishable quantum states that are, at the same time, very similar to each other.

In order to achieve these two goals,
we choose $M$ rank-$B/4$ subspace projectors $\Pi_1, \ldots, \Pi_M$ that obey $\Tr(\Pi_i \Pi_j) / r < 1 / 2$ for all $i \neq j$.
The probabilistic method asserts that such a projector configuration exists; see
Lemma~\ref{lem:probabexist} below. Now, we set
\begin{equation}
\rho_i = (1 - 3\epsilon) \frac{\mathbb{I}}{d} + 3 \epsilon \frac{4 \Pi_i}{B}, \quad \textrm{and} \quad  O_i = 2 \Pi_i, \quad \textrm{for all} \quad 1 \leq i \leq M.
\end{equation}
It is easy to check that this construction meets the requirement displayed in Eq.~\eqref{eq:decoding}.
The other condition -- Eq.~\eqref{eq:codestates} is verified in
Lemma~\ref{lem:randbound} below.

\begin{lemma}
\label{lem:probabexist}
If $M \leq \exp(rd / 32)$ and $d \geq 4r$, then $\exists M$ rank-$r$ subspace projectors $\Pi_1, \ldots, \Pi_M$ such that
\begin{equation}
\Tr(\Pi_i \Pi_j) / r < 1 / 2, \forall i \neq j.
\end{equation}
\end{lemma}
\begin{proof}
We find the subspace projectors using a probabilistic argument. We randomly choose $M$ rank-$r$ subspaces according to the unitarily invariant measure in the Hilbert space, the Grassmannian, and bound the probability that the randomly chosen subspaces do not satisfy the condition.
For a pair of fixed $i\neq j$, we have
\begin{equation}
\Pr\Bigg[\frac{1}{r} \Tr(\Pi_i \Pi_j) \geq \frac{1}{2}\Bigg] \leq \exp\Bigg(- r^2 f\Bigg(\frac{d}{2r} - 1\Bigg)\Bigg) < \exp\Bigg(- \frac{rd}{16}\Bigg),
\end{equation}
where we make use of  \cite[Lemma~6]{haah2017sample} in the first inequality and $f(z) = z - \log(1+z) > z/4$ for all $z \geq 1$ in the second inequality.
A union bound then asserts
\begin{equation}
\Pr\bigg[\exists i \neq j, \frac{1}{r} \Tr(\Pi_i \Pi_j) \geq \frac{1}{2}\bigg] < M^2 \exp\Bigg(- \frac{rd}{16} \Bigg) \leq 1.
\end{equation}
Because the probability is less than one, there must exist $\Pi_1, \ldots, \Pi_{M}$ that satisfy the desired property.
\end{proof}

\begin{lemma}
\label{lem:randbound}
Consider a set of $d$-dimensional quantum states $\{\rho_1, \ldots, \rho_M\}$ such that $\rho_i = (1-\alpha) \frac{\mathbb{I}}{d} + \alpha \frac{\Pi_i}{r}$, where $\Pi_i$ is a rank-$r$ subspace projector. Consider $U$ sampled from Haar measure, and $X$ sampled from $\{1, \ldots, M\}$ uniformly at random.  Consider any POVM measurement $F$. Then the information gain regarding $X$, conditioned on $U$, obtained from the measurement $F$ performed on the state $U\rho_X U^\dagger$ satisfies
\begin{equation}
I(X: F \mbox{ on } U\rho_X U^\dagger | U) \leq \frac{\alpha^2}{r}.
\end{equation}
\end{lemma}

\noindent Note that we can obtain the statement (\ref{eq:codestates}) by choosing $\alpha = 3 \epsilon$ and $r = B / 4$, hence completing the proof of Theorem \ref{thm:main2restate}.
\begin{proof}
First of all, let us decompose all POVM elements $\{F_1, \ldots, F_{l}\}$ to rank-$1$ elements $F' = \big\{w_i d \ket{v_i} \bra{v_i}\big\}_{i=1}^{l'}$, where $l \leq l'$. We can perform measurement $F$ by performing measurement with $F'$: when we measure a rank-$1$ element, we return the original POVM element the rank-$1$ element belongs to.
Using data processing inequality, we have $I(X: F \mbox{ on } U\rho_X U^\dagger | U) \leq I(X: \tilde{F} \mbox{ on } U\rho_X U^\dagger | U)$.
From now on, we can consider the POVM $\vec{F}$ to be $\big\{w_i d \ket{v_i} \bra{v_i}\big\}_{i=1}^l$.
Normalization demands
\begin{equation}
\Tr\Big(\sum_i w_i d \ket{v_i} \bra{v_i}\Big) = \Tr(\mathbb{I}) = d \quad \text{and therefore} \quad \sum_i w_i = 1.
\end{equation}
Let us define the probability vector $\vec{p} = \Tr(U\rho_1 U^\dagger \vec{F}),$ so $p_i = w_i d \bra{v_i} U\rho_1 U^\dagger \ket{v_i}.$
And the expression we hope to bound satisfies $I(X: F \mbox{ on } U\rho_X U^\dagger | U) = I(X, U : F \mbox{ on } U\rho_X U^\dagger) - I(U : F \mbox{ on } U\rho_X U^\dagger) \leq I(X, U : F \mbox{ on } U\rho_X U^\dagger)$ using the chain rule and the nonnegativity of mutual information.
We now bound
\begin{align}
    I(X, U : F \mbox{ on } U\rho_X U^\dagger) = & H\Big( \sum_{X=1}^M \frac{1}{M} \E_U [ \Tr(U \rho_X U^\dagger \vec{F}) ] \Big) - \sum_{X=1}^M \frac{1}{M} \E_U \Big[ H\Big(\Tr(U \rho_X U^\dagger \vec{F}) \Big)\Big] \nonumber\\
	= & H\Big( \Tr(\E_U [U \rho_1 U^\dagger] \vec{F}) \Big) - \E_U \Big[H\Big(\Tr(U \rho_1 U^\dagger \vec{F}) \Big)\Big] \nonumber\\
	= & \sum_i - (\E_U p_i) \log(\E_U p_i) + \E_U [p_i \log p_i] \nonumber\\
	\leq & \sum_i - (\E_U p_i) \log(\E_U p_i) + \E_U\Big[p_i \log(\E_U p_i) + p_i \frac{p_i - \E_U p_i}{\E_U p_i}\Big] \nonumber \\
	= & \sum_i \frac{\E_U[p_i^2] - \E_U[p_i]^2}{\E_U[p_i]}.
\end{align}
The second equality uses the fact that $\E_U f(U \rho_X U^\dagger) =\ E_U f(U \rho_1 U^\dagger), \forall X$ which follows from the fact that $\forall X, \exists U_X, \rho_X = U_X \rho_1 U_X^\dagger$.
The inequality uses the fact that $\log(x)$ is concave, so $\log(x) \leq \log(y) + \frac{x-y}{y}$.
Using properties of Haar random unitary $d \times d$ matrices, we conclude
\begin{equation}
\E_U[p_i] = w_i, \,\, \E_U[p_i^2] = w_i^2 \frac{d}{(d+1)} \Bigg(1 + \frac{1}{d} + \alpha^2\Big(\frac{1}{r} - \frac{1}{d}\Big)\Bigg).
\end{equation}
Therefore we have
\begin{equation}
\frac{\E_U[p_i^2] - \E_U[p_i]^2}{\E_U[p_i]} = w_i \alpha^2 \frac{d}{d+1} \Big(\frac{1}{r} - \frac{1}{d}\Big) \leq \frac{w_i \alpha^2}{r},
\end{equation}
which establishes the claim:
\begin{equation}
I(X: F \mbox{ on } U\rho_X U^\dagger | U) \leq \sum_i \frac{\E_U[p_i^2] - \E_U[p_i]^2}{\E_U[p_i]} \leq \frac{\alpha^2}{r}.
\end{equation}
\end{proof}

\section{Information-theoretic bounds on predicting local observables}
\label{sec:proofthm2P}

In Theorem~\ref{thm:main2restate}, we have shown that if a procedure can predict arbitrary observables with $\Tr(O_i^2) \leq B$, then it must use at least $\Omega(B \log(M) / \epsilon^2)$ single-copy measurements (as long as $M$ is not extraordinarily large).
A similar argument
can be used to show that if a procedure can predict arbitrary $k$-local observables, then it requires at least $\Omega(2^k \log(M) / \epsilon^2)$ single-copy measurements (when $M$ is not too large).
This is because if we focus on a $k$-qubit subsystem, then the guarantee allows us to predict arbitrary observables $0 \preceq O_i \preceq \mathbb{I}$ with $\Tr(O_i^2) \leq 2^k$.
In the following theorem, we show a stronger lower bound by focusing on local measurements.
A local measurement is a POVM $\{w_i d \ket{v_i}\!\bra{v_i}\}_i$ where $\ket{v_i} = \ket{v_i^{(1)}} \otimes \ldots \otimes \ket{v_i^{(n)}}$, $\sum_i w_i = 1$, and $d = 2^n$.
This is the same as not performing any entangling gates when implementing the measurement. (Random) Pauli basis measurements are a prominent example.

\begin{theorem}[Detailed restatement of Theorem~\ref{thm:main2} for exponential scaling in locality]
\label{thm:main2restateP}
Fix a sequence of local measurements $F_1,\ldots,F_N$ on $n$-qubit system, i.e., $F_j = \{w_{j, i} d \ket{v_{j, i}}\!\bra{v_{j, i}}\}_i$ where $\ket{v_{j, i}} = \ket{v_{j, i}^{(1)}} \otimes \ldots \otimes \ket{v_{j, i}^{(n)}}$, $\sum_i w_{j, i} = 1$, and $d = 2^n$.
Suppose that given any $M$ $k$-local observables $-\mathbb{I} \preceq O_1, O_2, \ldots, O_M \preceq \mathbb{I}$, there exists a machine (with arbitrary runtime as long as it always terminates) that can use the measurement outcomes of $F_1, \ldots, F_N$ on $N$ copies of an unknown quantum state $\rho$ to $\epsilon$-accurately predict
$\Tr(O_1 \rho), \ldots, \Tr(O_M \rho)$
with high probability.
Assuming $M \leq 3^k {n \choose k}$, then necessarily
\begin{equation}
N \geq \Omega \Bigg(\frac{3^k \log(M)}{\epsilon^2}\Bigg).
\end{equation}
\end{theorem}
\begin{proof}
The proof uses a quantum communication protocol between Alice and Bob, with Loki interfering in the middle.
Alice would encode some classical information in the quantum state and send to Bob.
Bob would then use the prediction procedure to decode the encoded classical information.
In the middle, Loki will alter the quantum state by applying a random unitary.
Loki would then reveal the random unitary to Bob after Bob performed quantum measurements on the quantum states.
An illustration of the communication protocol can be found in Supplementary Figure~\ref{fig:commillus}.
The quantum state Alice encodes, the unitary applied by Loki, and the features predicted by Bob is considerably simplified in this result compared to the previous proof.

We define $\rho_i = (\mathbb{I} + 3 \epsilon P_i) / 2^n, \forall i = 1, \ldots, M$.
$P_i$ is the $i$-th Pauli observable acting on $k$ qubits in the $n$-qubit system.
Any ordering of the Pauli observables is fine.
Note that there are at most $3^k {n \choose k}$ such Pauli observables.
This is the reason why we assume $M \leq 3^k {n \choose k}$.
The corresponding linear functions chosen by Bob are $O_i = P_i, \forall i = 1, \ldots, M$.
This guarantees the following relation:
\begin{equation}
\Tr(O_i \rho_j) = 3 \epsilon \delta_{ij}
\quad \text{for all $1 \leq i,j \leq M$,}
\end{equation}
where $\delta_{ij}$ is the Kronecker-delta ($\delta{ij}=1$ if $i=j$ and $\delta_{ij}=0$ otherwise).
The random unitary applied by Loki consists of random single-qubit unitary rotations, i.e.\ $U=U^{(1)} \otimes \ldots \otimes U^{(n)}$.
The complete communication protocol works as follows.
\begin{enumerate}
    \item Alice randomly selects an integer $X$ from $\{1, \ldots, M\}$.
    \item Alice prepares $N$ copies of the code-state $\rho_X$ according associated to $X$ and sends them to Bob.
    \item Loki intercepts the $N$ copies, samples a random unitary $U = U^{(1)} \otimes \ldots \otimes U^{(n)}$, applies $U$ on all copies of $\rho_X \rightarrow U \rho_X U^\dagger$, and sends to Bob.
    \item Bob performs local measurements $F_j$ on individual states and receives a string of measurement outcomes $Y$.
    \item Loki reveals the random unitary $U$ to Bob. Now Bob would have to predict the expectation value of $U O_1 U^\dagger, \ldots, U O_{M} U^\dagger$ instead of the original $O_1, \ldots, O_{M}$.
    \item Since $U O_1 U^\dagger, \ldots, U O_{M} U^\dagger$ are still $k$-local observables, Bob can input $Y$ into the feature prediction machine to estimate $\braket{U O_i U^\dagger}_{U \rho_X U^\dagger} = \Tr( O_i \rho_X ), \forall i = 1, \ldots, M$.
    \item Bob finds $\overline{X} \in \{1, \ldots, M\}$ that has the largest $\Tr( O_{\overline{X}} \rho_X)$.
\end{enumerate}

Because $\Tr( O_i \rho_X )$ are predicted to $\epsilon$ additive error, and $\Tr( O_i \rho_X ) = 3 \epsilon \delta_{i X}$, if the prediction procedure works as guaranteed, Bob's decoded information $\hat{X}$ would be equal to Alice's encoded information $X$ with high probability. Moreover, we assume that Alice selects her message uniformly at random. Fano's inequality then implies
\begin{equation}
I(X : \overline{X}) = H(X) - H(X | \overline{X}) \geq \Omega(\log(M)),
\end{equation}
where $I(X : \overline{X})$ is the mutual information, and $H(X)$ is the Shannon entropy.
By assumption, Loki chooses the random unitary $U$ regardless of the message $X$.
This implies $I(X:U)=0$ and, in turn
\begin{equation}
I(X : \overline{X}) \leq I(X : \overline{X}, U) = I(X : U) + I(X : \overline{X} | U) = I(X : \overline{X} | U).
\end{equation}
For fixed $U$, $\overline{X}$ is the output of the machine that only takes into account the measurement outcomes $Y$. The data processing inequality then implies
\begin{equation}
    \label{eq:lowerboundI}
    I(X : Y | U) \geq I(X : \overline{X} | U) \geq I(X : \overline{X}) \geq \Omega(\log(M)).
\end{equation}
Recall that $Y$ is the measurement outcome of the $N$ POVMs $F_1, \ldots, F_N$. We denote the measurement outcome of $F_j$ as $Y_j$.
Because $Y_1, \ldots, Y_N$ are random variables that depend on $X$ and $U$,
\begin{align}
    I(X : Y | U) & = H(Y_1, \ldots, Y_N | U) - H(Y_1, \ldots, Y_N | X, U) \nonumber\\
    & \leq H(Y_1 | U) + \ldots + H(Y_N | U) - H(Y_1, \ldots, Y_N | X, U) \nonumber\\
    & = \sum_{j=1}^N \Big( H(Y_j | U) - H(Y_j | X, U) \Big) = \sum_{j=1}^N I(X: F_j \mbox{ on } U\rho_X U^\dagger | U). \label{eq:upperboundI}
\end{align}
The second to last equality uses the fact that when $X, U$ are fixed, $Y_1, \ldots, Y_N$ are independent.
This part of the derivation is exactly the same as in Supplementary Section~\ref{sec:infotheoanalysis}.
All that is left is to properly upper bound $I(X: F_j \mbox{ on } U\rho_X U^\dagger | U)$.
First, by definition,
\begin{align}
    I(X: F_j \mbox{ on } U\rho_X U^\dagger | U) & = \E_U \left[ H(F_j \mbox{ on } U\rho_X U^\dagger) - H(X, F_j \mbox{ on } U\rho_X U^\dagger ) \right] \nonumber \\
    & = \E_U \left[ H\left( \E_X \Tr(U \rho_X U^\dagger \vec{F}_j) \right) - \E_X H\left( \Tr(U \rho_X U^\dagger \vec{F}_j) \right) \right] \nonumber \\
    & \leq H\left( \E_X \E_U \Tr(U \rho_X U^\dagger \vec{F}_j)  \right) - \E_X \E_U H\left( \Tr(U \rho_X U^\dagger \vec{F}_j) \right). \label{ineq:infotheory}
\end{align}
The last inequality exploits concavity of the Shannon entropy $H(\cdot)$.
By assumption, the $F_j$'s must be local measurements, i.e. $F_j = \{w_{j, i} d \ket{v_{k, i}}\!\bra{v_{k, i}}\}_i$ where $\ket{v_{k, i}} = \ket{v_{k, i}^{(1)}} \otimes \ldots \otimes \ket{v_{k, i}^{(n)}}$, $\sum_i w_i = 1$, and $d = 2^n$.
We define the probability of measuring $i$-th outcome using POVM $F_j$ as
\begin{equation}
\label{eq:probuni}
p_{j, i} = w_{j, i} d \bra{v_{j, i}} U \rho_X U^\dagger \ket{v_{j, i}},
\end{equation}
which is a random number depending on $X$ and $U$.
Using Equation~\eqref{ineq:infotheory} and the definition of $H(\cdot)$, we have
\begin{align}
    I(X: F_j \mbox{ on } U\rho_X U^\dagger | U) & \leq H\left( \E_X \E_U \Tr(U \rho_X U^\dagger \vec{F}^{(k)})  \right) - \E_X \E_U H\left( \Tr(U \rho_X U^\dagger \vec{F}^{(k)}) \right) \nonumber \\
    & = \sum_i \big(\E_{X, U} [ p_{j, i} \log(p_{j, i}) ] - \E_{X, U} [ p_{j, i}] \log( \E_{X, U} [ p_{j, i}] ) \big) \nonumber\\
    & \leq \sum_i - (\E_{X, U} p_{j, i}) \log(\E_{X, U} p_{j, i}) + \E_{X, U} \Big[p_{j, i} \log(\E_{X, U} p_{j, i}) + p_{j, i} \frac{p_{j, i} - \E_{X, U} p_{j, i}}{\E_{X, U} p_{j, i}}\Big] \nonumber \\
	& = \sum_i \frac{\E_{X, U}[p_{j, i}^2] - \E_{X, U}[p_{j, i}]^2}{\E_{X, U}[p_{j, i}]}. \label{eq:finallocal}
\end{align}
The second inequality uses the fact that $\log(x)$ is concave, so $\log(x) \leq \log(y) + \frac{x-y}{y}$.
We now compute $\E_{X, U}[p_{j, i}]$ and $\E_{X, U}[p_{j, i}^2]$ by using the following relation for single-qubit random unitary:
\begin{equation}
\E_{U^{(j)}} \left[ U^{(j)} \ket{v_{k, i}^{(j)}}\!\bra{v_{k, i}^{(j)}} (U^{(j)})^\dagger \right] = \frac{\mathbb{I}^{(j)}}{2}, \quad \E_{U^{(j)}} \left[ \left(U^{(j)} \ket{v_{k, i}^{(j)}}\!\bra{v_{k, i}^{(j)}} (U^{(j)})^\dagger \right)^{\otimes 2} \right] = \frac{\mathbb{I}^{(j)} \otimes \mathbb{I}^{(j)} + S^{(j)}}{3},
\end{equation}
where $j$ refers to the $j$-th qubit, and $S$ is the two qubit swap operator ($|\psi \rangle \otimes | \phi \rangle = | \phi \rangle \otimes | \psi \rangle$).
Recall the definition of $p_{j, i}$ in Equation~\eqref{eq:probuni}.
Together with the above relation, we have
\begin{align}
\E_{X, U}[p_{j, i}] =& \E_X \left[ w_{j, i} d \Tr\left(\rho_X \frac{\mathbb{I}}{2^n}\right) \right] = \E_X \left[ w_{j, i} 2^n \Tr\left(\frac{\mathbb{I} + 3 \epsilon P_X}{2^n} \frac{\mathbb{I}}{2^n}\right) \right] =  w_{j, i} \quad \text{and} \nonumber \\
\E_{X, U}[p_{j, i}^2] =& \E_X \left[ w_{j, i}^2 d^2 \Tr\left(\rho_X^{\otimes 2} \bigotimes_{j=1}^n \left(\frac{\mathbb{I}^{(j)}\otimes \mathbb{I}^{(j)} + S^{(j)}}{3}\right) \right) \right] = w_{j, i}^2 \left(1 + \frac{9 \epsilon^2}{3^k}\right).
\end{align}
Putting this computation into Inequality~\eqref{eq:finallocal}, we have obtained
\begin{equation}
I(X: F_j \mbox{ on } U\rho_X U^\dagger | U) \leq \sum_i w_{j, i} \frac{9 \epsilon^2}{3^k} = \frac{9 \epsilon^2}{3^k}.
\end{equation}
Combining the above result with Inequality~\eqref{eq:lowerboundI} and \eqref{eq:upperboundI}, we have
\begin{equation}
\frac{9 N \epsilon^2}{3^k} \geq I(X : Y | U) \geq \Omega(\log(M)) \quad \text{which implies} \quad  N \geq \Omega\left(\frac{3^k \log(M)}{\epsilon^2}\right).
\end{equation}
\end{proof}

\bibliographystyle{abbrv}

\end{document}